\documentclass[a4paper,UKenglish,cleveref, autoref, thm-restate]{lipics-v2021}

\usepackage{amssymb}
\usepackage{xcolor} 
\usepackage{stmaryrd}
\usepackage{balance}
\usepackage{nicefrac}
\usepackage{tikz}
\usepackage[Dfprime,bbsets]{math} 
\usepackage{tikz-cd}

\definecolor{semblue}{rgb}{0,0,0.7}
\definecolor{vgreen}{rgb}{.1,.5,0}
\definecolor{vdarkgreen}{rgb}{.06,.3,0}
\definecolor{vred}{rgb}{.7,0,0}
\definecolor{vblue}{rgb}{.1,.15,.62}
\definecolor{vgray}{rgb}{.35,.35,.35}
\definecolor{darkishgray}{rgb}{.35,.35,.35}
\definecolor{vvblue}{rgb}{.14,.21,.868}

\definecolor{blueish}{rgb}{0,.2,.6}

\newcommand{\fmlset}{\Gamma}

\newcommand{\fmlsetb}{\Delta}
\newcommand{\fmlsetc}{E}

\newcommand{\fml}{A}
\newcommand{\fmlb}{B}
\newcommand{\fmlc}{C}
\newcommand{\fmld}{D}
\newcommand{\fmle}{F}

\newcommand{\ftrue}{\top}
\newcommand{\ffalse}{\bot}
\newcommand{\truefalseaxiomset}{{\top}\kern-7.7pt{\bot}}

\newcommand{\tops}{X}

\newcommand{\patoms}{\mathsf{P}}

\newcommand{\topzero}{0}
\newcommand{\topone}{1}
\newcommand{\topimply}[2]{#1\limply#2}

\newcommand{\topand}{\sqcap}
\newcommand{\topor}{\sqcup}
\newcommand{\tople}{\sqsubset}

\newcommand{\settopowfunc}{\mathcal{P}^2}
\newcommand{\settopowob}[1][\setdoma]{\settopowfunc#1}
\newcommand{\settopowmorph}[1]{\settopowfunc(#1)}

\newcommand{\opcategory}[1]{#1^\mathrm{op}}

\newcommand{\Intersection}{\bigcap}
\newcommand{\intersection}{\cap}

\newcommand{\union}{\cup}
\newcommand{\setcomp}[1]{#1^\mathrm{c}}

\newcommand{\shs}{H}
\newcommand{\sbs}{B}

\newcommand{\shleq}{\sqsubset}
\newcommand{\shand}{\sqcap}
\newcommand{\shor}{\sqcup}
\newcommand{\shimply}{\rightarrow}

\newcommand{\saone}{1}
\newcommand{\sazero}{0}
\newcommand{\saa}{a}
\newcommand{\sab}{b}
\newcommand{\sac}{c}
\newcommand{\sad}{d}

\newcommand{\patom}{P}
\newcommand{\patomb}{Q}
\newcommand{\heyint}{I}
\newcommand{\heyintof}[2][\heyint]{#1(#2)}

\newcommand{\colenvelope}[2][black]{{\llbracket}{#2}{\rrbracket}}
\newcommand{\colenvelopeb}[2][black]{{\llparenthesis}{#2}{\rrparenthesis}}

\newcommand{\sasem}[3][\shs]{#2\colenvelope[olive]{#3}^{#1}}
\newcommand{\satrue}[3][\shs]{#1,#2\mathrel{\vDash}{#3}}
\newcommand{\saconseq}[2]{#1\mathrel{\vDash_{\mathrm{sH}}}{#2}}
\newcommand{\saconseqbool}[2]{#1\mathrel{\vDash_{\mathrm{sBA}}}{#2}}

\newcommand{\settrue}[3][\setdoma]{#1,#2\mathrel{\color{gray}\vDash}{#3}}
\newcommand{\setconseq}[2]{#1\mathrel{\color{gray}\vDash_{\mathrm{set}}}{#2}}

\newcommand{\notsetconseq}[2]{#1\mathrel{\color{gray}\not\vDash_{\mathrm{set}}}{#2}}
\newcommand{\setsem}[3][]{#2\colenvelope[gray]{#3}^{#1}}
\newcommand{\setsemcl}[3][]{{#2}\colenvelopeb[darkishgray]{#3}^{#1}}

\newcommand{\heyintderived}[1]{\heyint}

\newcommand{\bigshand}{\bigsqcap}

\DeclareMathOperator{\uparrowop}{\uparrow}
\DeclareMathOperator{\downarrowop}{\downarrow}

\newcommand{\shupset}[1]{#1\uparrowop}
\newcommand{\shdownset}[1]{#1\downarrowop}
\newcommand{\shiso}{\simeq}

\newcommand{\shyol}{\preccurlyeq}

\newcommand{\yonedaassociative}{Yoneda associative\xspace}
\newcommand{\yonedaacommutative}{Yoneda commutative\xspace}
\newcommand{\yonedaassociativity}{Yoneda associativity\xspace}
\newcommand{\yonedainjective}{Yoneda injective\xspace}
\newcommand{\yonedainjectivity}{Yoneda injectivity\xspace}

\newcommand{\semicat}{\mathsf{C}}
\newcommand{\semicatb}{\mathsf{D}}
\newcommand{\semifunc}{F}
\newcommand{\semifuncb}{G}

\newcommand{\adjbi}[1][\semiobj,\semiobjb]{\theta_{#1}}

\newcommand{\longshortmid}{\;\raisebox{0.05ex}{\vrule height 1.15ex width 0.6pt}\;}
\newcommand{\dashlen}{0.12em}
\newcommand{\dashsep}{0.08em}
\newcommand{\dashwidth}{0.6pt}
\newcommand{\dashhrule}{\raisebox{0.55ex}{\rule[0pt]{\dashlen}{\dashwidth}\hspace{\dashsep}\rule[0pt]{\dashlen}{\dashwidth}\hspace{\dashsep}\rule[0pt]{\dashlen}{\dashwidth}}}

\newcommand{\sadjs}{\mathrel{{\dashhrule\kern-2pt\longshortmid}}}

\newcommand{\strongadj}[2]{#1\dashv#2}

\newcommand{\semiobj}{a}
\newcommand{\semiobjb}{b}
\newcommand{\semiobjc}{c}

\renewcommand{\complement}[1]{#1^\mathrm{c}}

\newcommand{\derivedcat}[1][\semicat]{\widehat{#1}}
\newcommand{\derivedmor}[1][\catmorph]{\widehat{#1}}

\newcommand{\Homs}[3]{\mathrm{Hom}_{#1}(#2,#3)}
\newcommand{\Covhom}[2]{\Homs{#1}{#2}{-}}
\newcommand{\Conthom}[2]{\Homs{#1}{-}{#2}}

\newcommand{\naturalto}{\Rightarrow}

\newcommand{\yonedacatname}{derived category\xspace}

\newcommand{\catobj}{A}
\newcommand{\catobjb}{B}
\newcommand{\catmorph}{f}

\newcommand{\dercatmorph}{\eta}
\newcommand{\dercatmorphb}{\delta}

\newcommand{\dercatmorphpre}[1][\dercatmorph]{#1^1}
\newcommand{\dercatmorphpost}[1][\dercatmorph]{#1^2}

\newcommand{\derivedprecomp}[1]{\Homs{\semicat}{#1}{-}}
\newcommand{\derivedpostcomp}[1]{\Homs{\semicat}{-}{#1}}

\usepackage{paralist}
\newenvironment*{caselist}[1][- Case]{%
    \newcommand{\case}[1]{\item[\noindent\emph{#1 ##1}:]}
    \begin{inparaitem}
        }{%
    \end{inparaitem}
}

\newcommand{\catsba}{\mathsf{sBA}}
\newcommand{\catset}{\mathsf{Set}}

\newcommand{\SubreflLogic}{Subreflexive Logic\xspace}
\newcommand{\Subrefllogic}{Subreflexive logic\xspace}
\newcommand{\subrefllogic}{subreflexive logic\xspace}
\newcommand{\classicalsubrefl}{classical subreflexive logic\xspace}
\newcommand{\Classicalsubrefl}{Classical subreflexive logic\xspace}

\newcommand{\yonedaorder}{Yoneda preorder\xspace}
\newcommand{\yonedaisomorphic}{Yoneda equivalent\xspace}
\newcommand{\yonedaisomorphism}{Yoneda equivalence\xspace}

\newcommand{\subreflprov}[2]{#1\vdash#2}
\newcommand{\notsubreflprov}[2]{#1\not\vdash#2}

\newcommand{\classsubreflprov}[2]{#1\vdash_{\mathrm{c}}#2}
\newcommand{\notclasssubreflprov}[2]{#1\not\vdash_{\mathrm{c}}#2}

\newcommand{\relfmlset}{\Sigma}
\newcommand{\relfmlsetcl}{\overline{\relfmlset}}

\newcommand{\fmlcless}{\shleq_\relfmlset^\mathrm{c}}
\newcommand{\fmlless}{\shleq_\relfmlset}

\newcommand{\allfmls}{\mathcal{F}}
\newcommand{\allfmlssha}[1][\relfmlset]{\mathcal{F}_{#1}}
\newcommand{\allfmlssba}[1][\relfmlset]{\mathcal{F}^\mathrm{c}_{#1}}

\newcommand{\idset}{\mathsf{ID}}

\newcommand{\Saturatedness}{Closedness\xspace}
\newcommand{\saturatedness}{closedness\xspace}
\newcommand{\saturated}{closed\xspace}
\newcommand{\classicallysaturated}{classically-closed\xspace}
\newcommand{\constructive}{constructive\xspace}

\newcommand{\Constructivity}{Constructivity\xspace}

\newcommand{\subrefllogiccat}[1][\relfmlset]{\mathsf{Syn}_#1}

\newcommand{\LJ}{\ensuremath{\mathsf{LJ}}\xspace}
\newcommand{\LK}{\ensuremath{\mathsf{LK}}\xspace}

\newcommand{\derivation}{\mathcal{D}}
\newcommand{\derivationb}{\mathcal{E}}

\newcommand{\classical}{classical\xspace}
\newcommand{\classicality}{classicality\xspace}

\newcommand{\Classical}{Classical\xspace}

\newcommand{\minpwidth}{0.45\columnwidth}

\newcommand{\rex}{x}
\newcommand{\rey}{y}
\newcommand{\rez}{z}

\newcommand{\threealg}{\mathcal{H}_3}
\newcommand{\intha}{\mathcal{H}_\reals}
\newcommand{\midpoint}{\nicefrac{1}{2}}

\newcommand{\sel}{a}

\newcommand{\sela}{a_1}
\newcommand{\selb}{a_2}
\newcommand{\selop}[1][]{U_{#1}}
\newcommand{\selcl}[1][]{C_{#1}}

\newcommand{\selopa}{\selop[1]}
\newcommand{\selopb}{\selop[2]}
\newcommand{\selcla}{\selcl[1]}
\newcommand{\selclb}{\selcl[2]}

\newcommand{\filter}{\mathcal{F}}

\newcommand{\filterb}{\mathcal{G}}
\newcommand{\filtermax}{\mathcal{F}_\mathrm{max}}

\newcommand{\filterp}[1][\saa]{\mathcal{F}_{#1}}

\newcommand{\lathomalg}[1][\shs]{\tilde{#1}}
\newcommand{\lathom}{f}

\newcommand{\kleenethree}{\mathcal{K}_3}

\newcommand{\setdoma}{X}

\newcommand{\setinterp}{I}

\newcommand{\setinteralmpof}[2][\setinterp]{\setinterpalm[#1](#2)}
\newcommand{\setinterprob}[1][\setinterp]{#1}
\newcommand{\setinterpalm}[1][\setinterp]{\overline{#1}}

\newcommand{\sasemsimple}[3][]{\colenvelope[olive]{#3}}

\newcommand{\allsequents}[2]{\lsequent{#1}{\vec{#2}}}

\newcommand{\semiHeyting}{Heyting\xspace}
\newcommand{\semiHeytingalgebra}{Heyting semialgebra\xspace}

\newcommand{\SemiHeyting}{Heyting\xspace}
\newcommand{\SemiHeytingalgebra}{Heyting semialgebra\xspace}
\newcommand{\SemiHeytingAlgebra}{Heyting Semialgebra\xspace}

\newcommand{\semiBoolean}{Boolean\xspace}
\newcommand{\semiBooleanalgebra}{Boolean semialgebra\xspace}
\newcommand{\semiBooleanalgebraemph}{Boolean \emph{semialgebra}\xspace}

\newcommand{\SemiBooleanalgebra}{Boolean semialgebra\xspace}
\newcommand{\SemiBooleanAlgebra}{Boolean Semialgebra\xspace}
\newcommand{\semialgebra}{semialgebra\xspace}

\usepackage{cleveref}%
\crefname{enumi}{}{}%
\Crefname{enumi}{}{}%

\AtEndPreamble{%
    \usepackage{enumitem}
    \newcommand{\itrefof}[2]{\Cref{#1}(\ref{#2})}
    \usepackage[prefixflatinterpret,bracketmodalinterpret,fixformat,silentconst,sidenotecalculus,longseqcontext,seqArrow,boldmquant]{logic}%
    \usepackage[prefixflatinterpret,bracketmodalinterpret,fixformat,silentconst,differentialdL,simplenames]{dL}%

    \definecolor{darkishgray}{rgb}{.35,.35,.35}
    \renewcommand*{\irrulename}[1]{\text{\textcolor{darkishgray}{\upshape#1}}}
    \renewcommand{\linferenceRuleLabelPrefix}{}%\
    \renewcommand{\linferenceRuleNameSeparation}{\hspace{1em}}%
    \renewcommand{\linferPremissSeparation}{\hspace{2em}}
    \usepackage[createShortEnv]{proof-at-the-end}
}
 

\hideLIPIcs
\title{Subreflexive Logic: Completeness without Identity}

\author{Noah {Abou El Wafa}}{Karlsruhe Institute of Technology, Germany}{noah.abouelwafa@kit.edu}{https://orcid.org/0000-0002-3987-9919}{}
\author{André {Platzer}}{Karlsruhe Institute of Technology, Germany}{platzer@kit.edu}{https://orcid.org/0000-0001-7238-5710}{}
\authorrunning{N. Abou El Wafa, A. Platzer}

\Copyright{Noah {Abou El Wafa} and André {Platzer}} 


\nolinenumbers

\begin{document}

\maketitle

\begin{abstract}
    This paper shows that the substructural logic without the identity principle \(\fml\limply\fml\) (i.e., \emph{subreflexive logic}) has principled sound and complete semantics and supports a variety of applications.
    This decidable \emph{generalization} of propositional logic naturally interprets implication as \emph{robust consequence}.
    Subreflexive logic is proved to admit syntactic cut elimination.

    \SemiHeyting{} and \semiBooleanalgebraemph{}s are introduced as generalizations of Heyting and Boolean algebras and are shown to provide complete algebraic semantics without inadvertently reintroducing reflexivity.
    \emph{Semi}-adjunctions on \emph{semi}-categories and (identity-free) (co-)units are defined to give complete \emph{semi}-categorical semantics.
    In the classical case, denotational set semantics that interpret implication as \emph{robust material implication} are proved complete for subreflexive logic.
\end{abstract}

\section{Introduction}

Substructural logics such as linear logic \cite{DBLP:journals/tcs/Girard87} and relevance logic \cite{DBLP:journals/jsyml/Belnap60} have enabled many important applications \cite{DBLP:conf/concur/CairesP10,DBLP:journals/tcs/Abramsky93,Tedder2022-TEDRPD}, by dropping seemingly innocuous but highly restrictive structural assumptions.
This paper introduces \emph{\subrefllogic} as the substructural generalization of ordinary propositional logic \emph{without the identity principle} \(\fml\limply\fml\).
This principle is a fundamental part of most logical calculi for two reasons:
Firstly, it is thought of as the tautological empty statement: if \(\fml\) is true then \(\fml\) is true.
Secondly, despite its presumed emptiness, the identity principle plays a crucial role in the development of most logical theories, their semantics and categorical interpretations, where it is a load-bearing assumption.
\Subrefllogic removes this critical assumption by \emph{completely} developing logic \emph{substructurally without the identity principle} and, instead, makes it possible to freely choose axioms without imposing identity.
This paper shows that \subrefllogic is a well-behaved, structurally sound logic, as witnessed by its clean proof theory and its threefold semantics: nonreflexive algebraic, denotational set-theoretic and semicategorical semantics.
\Subrefllogic reveals a close connection between the identity principle, totality of functions and three-valued logic.
These connections open up many applications for \subrefllogic where implication is interpreted as \emph{robust consequence}.

Contrary to the usual interpretation of the identity principle as meaningless, it does significantly restrict expressivity.
There are two classes of applications where the identity principle is too restrictive.
The first class of examples relates to partiality, which is indispensable when reasoning, for example, about program termination, since a program may accept, reject or not terminate \cite{scottoutline}.
Similarly, logical provability requires partiality since a formula may be provable, disprovable, or neither.
These problems are best described by Kleene's three-valued logic, in which the principle of identity does not hold \cite{DBLP:journals/jsyml/Kleene38a}.
\Subrefllogic subsumes this logic and a completeness theorem demonstrates that removing the identity principle corresponds to the important step of generalizing from total to partial functions:
\[\text{Principle of Identity}\quad\leftrightsquigarrow\quad\text{Totality}\]
Thus, \subrefllogic enables reasoning about partial functions by familiar logic simply by dropping the axiom of identity.
This phenomenon has motivated recent interest in non-reflexive logics of various kinds to handle logical paradoxes and truth predicates \cite{DBLP:journals/jphil/BarrioRT15,DBLP:journals/sLogica/NicolaiR23,DBLP:journals/sLogica/ReSC24,Schroeder-Heister2016}.
This paper, via a structural investigation of identity, provides a set-theoretic completeness result and algebraic tools that are the mathematical foundations for such instances.

The second class of applications where identity is too restrictive relates to robustness.
Robust implication requires fine control over the axioms of implication, where for example \(x\geq 0.1 \limply x>0\) should be an axiom, but \(x=0\limply x=0\) should not be an axiom, since it is not preserved by arbitrarily small local perturbations of \(x\).
This is in conflict with the principle of identity, but covered by \subrefllogic.
The intuition behind robust implication is that unbounded precision can be \emph{assumed} in the antecedent (\(x \geq 0.1\) is true for infinitely many decimal places) but only bounded precision can be \emph{asserted} in the succedent (\(x>0\) is determined after finitely many decimal places).
This interpretation of implication is practically relevant in floating-point arithmetic and real number computability \cite{DBLP:series/txtcs/Weihrauch00}.
The restriction is a key ingredient to turn an incomplete calculus into a complete one by removing incompleteness arising from excessive precision.
In particular, robustness gives a complete proof calculus for a powerful fragment of mixed integer and real arithmetic and for hybrid systems reachability  \cite{DBLP:conf/ijcar/AbouElWafaP26}.
Identity corresponds to a strict separability of truth and falsity:
\[\text{Principle of Identity}\quad\leftrightsquigarrow\quad\text{Strict Separability of True and False}\]
These applications of \subrefllogic to partiality and robustness motivate the study of subreflexive logic in its own right as the fundamental logic underlying these phenomena. 

\medskip
Subreflexive logic is introduced in \Cref{sec:subrefllogic} as a generalization of ordinary sequent calculus.
Instead of being based on the identity axiom, it is defined relative to an arbitrary choice of axioms.
This enables precise control over the available axioms, and in particular about the properties of implication.
The close proof-theoretic relationship with ordinary propositional logic means that nothing is sacrificed, but significant flexibility for robustness and partiality is gained.
The proof that \subrefllogic admits cut-elimination for suitable sets of axioms shows that it is a well-defined logic.
 
The major technical challenge for \subrefllogic is that standard algebraic, set-theoretic and categorical semantics are all reflexive and, thus, cannot capture subreflexive logic completely.
Because reflexivity plays a critical role in the structural integrity of all these semantic models, new semantics are defined here, which capture the meaning of the logic without using reflexivity and, crucially, without reintroducing it inadvertently.
First, algebraic semantics for \subrefllogic are presented in \Cref{sec:algebraicsem}, which introduces \semiHeyting{} and \semiBooleanalgebra{s} as generalizations of Heyting and Boolean algebras.
These algebraic counterparts to intuitionistic and classical \subrefllogic provide sound and complete algebraic semantics for \subrefllogic.
The key insight is to change the perspective from the entailment relation and to look at relative provability.
Second, and more concretely, \classical \subrefllogic is shown to be sound and complete with respect to natural partial-function and robust-set semantics in \Cref{sec:denotational}.
This shows that in \classical \subrefllogic the principle of identity is equivalent to the law of the excluded middle, \(\fml\limply\fml\equiv\lnot\fml\lor\fml\).
\Classical \subrefllogic can thus be viewed as an alternative (to intuitionistic logic) way of dropping the law of the excluded middle while retaining \emph{all proof rules of classical logic}.
The proof of completeness proceeds through a representation theorem for \semiBooleanalgebra{s}, which realizes every \semiBooleanalgebra as a semialgebra of partial functions, where Kleene's three-valued logic plays an important role as the dualizing object.
Finally, \Cref{sec:categorical} investigates the structural role the principle of identity plays by interpreting \subrefllogic in semicategories.
This requires new definitions of the structures that reflect the logical connectives which usually rely on, or otherwise inadvertently reintroduce, identity morphisms.
To achieve this, a new notion of semiadjunctions is defined that unifies the interpretation and reveals that while identity morphisms are convenient they can be avoided and are not necessary for the structural integrity.

The \textbf{contributions} of this paper are fourfold.
First, a {natural identity-free proof system} fully flexible in its assumptions and a proof of {cut elimination} are presented.
Second, \semiHeyting{} and \semiBooleanalgebra{}s are introduced as {principled nonreflexive generalizations} of Heyting and Boolean algebras and shown to provide sound and complete {algebraic semantics} for \subrefllogic.
Third, a {set-theoretic denotational semantics} is given that interprets truth as robust truth, and a {completeness} proof shows that \subrefllogic axiomatizes these models.
Fourth, a {categorical foundation} for the proof system and the algebraic structures is developed through characterizations of connectives as {semiadjunctions} of semicategories.

\section{\SubreflLogic}\label{sec:subrefllogic}

The calculus of \subrefllogic is a generalization of the usual classical and intuitionistic propositional calculi.
The formulas of \emph{\subrefllogic} are those of propositional logic:
\[\fml ::= \ftrue \mid \ffalse \mid \patom \mid \fml\land\fmlb \mid \fml\lor\fmlb \mid \fml\limply\fmlc,\]
where \(\patom\in\patoms\) are propositional atoms from a set \(\patoms\), and negation $\lnot \fml$ is defined as $\fml \limply \ffalse$.

A \emph{sequent} \(\lsequent{\fmlset}{\fml}\) consists of an (unordered multi-)set \(\fmlset\) of antecedent formulas and a single succedent formula~\(\fml\).\footnote{The use of multisets is a formality, as contraction will be implicit in the calculus.}
As usual \(\fml,\fmlset\) denotes \(\{\fml\}\union \fmlset\).
In the sequent calculus, the connectives are characterized by their left and right rules:

\begin{calculuscollection}{\renewcommand{\linferPremissSeparation}{\hspace{0.3cm}}
        \hspace{-1.5em}
        \begin{calculus}
            \cinferenceRule[andL|$\land$L]{and left proof rule}
            {
                \linferenceRule[sequent]
                {
                    \lsequent{\fml_i,\fml_1\land\fml_2,\fmlset} {\fmlc}
                }
                {\lsequent{\fml_1\land\fml_2,\fmlset} {\fmlc}}\;
            }{\color{darkishgray}$i=0,1$}
            \cinferenceRule[andR|$\land$R]{and right proof rule}
            {
                \linferenceRule[sequent]
                {
                    \lsequent{\fmlset} {\fml}
                    &
                    \lsequent{\fmlset} {\fmlb}
                }
                {\lsequent{\fmlset} {\fml\land\fmlb}}
            }{}
        \end{calculus}
        \quad
        \begin{calculus}
            \cinferenceRule[orL|$\lor$L]{or left proof rule}
            {
                \linferenceRule[sequent]
                {\lsequent{\fml,\fmlset} {\fmlc}
                    &\lsequent{\fmlb,\fmlset} {\fmlc}}
                {\lsequent{\fml\lor \fmlb,\fmlset} {\fmlc}}
            }{}
            \cinferenceRule[orR|$\lor$R]{or right proof rule}
            {
                \linferenceRule[sequent]
                {\lsequent{\fmlset} {\fml_i}}
                {\lsequent{\fmlset} {\fml_1\lor\fml_2}}
            }{\color{darkishgray}$i=0,1$}
        \end{calculus}}
    \quad
    \begin{calculus}
        \cinferenceRule[implyL|$\limply$L]{imply left proof rule}
        {
            \linferenceRule[sequent]
            {\lsequent{\fml\limply\fmlb,\fmlset} {\fml}
                &
                \lsequent{\fmlb,\fmlset} {\fmlc}
            }
            {\lsequent{\fml\limply\fmlb,\fmlset} {\fmlc}}
        }{}
        \cinferenceRule[implyR|$\limply$R]{imply right proof rule}
        {
            \linferenceRule[sequent]
            {\lsequent{\fmlset,\fml} {\fmlb}}
            {\lsequent{\fmlset} {\fml\limply \fmlb}}.
        }{}
    \end{calculus}
\end{calculuscollection}

\noindent
The basic rules of \subrefllogic are exactly the same as in the ordinary intuitionistic sequent calculus \(\LJ\) \cite{DBLP:conf/lics/Pfenning95,DBLP:journals/iandc/Pfenning00}.
However, the identity axiom \(\fml\limply\fml\) is omitted, making \subrefllogic a \emph{substructural logic}.
While these rules characterize all connectives \emph{harmoniously}, there is nothing that can be derived without a terminal sequent.
While in \(\LJ\) only the identity sequents \(\lsequent{\fml}{\fml}\) play this role, in \subrefllogic, derivations are defined relative to a set \(\relfmlset\) of \emph{assumption sequents}.
This section explains how assumptions integrate into \subrefllogic.
For example, the set \(\truefalseaxiomset = \{\lsequent{} {\ftrue}\}\union \{\lsequent{\ffalse}{\patom} : \patom\in\patoms\}\) contains important assumption sequents.
Crucially, it does not suffice to merely allow derivations to be trees where leaves are labeled by terminal assumption sequents (\irref{assumption}).
Additionally, in \subrefllogic the rules \irref{assumptionL} and \irref{assumptionR} handle assumption sequents on the left and on the right

\begin{center}
    \begin{calculuscollection}\small
        \begin{calculus}\renewcommand{\linferPremissSeparation}{\hspace{0.2cm}}
            \cinferenceRule[assumption|$\relfmlset^*$]{assumption axiom}
            {
                \lsequent{\fmlset} {\fml}\;\;
            }{\color{darkishgray}$\lsequent{\fmlset}{\fml}\in \relfmlsetcl$}
        \end{calculus}
        \qquad\hfill%
        \begin{calculus}
            \cinferenceRule[assumptionL|$\relfmlset$L]{assumption left proof rule}
            {
                \linferenceRule[sequent]
                {
                    \lsequent{\fmlsetb,\fmlset} {\fmlb}
                }
                {\lsequent{\fmlset} {\fmlb}}\;
            }{\color{darkishgray}$\allsequents{\fmlset}{\fmlsetb}\in\relfmlsetcl$}
        \end{calculus}
        \qquad\hfill%
        \begin{calculus}
            \cinferenceRule[assumptionR|$\relfmlset$R]{assumption right proof rule}
            {
                \linferenceRule[sequent]
                {
                    \lsequent{\fmlset} {\fml}  \text{ for } \fml\in\fmlsetb
                }
                {\lsequent{\fmlset} {\fmlb}}\;
            }{\color{darkishgray}$\lsequent{\fmlset,\fmlsetb}{\fmlb}\in\relfmlsetcl$}
        \end{calculus}
    \end{calculuscollection}
\end{center}

\noindent
where \(\lsequent{\fmlset}{\fml}\in\relfmlsetcl\) if there is a (repetition-free) subset\footnote{This is a technicality to ensure weakening and contraction are implicit.} \(\fmlset'\subseteq\fmlset\) such that \(\lsequent{\fmlset'}{\fml}\in\relfmlset\) and \(\allsequents{\fmlset}{\fmlsetb}\in\relfmlsetcl\) means  \(\lsequent{\fmlset}{\fmlb}\in\relfmlsetcl\) for all \(\fmlb\in\fmlsetb\).
The \emph{single} axiom scheme of \subrefllogic \irref{assumption} closes a branch, if its last sequent is in the assumption set.
The rule \irref{assumptionL} introduces \(\fmlsetb\) in the antecedent if it follows from \(\fmlset\) by the assumption set and the rule \irref{assumptionR} reduces the succedent to conditions \(\fmlsetb\) that are sufficient according to the assumption set \(\relfmlset\).

A \emph{derivation} of a sequent \(\lsequent{\fmlset}{\fml}\) from the set \(\relfmlset\) of assumption sequents is a tree of sequents with root \(\lsequent{\fmlset}{\fml}\) built according to the rules of \subrefllogic where the leaves are instances of axiom \irref{assumption}.
A formula \(\fml\) is \emph{provable}  from \(\fmlset\) in \subrefllogic \(\subreflprov{\relfmlset}{\fml}\) if there is a derivation of \(\lsequent{}{\fml}\) from \(\relfmlset\union\truefalseaxiomset\).

\Subrefllogic generalizes ordinary intuitionistic logic.
In fact, Gentzen's intuitionistic sequent calculus \LJ is the special case of \subrefllogic where the assumption set \(\idset =\{\lsequent{\patom}{\patom}:\patom\in\patoms\}\) consists of all identity sequents \(\lsequent{\patom}{\patom}\) for atoms \(\patom\in\patoms\).
As with \(\LJ\), \subrefllogic can be specialized to \emph{classical} \subrefllogic by simply allowing multiple succedents in the proof rules in a \emph{classical derivation}.
A formula \(\fml\) is \emph{provable}  from \(\relfmlset\) in \classicalsubrefl \(\classsubreflprov{\relfmlset}{\fml}\) iff there is a \emph{classical} derivation of \(\lsequent{}{\fml}\) from~\(\relfmlset\union\truefalseaxiomset\).
(For a formal definition see \Cref{sec:classicalrules}.)
\Classicalsubrefl generalizes Gentzen's \(\LK\) \cite{Gentzen1935}.

Because of the assumption sets, a derivation may mention formulas that are not subformulas of the conclusion.
However any formula appearing in a derivation is either a subformula of the conclusion or of an assumption \(\relfmlset\) or \(\ftrue,\ffalse\).
Consequently, for finite \(\relfmlset\) there are only finitely many proofs of \(\subreflprov{\relfmlset}{\fml}\) up to repetition and provability can be decided by proof search.

\textEnd{
The following lemma is a useful observation on the axiomatization of \subrefllogic.
It says that the following two rules can be used in place of \irref{assumptionL}.
\begin{center}
    \begin{calculuscollection}\small
        \begin{calculus}
            \cinferenceRule[assumptionOrL|$\relfmlset\lor$]{assumption or left proof rule}
            {
                \linferenceRule[sequent]
                {
                    \lsequent{\fmlset,\fml} {\fmlc}
                    &
                    \lsequent{\fmlset,\fmlb} {\fmlc}
                }
                {\lsequent{\fmlset} {\fmlc}}\;
            }{\color{darkishgray}$\lsequent{\fmlset}{\fml\lor\fmlb}\in\relfmlsetcl$}
        \end{calculus}
        \qquad\qquad%
        \begin{calculus}
            \cinferenceRule[assumptionImplyL|$\relfmlset{\limply}$]{assumption imply left proof rule}
            {
                \linferenceRule[sequent]
                {
                    \lsequent{\fmlset}{\fml}  
                    &
                    \lsequent{\fmlset,\fmlb} {\fmlc}  
                }
                {\lsequent{\fmlset} {\fmlc}}\;
            }{\color{darkishgray}$\lsequent{\fmlset}{\fml\limply\fmlb}\in\relfmlsetcl$}
        \end{calculus}
    \end{calculuscollection}
\end{center}}%
\begin{lemmaE}[][all end]\label{lem:alternativeAssumptionLeft}
    Any sequent that is derivable in \subrefllogic is derivable using the derived rules \irref{assumptionOrL} and \irref{assumptionImplyL} instead of \irref{assumptionL}.
\end{lemmaE}%
\begin{proofE}
    Immediate by induction on the length of a derivation.
\end{proofE}%
\textEnd{\label{p:wellfounded}
\begin{definition}
    A set \(\relfmlset\) is \emph{well-founded} if \(\relfmlset\) is decidable and
    \begin{enumerate}
        \item \(\fml\in\patoms\) is atomic for all \(\lsequent{\fmlset}{\fml}\in\relfmlset\) (i.e., \(\relfmlset\) consists of \emph{definitions}) and
        \item there is no infinite sequence \(\lsequent{\fmlset_i}{\patom_i}\in\relfmlset\) such that \(\patom_{i+1}\) appears in \(\fmlset_i\)
    \end{enumerate}
\end{definition}
}%
\begin{propositionE}
    Provability \(\subreflprov{\relfmlset}{\fml}\) is decidable for \(\relfmlset\) finite or well-founded\footnote{See \Cpageref{p:wellfounded} for the full definition of well-foundedness.}.
\end{propositionE}%
\begin{proofE}
    If \(\relfmlset\) is finite, then there are only finitely many sequents that may appear in a derivation.
    Thus, there are only finitely many proofs (without repetition of sequents) and provability can be decided by searching the space of possible proofs.
    
    Next consider a well-founded \(\relfmlset\).
    Since \(\relfmlset\) contains only sequents with atomic succedent, by \Cref{lem:alternativeAssumptionLeft} provability \(\subreflprov{\relfmlset}{\fml}\) is equivalent to existence of a derivation without an assumption left rule, i.e., \irref{assumptionL}, \irref{assumptionOrL} and \irref{assumptionImplyL}.
    Let \(F_\fml\) be the smallest set of formulas containing \(\fml\) that is closed under subformulas and contains \(\fmlsetb\subseteq F_\fml\) for all \(\lsequent{\fmlsetb}{\patom}\in\relfmlset\) for \(\patom \in F_\fml\).
    Then by well-foundedness of \(\relfmlset\) this set \(F_\fml\) is finite.
    Again, there are only finitely many proofs (without repetition of sequents) witnessing \(\subreflprov{\relfmlset}{\fml}\) and provability can be decided by enumerating these proofs.
\end{proofE}
To ensure that \subrefllogic is a proper logic, it is important that it admits cut-elimination, i.e., that the rule \irref{cut} is (syntactically proved) admissible:
\begin{center}
    \cinferenceRule[cut|cut]{cut proof rule}
    {
        \linferenceRule[sequent]
        {\lsequent{\fmlset} {\fmlb}
            &\lsequent{\fmlb,\fmlset} {\fmlc}}
        {\lsequent{\fmlset} {\fmlc}}
    }{}.
\end{center}
Cut-elimination is proved by moving a cut upwards in the proof tree.
However, the assumption rules \irref{assumption}, \irref{assumptionL} and \irref{assumptionR} introduce challenges.
The proof of cut-elimination shows for which assumption sets \subrefllogic is a properly defined logic.
In particular, the inverse rules and contraction (see \Cref{sec:derivedinverse}) need to be admissible, which is ensured by \saturatedness.

\begin{definition}
    A set of assumption sequents \(\relfmlset\) is (classically) \emph{\saturated} if it is closed \saturated under the (classical) inverse rules, contraction \irref{contraction} and \irref{cut}.
\end{definition}

\begin{lemma}
    The (classical) inverse rules are admissible for (classically) closed \(\relfmlset\).
\end{lemma}

For a proof see \Cref{sec:derivedinverse}. 
\Saturatedness does not restrict the generality of \subrefllogic, as every (finite) set of assumption sequents is contained in a larger (finite) \saturated set and the additional assumption sequents cannot derive anything that was not previously derivable.
A typical example of a \saturated set of assumption sequents is a set of Horn clauses of the form \(\lsequent{\patom_1,\ldots,\patom_n}{\patomb}\) with atomic propositions, \(\patom_1,\ldots,\patom_n,\patomb\in\patoms\).
This can be viewed as a directed hypergraph of atomic propositions where each edge has a single target.
From this perspective, \subrefllogic generalizes to arbitrary such hypergraphs from merely discrete reflexive ones.

In the intuitionistic case, in addition to \saturatedness, the \emph{partial} inverses of \irref{implyL} and \irref{orR} need to be total on the introduced assumption sequents of \(\relfmlset\), which need to be \constructive:

\begin{definition}
    A set of assumption sequents \(\relfmlset\) is \emph{\constructive} if
    \begin{enumerate}
        \item \(\lsequent{\fmlset}{\fml\lor\fmlb}\in\relfmlset\) implies \(\lsequent{\fmlset}{\fml}\in\relfmlset\) or \(\lsequent{\fmlset}{\fmlb}\in\relfmlset\) and
        \item \(\lsequent{\fmlset,\fml\limply\fmlb}{\fmlc}\in\relfmlset\) implies \(\lsequent{\fmlset,\fml\limply\fmlb}{\fml}\in\relfmlset\) or \(\lsequent{\fmlset}{\fmlc}\in\relfmlset\).
    \end{enumerate}
\end{definition}

\Constructivity means that \(\relfmlset\) cannot assume a disjunction without a choice of disjunct and cannot assume an implication without constructive justification.
Assumption sets in Horn-clause are vacuously \constructive.

\begin{propositionE}[Cut Elimination][]
    The \irref{cut} rule \noindent is admissible in intuitionistic \subrefllogic for \emph{\saturated and \constructive} \(\relfmlset\).

    The \irref{cut} rule \noindent is admissible in classical \subrefllogic for \emph{\classicallysaturated} \(\relfmlset\).
\end{propositionE}

\begin{proofE}
    Instead of proving admissibility of \irref{cut} directly, admissibility of the stronger \irref{cutplus} rule is proved.
    \begin{center}
        \cinferenceRule[cutplus|cut${}^+$]{extended cut proof rule}
        {
            \linferenceRule[sequent]
            {\lsequent{\fmlset}{\fmlb}
                &\lsequent{\fmlb,\fmlsetb,\fmlset}{\fmlc}}
            {\lsequent{\fmlset} {\fmlc}}\quad
        }{\color{darkishgray}$\allsequents{\fmlb,\fmlset}{\fmlsetb}\in\relfmlsetcl$}
    \end{center}
    Note that this generalizes \irref{cut}, since by \(\truefalseaxiomset\) the rule is always applicable with \(\fmlsetb = \ftrue\).

    The proof proceeds by a triple nested induction.
    Begin by induction on the complexity of formula \(\fmlb\).
    So \textbf{assume inductively} that \irref{cutplus} is admissible in all instances where the cut formula \(\fmlb'\) is of lower complexity than \(\fmlb\).

    To prove admissibility of \irref{cutplus} for the cut formula \(\fmlb\), again proceed by induction on the length of the right derivation.
    The inductive claim is that for any \(n\geq 0\) any instance of \irref{cutplus} with formula \(\fmlb\), where the right derivation is of depth at most~\(n\) and the left derivation is arbitrary, is admissible.
    \textbf{Suppose inductively} the claim holds for~\(n\).

    It needs to be shown that instances of \irref{cutplus} for \irref{cut} formula \(\fmlb\) where the right derivation is of depth at most~\({n+1}\) are admissible.
    To show this, proceed by another induction on the length of a derivation:
    Suppose for some $\allsequents{\fmlb,\fmlset}{\fmlsetb}\in\relfmlsetcl$, there are (\irref{cut}-free) derivations \(\derivation\) of \(\lsequent{\fmlset}{\fmlb}\) and \(\derivationb\) of \(\lsequent{\fmlsetb,\fmlset}{\fmlc}\), such that \(\derivationb\) is of depth at most \(n+1\).
    We will find a \irref{cut}-free derivation of \(\lsequent{\fmlset}{\fmlc}\) by a case distinction based on the last rule applied in derivation \(\derivationb\).

    \begin{caselist}
        \case{\irref{assumption}}
        In this case \(\allsequents{\fmlb,\fmlset}{\fmlsetb}\in\relfmlsetcl\) and \(\lsequent{\fmlb,\fmlsetb,\fmlset}{\fmlc}\in\relfmlset\).
        Because \(\relfmlset\) is \saturated also \(\lsequent{\fmlb,\fmlset}{\fmlc}\in\relfmlsetcl\).
        Hence the derivation can be reduced to an instance of \irref{assumptionR}.

        \case{\irref{assumptionL}}
        Suppose \(\derivationb\) is of the form
        \begin{sequentdeduction}
            \linfer[assumptionL]
            {
                \linfer
                {
                    \derivationb'
                }
                {   \lsequent{\fmlsetc,\fmlb,\fmlsetb,\fmlset}{\fmlc} }
            }
            {\lsequent{\fmlb,\fmlsetb,\fmlset}{\fmlc}}
        \end{sequentdeduction}
        where \(\allsequents{\fmlb,\fmlsetb,\fmlset}{\fmlsetc}\in\relfmlsetcl\).
        As \(\relfmlset\) is \saturated then \(\allsequents{\fmlb,\fmlset}{(\fmlsetb,\fmlsetc)}\in\relfmlsetcl\).
        Hence the instance
        \begin{sequentdeduction}
            \linfer[cutplus]
            {
                \linfer
                {
                    \derivation
                }
                {   \lsequent{\fmlset}{\fmlb} }
                &
                \linfer
                {
                    \derivationb'
                }
                {   \lsequent{\fmlb,\fmlsetb,\fmlsetc,\fmlset}{\fmlc} }
            }
            {\lsequent{\fmlset}{\fmlc}}
        \end{sequentdeduction}
        of \irref{cutplus} is admissible as \(\derivationb'\) is of depth at most \(n\).

        \case{\irref{assumptionR}}
        If \(\derivationb\) is of the form
        \begin{sequentdeduction}
            \linfer[assumptionR]
            {
                \linfer[]
                {
                    \derivationb'
                }
                {   \lsequent{\fmlb,\fmlsetb,\fmlset}{\fmld}}
            }
            {\lsequent{\fmlb,\fmlsetb,\fmlset}{\fmlc}}
        \end{sequentdeduction}
        where \(\lsequent{\fmlb,\fmlsetb,\fmlset,\fmld}{\fmlc}\in\relfmlsetcl\).
        Because \(\relfmlset\) is \saturated and \(\allsequents{\fmlb,\fmlset}{\fmlsetb}\in\relfmlsetcl\) then \(\lsequent{\fmlb,\fmlset,\fmld}{\fmlc}\in\relfmlsetcl\).
        Hence
        \begin{sequentdeduction}
            \linfer[assumptionR]
            {
                \linfer
                {
                    \derivation
                }
                {   \lsequent{\fmlset}{\fmlb} }
                &
                \linfer[cutplus]
                {
                    \linfer
                    {
                        \derivation
                    }
                    {   \lsequent{\fmlset}{\fmlb} }
                    &
                    \linfer
                    {
                        \derivationb'
                    }
                    {   \lsequent{\fmlb,\fmlsetb,\fmlset}{\fmld} }
                }
                {   \lsequent{\fmlset}{\fmld} }
            }
            {\lsequent{\fmlset}{\fmlc}}
        \end{sequentdeduction}
        gives rise to a \irref{cut}-free derivation, as \(\derivationb'\) is of depth at most \(n\).

        \case{\irref{andR}}
        The cases where a right rule is the last rule applied in \(\derivation\) are straightforward as the cut simply commutes.
        Suppose \(\fmlc\equiv\fmlc_1\land\fmlc_2\) and \(\derivationb\) is of the form
        \begin{sequentdeduction}
            \linfer[andR]
            {
                \linfer[]
                {
                    \derivationb_1
                }
                {   \lsequent{\fmlb,\fmlsetb,\fmlset}{\fmlc_1}}
                &
                \linfer[]
                {
                    \derivation_2
                }
                {   \lsequent{\fmlb,\fmlsetb,\fmlset}{\fmlc_2}}
            }
            {\lsequent{\fmlb,\fmlsetb,\fmlset}{\fmlc}}
        \end{sequentdeduction}
        Then by the inductive hypothesis the derivation
        \begin{sequentdeduction}
            \linfer[andR]
            {
                \linfer[cutplus]
                {
                    \linfer
                    {
                        \derivation
                    }
                    {   \lsequent{\fmlset}{\fmlb} }
                    &
                    \linfer
                    {
                        \derivationb_1
                    }
                    {   \lsequent{\fmlb,\fmlsetb,\fmlset}{\fmlc_1} }
                }
                {   \lsequent{\fmlset}{\fmlc_1} }
                &
                \linfer[cutplus]
                {
                    \linfer
                    {
                        \derivation
                    }
                    {   \lsequent{\fmlset}{\fmlb} }
                    &
                    \linfer
                    {
                        \derivationb_2
                    }
                    {   \lsequent{\fmlb,\fmlsetb,\fmlset}{\fmlc_2} }
                }
                {   \lsequent{\fmlset}{\fmlc_2} }
            }
            {\lsequent{\fmlset}{\fmlc}}
        \end{sequentdeduction}
        gives rise to a \irref{cut}-free derivation, as \(\derivationb_1\) and \(\derivationb_2\) are of depth at most~\(n\).

        \case{\irref{orR}}
        Suppose \(\fmlc\equiv\fmlc_1\lor\fmlc_2\) and \(\derivationb\) is of the form
        \begin{sequentdeduction}
            \linfer[orR]
            {
                \linfer[]
                {
                    \derivationb'
                }
                {   \lsequent{\fmlb,\fmlsetb,\fmlset}{\fmlc_i}}
            }
            {\lsequent{\fmlb,\fmlsetb,\fmlset}{\fmlc}}
        \end{sequentdeduction}
        for some \(i\).
        Then by the inductive hypothesis the derivation
        \begin{sequentdeduction}
            \linfer[orR]
            {
                \linfer[cutplus]
                {
                    \linfer
                    {
                        \derivation
                    }
                    {   \lsequent{\fmlset}{\fmlb} }
                    &
                    \linfer
                    {
                        \derivationb'
                    }
                    {   \lsequent{\fmlb,\fmlsetb,\fmlset}{\fmlc_i} }
                }
                {   \lsequent{\fmlset}{\fmlc_i} }
            }
            {\lsequent{\fmlset}{\fmlc}}
        \end{sequentdeduction}
        gives rise to a \irref{cut}-free derivation, as \(\derivationb'\) is of depth at most~\(n\).

        \case{\irref{implyR}}
        Suppose \(\fmlc\equiv\fmlc_1\limply\fmlc_2\) and \(\derivationb\) is of the form
        \begin{sequentdeduction}
            \linfer[implyR]
            {
                \linfer[]
                {
                    \derivationb'
                }
                {   \lsequent{\fmlb,\fmlsetb,\fmlset,\fmlc_1}{\fmlc_2}}
            }
            {\lsequent{\fmlb,\fmlsetb,\fmlset}{\fmlc}}
        \end{sequentdeduction}
        Then by the inductive hypothesis the derivation
        \begin{sequentdeduction}
            \linfer[orR]
            {
                \linfer[cutplus]
                {
                    \linfer
                    {
                        \derivation
                    }
                    {   \lsequent{\fmlset,\fmlc_1}{\fmlb} }
                    &
                    \linfer
                    {
                        \derivationb'
                    }
                    {   \lsequent{\fmlb,\fmlsetb,\fmlset,\fmlc_1}{\fmlc_2} }
                }
                {   \lsequent{\fmlset,\fmlc_1}{\fmlc_2} }
            }
            {\lsequent{\fmlset}{\fmlc}}
        \end{sequentdeduction}
        gives rise to a \irref{cut}-free derivation, as \(\derivationb'\) is of depth at most~\(n\).
        (Note that technically the derivations \(\derivation\) and \(\derivationb\) derive stronger sequents, but as \irref{weak} is admissible without increasing the depth of the derivation this can be safely ignored.
        the contraction rule \irref{contraction} is admissible without increasing the depth, it can be left implicit.)

        \case{\irref{andL}}
        For the left rules, there is a distinction to be made on whether the rule is applied to a formula from \(\fmlset\) or from \(\fmlsetb\) or to \(\fmlb\).
        In general for some \(\fmle_1\land\fmle_2\in \fmlset,\fmlsetb,\fmlb\) the derivation \(\derivationb\) is of the form

        \begin{sequentdeduction}
            \linfer[andL]
            {
                \linfer[]
                {
                    \derivationb'
                }
                {   \lsequent{\fmle_i,\fmlb,\fmlsetb,\fmlset}{\fmlc}}
            }
            {\lsequent{\fmlb,\fmlsetb,\fmlset}{\fmlc}}
        \end{sequentdeduction}

        \emph{Subcase 1}: Suppose the principal formula of \irref{andL} is \(\fmle_1\land\fmle_2\in\fmlset\).
        The derivation
        \begin{sequentdeduction}
            \linfer[andL]
            {
                \linfer[cutplus]
                {
                    \linfer
                    {
                        \derivation
                    }
                    {   \lsequent{\fmle_i,\fmlset}{\fmlb} }
                    &
                    \linfer
                    {
                        \derivationb'
                    }
                    {   \lsequent{\fmle_i,\fmlb,\fmlsetb,\fmlset}{\fmlc} }
                }
                {   \lsequent{\fmle_i,\fmlset}{\fmlc} }
            }
            {\lsequent{\fmlset}{\fmlc}}
        \end{sequentdeduction}
        yields a cut-free derivation by induction hypothesis as \(\derivationb'\) is of depth at most \(n\).

        \emph{Subcase 2}: Suppose the principal formula of \irref{andL} is \(\fmle_1\land\fmle_2\in\fmlsetb\).
        Then because \(\relfmlset\) is \saturated and \(\lsequent{\fmlset,\fmlb}{\fmle_1\land\fmle_2}\in\relfmlset\) also \(\allsequents{\fmlset,\fmlb}{(\fmlsetb,\fmle_i)}\in\relfmlset\).
        Then
        \begin{sequentdeduction}
            \linfer[cutplus]
            {
                \linfer
                {
                    \derivation
                }
                {   \lsequent{\fmlset}{\fmlb} }
                &
                \linfer
                {
                    \derivationb'
                }
                {   \lsequent{\fmlb,\fmlsetb,\fmle_i,\fmlset}{\fmlc} }
            }
            {\lsequent{\fmlset}{\fmlc}}
        \end{sequentdeduction}
        yields a cut-free derivation by induction hypothesis as \(\derivationb'\) is of depth at most \(n\).
        Note that the cut is not shifted upwards, but absorbed by the assumptions \(\relfmlset\).

        \emph{Subcase 3}: Suppose the principal formula of \irref{andL} is \(\fmlb\equiv\fmlb_1\land\fmlb_2\).
        The derivation

        \begin{sequentdeduction}
            \linfer[cut]
            {
                \linfer[andLminus]
                {
                    \linfer
                    {
                        \derivation
                    }
                    {   \lsequent{\fmlset}{\fmlb} }
                }
                {   \lsequent{\fmlset}{\fmlb_i} }
                &
                \linfer[cutplus]
                {
                    \linfer
                    {
                        \derivation
                    }
                    {   \lsequent{\fmlb_i,\fmlset}{\fmlb} }
                    &
                    \linfer
                    {
                        \derivationb'
                    }
                    {   \lsequent{\fmlb,\fmlsetb,\fmlb_i,\fmlset}{\fmlc} }
                }
                {   \lsequent{\fmlb_i,\fmlset}{\fmlc} }
            }
            {\lsequent{\fmlset}{\fmlc}}
        \end{sequentdeduction}
        can by the inductive assumptions be turned into a cut-free derivation.
        The upper instance of \irref{cutplus} is admissible by the inductive hypothesis as \(\derivationb'\) is of depth at most \(n\) and the lower \irref{cut} is admissible by the outer induction, as \(\fmlb_i\) is of lower complexity.

        \case{\irref{orL}} For some \(\fmle_1\lor\fmle_2\in \fmlset,\fmlsetb,\fmlb\) the derivation \(\derivationb\) is of the form
        \begin{sequentdeduction}
            \linfer[orL]
            {
                \linfer[]
                {
                    \derivationb_1
                }
                {   \lsequent{\fmle_1,\fmlb,\fmlsetb,\fmlset}{\fmlc}}
                &
                \linfer[]
                {
                    \derivationb_2
                }
                {   \lsequent{\fmle_2,\fmlb,\fmlsetb,\fmlset}{\fmlc}}
            }
            {\lsequent{\fmlb,\fmlsetb,\fmlset}{\fmlc}}
        \end{sequentdeduction}

        \emph{Subcase 1}: Suppose the principal formula of \irref{orL} is \(\fmle_1\lor\fmle_2\in\fmlset\).
        Then the derivation
        \begin{sequentdeduction}
            \linfer[orL]
            {
                \linfer[cutplus]
                {
                    \linfer
                    {
                        \derivation
                    }
                    {   \lsequent{\fmle_1,\fmlset}{\fmlb} }
                    &
                    \linfer
                    {
                        \derivationb_1
                    }
                    {   \lsequent{\fmle_1,\fmlb,\fmlsetb,\fmlset}{\fmlc} }
                }
                {   \lsequent{\fmle_1,\fmlset}{\fmlc} }
                &
                \linfer[cutplus]
                {
                    \linfer
                    {
                        \derivation
                    }
                    {   \lsequent{\fmle_2,\fmlset}{\fmlb} }
                    &
                    \linfer
                    {
                        \derivationb_2
                    }
                    {   \lsequent{\fmle_2,\fmlb,\fmlsetb,\fmlset}{\fmlc} }
                }
                {   \lsequent{\fmle_2,\fmlset}{\fmlc} }
            }
            {\lsequent{\fmlset}{\fmlc}}
        \end{sequentdeduction}
        yields a cut-free derivation by induction hypothesis as \(\derivationb_1,\derivationb_2\) are of depth at most \(n\).

        \emph{Subcase 2}: Suppose the principal formula of \irref{orL} is \(\fmle_1\lor\fmle_2\in\fmlsetb\).
        Then because \(\relfmlset\) is \constructive and \(\lsequent{\fmlset,\fmlb}{\fmle_1\lor\fmle_2}\in\relfmlset\) also \(\allsequents{\fmlset,\fmlb}{(\fmlsetb,\fmle_i)}\in\relfmlset\) for some \(i\in\{0,1\}\).
        Then
        \begin{sequentdeduction}
            \linfer[cutplus]
            {
                \linfer
                {
                    \derivation
                }
                {   \lsequent{\fmlset}{\fmlb} }
                &
                \linfer
                {
                    \derivationb_i
                }
                {   \lsequent{\fmlb,\fmlsetb,\fmle_i,\fmlset}{\fmlc} }
            }
            {\lsequent{\fmlset}{\fmlc}}
        \end{sequentdeduction}
        yields a cut-free derivation by induction hypothesis as \(\derivationb_i\) is of depth at most \(n\).

        \emph{Subcase 3}: Suppose the principal formula of \irref{orL} is \(\fmlb\equiv\fmlb_1\lor\fmlb_2\).
        Note by induction on derivations that there is a derivation \(\derivation'\) of \(\lsequent{\fmlset}{\fmlb_i} \) for some \(i\in\{1,2\}\).
        \begin{sequentdeduction}
            \linfer[cut]
            {
                \linfer
                {
                    \derivation'
                }
                {   \lsequent{\fmlset}{\fmlb_i} }
                &
                \linfer[cutplus]
                {
                    \linfer
                    {
                        \derivation
                    }
                    {   \lsequent{\fmlb_i,\fmlset}{\fmlb} }
                    &
                    \linfer
                    {
                        \derivationb_i
                    }
                    {   \lsequent{\fmlb,\fmlsetb,\fmlb_i,\fmlset}{\fmlc} }
                }
                {   \lsequent{\fmlb_i,\fmlset}{\fmlc} }
            }
            {\lsequent{\fmlset}{\fmlc}}
        \end{sequentdeduction}
        yields a cut-free derivation, where the upper \irref{cutplus} is admissible by the inductive hypothesis as \(\derivationb_i\) is of depth at most \(n\) and the lower \irref{cut} is admissible by the outer induction for any derivation depth, as \(\fmlb_i\) is of lower complexity.

        \case{\irref{implyL}} For some \(\fmle_1\limply\fmle_2\in \fmlset,\fmlsetb,\fmlb\) the derivation \(\derivationb\) is of the form
        \begin{sequentdeduction}
            \linfer[implyL]
            {
                \linfer[]
                {
                    \derivationb_1
                }
                {   \lsequent{\fmlb,\fmlsetb,\fmlset}{\fmle_1}}
                &
                \linfer[]
                {
                    \derivationb_2
                }
                {   \lsequent{\fmle_2,\fmlb,\fmlsetb,\fmlset}{\fmlc}}
            }
            {\lsequent{\fmlb,\fmlsetb,\fmlset}{\fmlc}}
        \end{sequentdeduction}

        \emph{Subcase 1}: Suppose the principal formula of \irref{implyL} is \(\fmle_1\limply\fmle_2\in\fmlset\).
        The derivation
        \begin{sequentdeduction}
            \linfer[implyL]
            {
                \linfer[cutplus]
                {
                    \linfer
                    {
                        \derivation
                    }
                    {   \lsequent{\fmlset}{\fmlb} }
                    &
                    \linfer
                    {
                        \derivationb_1
                    }
                    {   \lsequent{\fmle_i,\fmlb,\fmlsetb,\fmlset}{\fmlc} }
                }
                {   \lsequent{\fmlset}{\fmle_1} }
                &
                \linfer[cutplus]
                {
                    \linfer
                    {
                        \derivation
                    }
                    {   \lsequent{\fmle_2,\fmlset}{\fmlb} }
                    &
                    \linfer
                    {
                        \derivationb_2
                    }
                    {   \lsequent{\fmle_2,\fmlb,\fmlsetb,\fmlset}{\fmlc} }
                }
                {   \lsequent{\fmle_2,\fmlset}{\fmlc} }
            }
            {\lsequent{\fmlset}{\fmlc}}
        \end{sequentdeduction}
        yields a cut-free derivation by induction hypothesis as both \(\derivationb_i\) are of depth at most \(n\).

        \emph{Subcase 2}: Suppose the principal formula of \irref{implyL} is \(\fmle_1\limply\fmle_2\in\fmlsetb\).
        Then because \(\relfmlset\) is \constructive and \(\lsequent{\fmlset,\fmlb}{\fmle_1\limply\fmle_2}\in\relfmlset\) either \(\allsequents{\fmlset,\fmlb}{(\fmlsetb,\fmle_1)}\in\relfmlset\) or \(\allsequents{\fmlset}{(\fmlsetb,\fmle_2)}\in\relfmlset\).
        In either case for the respective \(i\in\{0,1\}\) the derivation
        \begin{sequentdeduction}
            \linfer[cutplus]
            {
                \linfer
                {
                    \derivation
                }
                {   \lsequent{\fmlset}{\fmlb} }
                &
                \linfer
                {
                    \derivationb_i
                }
                {   \lsequent{\fmlb,\fmlsetb,\fmle_i,\fmlset}{\fmlc} }
            }
            {\lsequent{\fmlset}{\fmlc}}
        \end{sequentdeduction}
        yields a cut-free derivation by induction hypothesis as \(\derivationb_i\) is of depth at most \(n\).

        \emph{Subcase 3}: Suppose the principal formula of \irref{implyL} is \(\fmlb\equiv\fmlb_1\limply\fmlb_2\).
        The derivation
        \begin{sequentdeduction}{\renewcommand{\linferPremissSeparation}{\hspace{0.6cm}}
                \linfer[cut]
                {
                    \linfer[cut]
                    {
                        \linfer[cutplus]
                        {\renewcommand{\linferPremissSeparation}{\hspace{0.2cm}}
                            \linfer
                            {
                                \derivation
                            }
                            {   \lsequent{\fmlset}{\fmlb} }
                            &
                            \linfer
                            {
                                \derivationb_1
                            }
                            {   \lsequent{\fmlb,\fmlsetb,\fmlset}{\fmlb_1} }
                        }
                        {   \lsequent{\fmlset}{\fmlb_1} }
                        &
                        \linfer[implyRminus]
                        {
                            \linfer
                            {
                                \derivation
                            }
                            {   \lsequent{\fmlset}{\fmlb} }
                        }
                        {   \lsequent{\fmlb_1,\fmlset}{\fmlb_2} }
                    }
                    {   \lsequent{\fmlset}{\fmlb_2} }
                    &
                    \linfer[cutplus]
                    {\renewcommand{\linferPremissSeparation}{\hspace{0.2cm}}
                        \linfer
                        {
                            \derivation
                        }
                        {   \lsequent{\fmlset}{\fmlb} }
                        &
                        \linfer
                        {
                            \derivationb_2
                        }
                        {   \lsequent{\fmlb_2,\fmlsetb,\fmlset}{\fmlc} }
                    }
                    {   \lsequent{\fmlb_2,\fmlset}{\fmlc} }
                }
                {\lsequent{\fmlset}{\fmlc}}
            }\end{sequentdeduction}
        gives a cut-free derivation, where the two instances of \irref{cut} are admissible by the outer induction hypothesis on \(\fmlb_1,\fmlb_2\) of lower complexity and the instances of \irref{cutplus} are admissible as \(\derivationb_1,\derivationb_2\) are of depth at most \(n\).
    \end{caselist}

    The classical case follows the same inductive approach and is a straightforward adaptation.
\end{proofE}
The proof is similar to proofs of cut-elimination in ordinary logic \cite{DBLP:conf/lics/Pfenning95,DBLP:journals/iandc/Pfenning00}.
It proceeds by a doubly nested induction, first on the cut formula, and then on the length of the right derivation to push the cut upwards in the derivation.
Additional difficulties arise with respect to the assumption sequents.
Most importantly, the proof reveals which assumption sets are suitable for \subrefllogic.
In general, proof-theoretically \subrefllogic is very similar to ordinary logic and dropping identity interestingly is not a proof-theoretic sacrifice so that many standard results easily transfer to \subrefllogic.

\section{Semialgebraic Semantics}\label{sec:algebraicsem}

Defining complete semialgebraic semantics for \subrefllogic requires a possibly nonreflexive semantic entailment relation.
The challenge is to identify the correct semantic structure that characterizes the connectives fully, without inadvertently reintroducing reflexivity.
\Cref{sec:semiheytings,sec:semibool} introduce \semiHeyting{} (\semiBoolean) \semialgebra{}s as principled, natural and appropriate nonreflexive generalizations of Heyting and Boolean algebras.
The soundness and completeness proof in \Cref{sec:semisemantics} confirms that the Heyting (Boolean) semialgebras correctly and fully capture subreflexive reasoning.

\subsection{\SemiHeytingAlgebra{}s}\label{sec:semiheytings}

A transitive semantic entailment relation \(\shleq\) usually satisfies the typical properties of a conjunction \(\saa\shleq\saa\shand\saa\) and \(\saa\shand\sab\shleq\saa\).
However, this implicitly and inadvertently reintroduces reflexivity~\(\saa\shleq\saa\).
The key insight is to work instead with \emph{relative semantic entailment} via the upper \(\shupset{\saa}=\{\sab: \saa\shleq\sab\}\) and the lower set \(\shdownset{\saa}=\{\sab: \sab\shleq\saa\}\) of an element \(\saa\).
The following new derived order captures this notion.

\begin{definition}[\yonedaorder]
    Given a binary relation \(\shleq\), its associated  \emph{\yonedaorder}~\(\shyol\) is defined by: \(\saa\shyol\sab\) if and only if \(\shdownset{\saa}\subseteq\shdownset{\sab}\) and \(\shupset{\sab}\subseteq\shupset{\saa}\).

    Elements \(\saa,\sab\) are \emph{\yonedaisomorphic} \(\saa\shiso\sab\)  if \(\saa\shyol\sab\) and \(\sab\shyol\saa\).
\end{definition}

For an entailment relation \(\shleq\), the \yonedaorder \(\saa\shyol\sab\) describes that anything entailing \(\saa\) entails \(\sab\), and anything entailed by \(\sab\) is entailed by \(\saa\).
\yonedaisomorphism \(\saa\shiso\sab\) says that two elements behave the same with respect to the relation, although they may be different elements.
Note that the \yonedaorder refines the original order \(\shleq\) as \(\saa\shyol\sab\) for any \(\saa\shleq\sab\).
Conversely, \(\shleq\) is reflexive if and only if it coincides with \(\shyol\).

In the reflexive contexts, meets \(\shand\), joins \(\shor\), and relative pseudocomplements \(\shimply\) are characterized by adjunctions such as \(\sac \shleq \saa\shand\sab\) if and only if \(\sac\shleq\saa\) and \(\sac\shleq\sab\).
In combination with reflexivity, this is enough to ensure that the operations behave as required.
The adjunction remains the core in the subreflexive context, however it demands a more subtle definition, since, for example, \(\shand\), defined by its adjunction, need not even be monotone without reflexivity.
The following defines the appropriate nonreflexive generalization of Heyting algebras.

\begin{definition}[\SemiHeytingAlgebra{}s]\newcommand{\rightalign}{\hfill}\label{def:semiheytingalgebra}
    A \emph{\semiHeytingalgebra{}} \((\shs,\shleq,\sazero,\saone,\shand,\shor,\shimply)\) is a structure, with a binary relation \(\shleq\), constants \(\saone\), \(\sazero\) and binary operations \(\shand\), \(\shor\) and \(\shimply\) with:
    \begin{enumerate}[label=H\arabic*]
        \item \(\shleq\) is a transitive relation \rightalign(semicategory)\label{it:transitive}\label{it:associativeand}
        \item \(\shand\) is \yonedaassociative, i.e., \((\saa\shand\sab)\shand\sac\shiso\saa\shand(\sab\shand\sac)\) \rightalign(associativity)
        \item \(\sazero\shleq\saa\shleq\saone\)\rightalign(terminal)\label{it:zeroone}
        \item\label{it:andprops} \begin{enumerate}[leftmargin=*]
                  \item \(\sac \shleq \saa\shand\sab\) iff \(\sac\shleq\saa\) and \(\sac\shleq\sab\) \rightalign(product adjunction)\label{it:adjand}
                  \item \(\saa\shyol\saa\shand\saa\) \rightalign(contraction)\label{it:anddouble}
                  \item \(\saa\shand\sab\shyol\saa\) and \(\saa\shand\sab\shyol\sab\) \rightalign(projections)\label{it:andproj}
              \end{enumerate}
        \item\label{it:orprops}
              \begin{enumerate}[leftmargin=*]
                  \item \(\saa\shor\sab \shleq \sac\) iff \(\saa\shleq\sac\) and \(\sab\shleq\sac\) \rightalign(sum adjunction)\label{it:adjor}
                  \item \(\saa\shyol\saa\shor\sab\) and \(\sab\shyol\saa\shor\sab\) \rightalign(injections)\label{it:orinj}
                  \item \(\saa\shor\saa\shyol\saa\) \rightalign(contraction)\label{it:ordouble}
              \end{enumerate}
        \item\label{it:implprops} \begin{enumerate}[leftmargin=*]
                  \item \(\saa\shand\sab \shleq \sac\) iff \(\saa\shleq (\sab \shimply \sac)\) \rightalign(exponential adjunction)\label{it:adjimply}
                  \item \(\sac\shyol\saa\shimply(\sab\shand\sac)\) if \(\saa\shleq\sab\)\rightalign(weakening)\label{it:implw}
                  \item \((\sab\shimply\sac)\shand\saa\shyol\sac\) if \(\saa\shleq\sab\) \rightalign(evaluation)\label{it:implev}
              \end{enumerate}
    \end{enumerate}
    In a \semiHeytingalgebra{}, negation \(\lnot\saa\) is defined as \(\saa\shimply\sazero\).
\end{definition}
\noindent
Most of the clauses correspond to properties in ordinary Heyting algebras.
However, the careful definitions in terms of \(\shleq\) and \(\shyol\) ensure that they capture the desired meaning without inadvertently reintroducing reflexivity.
In particular, \(\shiso\) plays the role of equality in a \semiHeytingalgebra and, since \(\shyol\) is not necessarily antisymmetric, associativity can only be assumed up to equivalence \(\shiso\) (\ref{it:associativeand}).
Categorically speaking, the items \ref{it:anddouble}, \ref{it:orinj} and \irref{it:implw} (\ref{it:andproj}, \irref{it:ordouble} and \irref{it:implev}) are versions of the units (counits) of the adjunctions \ref{it:adjand}, \ref{it:adjor} and \ref{it:adjimply}, respectively.
The side condition \(\saa\shleq\sab\) in \ref{it:implw} and \ref{it:implev} is the key and a major subtle difference from Heyting algebras.
For example, adding \((\saa\shimply\sac)\shand\saa\shyol\sac\) instead of \ref{it:implev} would be too restrictive.

\Cref{sec:categorical} will explain in detail why the structure, in particular \ref{it:implw} and \ref{it:implev}, is defined in this way.
While the connectives are monotone, unlike in a Heyting algebra, they are only monotone \emph{in both places simultaneously}, e.g., \(\saa\shand\sab\shleq\saa'\shand\sab'\) if \(\saa\shleq\saa'\) and \(\sab\shleq\sab'\) (as shown in \Cref{lem:monotonicities}).
These monotonicity properties can be used in place of \ref{it:adjand}, \ref{it:adjor} and \ref{it:adjimply} in \Cref{def:semiheytingalgebra}.
\Cref{sec:categorical} characterizes the defining properties of a \semiHeytingalgebra all as forms of adjunctions in semicategories.

\begin{lemmaE}[Monotonicity in \(\shleq\)][all end]\label{lem:monotonicities}
    For \(\saa\shleq\saa'\) and \(\sab\shleq\sab'\):
    \begin{enumerate}
        \item \(\saa\shand\sab\shleq\saa'\shand\sab'\) \label{it:andmonotonicity}
        \item \(\saa\shor\sab\shleq\saa'\shor\sab'\)\label{it:ormonotonicity}
        \item \(\saa'\shimply\sab\shleq\saa\shimply\sab'\)\label{it:implymonotonicity}
    \end{enumerate}
\end{lemmaE}
\begin{proofE}
    \begin{caselist}
        \case{(\ref{it:andmonotonicity})}
        Because \(\saa\shleq\saa'\) and \(\saa\shand\sab\shyol\saa\) (\ref{it:andproj}) also \(\saa\shand\sab\shleq\saa'\) and symmetrically \(\saa\shand\sab\shleq\sab'\).
        By \ref{it:adjand} this implies \(\saa\shand\sab\shleq\saa'\shand\sab'\).

        \case{(\ref{it:ormonotonicity})}
        Because \(\saa\shleq\saa'\) and \(\saa'\shyol\saa'\shor\sab'\) (\ref{it:orinj}) also \(\saa\shleq\saa'\shor\sab'\) and symmetrically \(\sab\shleq\saa'\shor\sab'\).
        By \ref{it:adjor} this implies \(\saa\shor\sab\shleq\saa'\shor\sab'\).

        \case{(\ref{it:implymonotonicity})}
        Note that \(\saa\shand(\saa'\shimply\sab)\shyol\sab\) by \ref{it:implev}, so that \(\saa\shand(\saa'\shimply\sab)\shleq \sab'\).
        Using \ref{it:adjimply} this implies \(\saa'\shimply\sab\shleq\saa\shimply\sab'\).
    \end{caselist}
\end{proofE}

\begin{example}
    The interval \([0,1]\) ordered by \emph{strict} inclusion \(<\) (with added relations \(0<0\) and \(1<1\)), minimum \(\rex\shand\rey=\min\{\rex,\rey\}\), maximum \(\rex\shor\rey=\max\{\rex,\rey\}\) and threshold implication
    \[
        \rex\shimply\rey =\begin{cases}
            \saone & \text{if } \rex<\rey \\
            y      & \text{otherwise}
        \end{cases}
    \]
    forms a \semiHeytingalgebra{}.
    The derived \yonedaorder is \(\rex\shyol\rey\) iff \(\rex\leq\rey\).
    The interval algebra is denoted \(\intha\) and is a nonreflexive version of Gödel's many-valued logic \cite{10.1093/oso/9780195147209.003.0026}.

    \label{ex:realthree}
    The three-element subalgebra \(\threealg=\{0,\midpoint,1\}\) of \(\intha\) is a minimal example of a \semiHeytingalgebra{} that is not reflexive.
    In this the basic instance for intuitionistic robust reasoning, \(0\) represents falsehood, \(1\) truth and \(\midpoint\) represents marginal truth.
    This is a nonreflexive version of Heyting's three-valued logic \cite{HeytingDieFR}, except that \(\midpoint\shimply\midpoint=\midpoint\) and not \(\midpoint\shimply\midpoint=1\).
\end{example}

\begin{example}\label{ex:strangethree}
    Although \semiHeytingalgebra{s} superficially only differ in the assumption of reflexivity, this has far-reaching implications.
    Let \(\threealg'\) be exactly the same as \(\threealg\), except that \(\midpoint\shor\sazero=\saone\); then \(\threealg'\) is still a \semiHeytingalgebra.
    This example shows that~\(\shor\) need not be \yonedaassociative, \yonedaacommutative \(\rex \shor\rey\not\shiso\rey\shor\rex\) nor necessarily distributive \(\rex\shor(\rey\shand\rez)\not\shiso(\rex\shor\rey)\shand(\rex\shor\rez)\) (e.g., \(\rex=\midpoint=\rez\) and \(\rey=0\)).
    In fact, \(\shor\) is not even necessarily monotone with respect to \(\shyol\), i.e., there are \(\rey\shyol\rez\) such that \(\rex\shor\rey\not\shyol\rex\shor\rez\).
\end{example}

\Cref{ex:strangethree} shows that the adjunction property of \(\shor\) cannot be lifted to the \(\shyol\) order; the more complex description of \semiHeytingalgebra{s} in terms of the two orders \(\shleq\) and \(\shyol\) is thus necessary.
In contrast, due to \ref{it:adjimply}, the adjunction for \(\shand\) lifts to the \yonedaorder:
\[\saa\shyol\sab\shand\sac \iff \saa\shyol\sab\text{ and } \saa\shyol\sac\]
and consequently \(\shand\) is \yonedaacommutative \(\saa\shand\sab\shiso\sab\shand\saa\) and monotone with respect to the \yonedaorder (as shown in \Cref{lem:andproperties}).

\begin{lemmaE}[Properties of \(\shand\)][all end]\label{lem:andproperties}
    \begin{enumerate}
        \item \(\saa\shand\sab\shiso\sab\shand\saa\)\label{it:commutativeand}
        \item \(\saa\shand\sab\shyol\saa'\shand\sab'\) for all \(\saa\shyol\saa'\) and \(\sab\shyol\sab'\) \label{it:yonedaandmonotonicity}
        \item \(\saa\shyol\sab\shand\sac\) iff \(\saa\shyol\sab\) and \(\saa\shyol\sac\)\label{it:andyonadj}
        \item \(\saone\shand\saa\shiso\saa\shiso\saa\shand\saone\)\label{it:andoneneutral}
        \item \(\lnot\sab\shand\saa\shiso \sazero\) if \(\saa\shleq\sab\)\label{it:andnegzero}
        \item \(\saa\shand\sazero\shiso \sazero\shiso\sazero\shand\saa\)\label{it:andzero}
    \end{enumerate}
\end{lemmaE}
\begin{proofE}
    Before proving the lemma we prove \(\saa\shand\sab\shyol\saa'\shand\sab\) for all \(\saa\shyol\saa'\).
    Consider \(\sac\shleq\saa\shand\sab\).
    By \ref{it:adjand} then \(\sac\shleq\saa\) and \(\sac\shleq\sab\).
    Because \(\saa\shyol\saa'\) also \(\sac\shleq\saa'\) and hence \(\sac\shleq\saa'\shand\sab\) again by \ref{it:adjand}.
    Consider next \(\saa'\shand\sab\shleq\sac\).
    By \ref{it:adjimply} then \(\saa'\shleq\sab\shimply\sac\) and by \(\saa\shyol\saa'\) also \(\saa\shleq\sab\shimply\sac\).
    And again by \ref{it:adjimply} then \(\saa\shand\sab\shleq\sac\).

    \begin{caselist}
        \case{(\ref{it:commutativeand})}
        By \ref{it:anddouble} note \(\saa\shand\sab\shyol(\saa\shand\sab)\shand(\saa\shand\sab)\).
        Using the claim above with \(\saa\shand\sab\shyol\sab\) (by \ref{it:andproj}) it follows that \(\saa\shand\sab\shyol\sab\shand(\saa\shand\sab)\).
        Then by \yonedaassociativity of \(\shand\) and \ref{it:andproj} again \(\saa\shand\sab\shyol\sab\shand\saa\).

        \case{(\ref{it:yonedaandmonotonicity})} Follows by the preliminary claim and commutativity \ref{it:commutativeand}.

        \case{(\ref{it:andyonadj})}
        The forward implication is immediate because \(\sab\shand\sac\shyol\sab\) and \(\sab\shand\sac\shyol\sac\) by \ref{it:andproj}.
        For the converse consider \(\sad\shleq\saa\) and note that then \(\sad\shleq\sab\) and \(\sad\shleq\sac\) since by assumption \(\saa\shyol\sab\) and \(\saa\shyol\sac\).
        Hence by \ref{it:adjand} also \(\sad\shleq\sab\shand\sac\).
        Now consider \(\sab\shand\sac\shleq\sad\), then \(\sab\shleq\sac\shimply\sad\) by \ref{it:adjimply}.
        Because \(\saa\shyol\sab\) by assumption then \(\saa\shleq\sac\shimply\sad\) and hence by \ref{it:adjimply} again \(\saa\shand\sac\shleq\sad\).
        By commutativity (\ref{it:commutativeand}) then \(\sac\shand\saa\shleq\sad\).
        Again by \ref{it:adjimply} then \(\sac\shleq\saa\shimply\sad\) and thus \(\saa\shleq\saa\shimply\sad\) using \(\saa\shyol\sac\).
        Finally by \ref{it:adjimply} again \(\saa\shand\saa\shleq\sad\) and then by \ref{it:anddouble} then \(\saa\shleq\sad\).

        \case{(\ref{it:andoneneutral})}
        Note that \(\saone\shand\saa\shyol\saa\) by \Cref{it:andproj}.
        By \Cref{it:andmonotonicity} (using \(\saa\shyol\saone\) by \ref{it:zeroone}) then \(\saa\shand\saa\shleq\saone\shand\saa\).
        So the claim follows as \(\saa\shyol\saa\shand\saa\) by \ref{it:anddouble}.

        \case{(\ref{it:andzero})}
        Note that \(\saa\shand\sazero\shyol\sazero\) by \ref{it:andproj} and \(\sazero\shyol\saa\shand\sazero\) by \irref{it:zeroone}.
        The other equivalence follows from commutativity~\ref{it:commutativeand}.

        \case{(\ref{it:andzero})}
        Note that \(\saa\shand\sazero\shyol\sazero\) by \ref{it:andproj} and \(\sazero\shyol\saa\shand\sazero\) by \irref{it:zeroone}.
        The other equivalence follows from commutativity~\ref{it:commutativeand}.
    \end{caselist}
\end{proofE}

A \emph{homomorphism of \semiHeytingalgebra{}s} \(f:\shs\to\shs'\) is a \(\shleq\)-monotone function respecting the \semiHeytingalgebra{} structure, i.e.,
\(f(\sazero)=\sazero\), \(f(\saone)=\saone\),
\(\saa\shleq\sab \implies f(\saa)\shleq f(\sab) \),
\(f(\sab \shand \sac)= f(\sab) \shand f(\sac)\),
\(f(\sab \shor \sac) = f(\sab) \shor f(\sac)\) and
\( f(\sab\shimply\sac)=f(\sab) \shimply f(\sac)\).
If the homomorphism is \emph{order-reflecting}, i.e., \(\saa\shleq\sab\) whenever \(f(\saa)\shleq f(\sab)\), \(\shs\) is \emph{reflected} in \(\shs'\).

The next proposition shows that \semiHeytingalgebra{s} exactly generalize Heyting algebras to the nonreflexive setting.
It shows that while \SemiHeytingalgebra{} do not need to be reflexive, they reduce to ordinary Heyting algebras in case they are reflexive.

\begin{propositionE}[][]
    A \semiHeytingalgebra{} is a pre-Heyting\footnote{A \emph{pre-Heyting algebra} is a Heyting algebra, except that the order does not need to be antisymmetric} algebra iff \(\shleq\) is reflexive.\label{prop:reflexivesemiheyting}

    Every reflexive \semiHeytingalgebra is reflected in a Heyting algebra.
\end{propositionE}

\begin{proofE}
    If \(\shs\) is a Heyting algebra, then \(\shleq\) is reflexive by definition.
    Conversely, assume that \(\shs\) is a reflexive \semiHeytingalgebra{}, so that \(\saa\shyol\sab\) iff \(\saa\shleq\sab\).
    It is easy to see that \(\shs\) is a bounded lattice with meet \(\shand\) and join \(\shor\).
    Moreover by \ref{it:adjimply} the operation \(\shimply\) is a relative pseudo-complement.

    The second claim follows since every pre-Heyting algebra is reflect in a Heyting algebra.
\end{proofE}
The reason that a reflexive \semiHeytingalgebra is merely a pre-Heyting algebras is that it is a Heyting algebra with respect to \(\shiso\), which does not need to be strict equality.
In fact, as \Cref{ex:strangethree} shows, it is not necessarily a congruence relation, and therefore might not even admit an antisymmetric quotient.

\subsection{\SemiBooleanAlgebra{}s}\label{sec:semibool}

Just like Heyting algebras, \SemiHeytingalgebra{}s capture intuitionistic reasoning and classical reasoning corresponds to the subclass of Boolean algebras.
Ordinarily, a Boolean algebra is a Heyting algebra with actual complements, i.e., satisfying \(\saa\shor\lnot\saa\shiso\saone\).
However, in the context of the other rules, this assumption again inadvertently reintroduces the identity axiom.
The following more careful definition avoids this and captures classical subreflexive reasoning.

\begin{definition}[\SemiBooleanAlgebra{}]
    A \semiHeytingalgebra{} is \emph{\classical} iff \({\saa\shor\sab}\shiso{\lnot\saa\shimply\sab}\) and \({\lnot\saa\shor\sab}\shiso{\saa\shimply\sab}\).
    A \classical \semiHeytingalgebra{} is called a \emph{\semiBooleanalgebra{}}.
\end{definition}

Classicality removes many of the difficulties of \semiHeytingalgebra{s} (\Cref{ex:strangethree}) and brings the theory closer to the standard theory of Boolean algebras.

For example, in a \SemiBooleanalgebra, \(\shor\) is \yonedaassociative, \yonedaacommutative and monotone with respect to \(\shiso\).
Moreover, \(\shor\) satisfies the disjunction adjunction with respect to \(\shyol\), i.e., \(\saa\shor\sab\shyol\sac\) iff \(\saa\shyol\sac\) and \(\saa\shyol\sac\) (as shown in \Cref{lem:orproperties}).
Moreover, the distributive laws for conjunction and disjunction hold in \SemiBooleanalgebra{s} (as demonstrated in \Cref{distributive}) so that a \semiBooleanalgebra is a distributive lattice with respect to \(\shyol\).
In contrast, for \semiHeytingalgebra{s} this \emph{requires reflexivity}.

\begin{propositionE}[Classical Negation][all end]\label{lem:booleanproperties}\label{lem:implyproperties}
    In \semiBooleanalgebra{}s:
    \begin{enumerate}
        \item \(\saa\shleq\sazero\shor\saa\) iff \(\saa\shleq\sab\)\label{it:orzero}
        \item \(\saa\shleq\sab\) iff \(\saa\shand\lnot\sab\shleq\sazero\)\label{it:orless}
        \item \(\saa\shiso\lnot\lnot\saa\)\label{it:booleandne}
        \item \(\saa\shleq\lnot\sab\) iff \(\sab\shleq\lnot\saa\)\label{it:negationswap}
        \item \(\saa\shleq\sab\) iff \(\lnot\sab\shleq\lnot\saa\) \label{it:negantitone}
        \item \(\lnot\sab\shyol\lnot\saa\) iff \(\saa\shyol\sab\)\label{it:yonedaneg}
        \item \(\saa\shand\lnot\sab\shleq\sac\) iff \(\saa\shand\lnot\sac\shleq\sab\)\label{it:moveing}
    \end{enumerate}
\end{propositionE}
\begin{proofE}
    \begin{caselist}
        \case{(\ref{it:orzero})}
        For the forward direction assume \(\saa\shleq\sazero\shor\sab\).
        Then by \classicality \(\saa\shleq\lnot\sazero\shimply\sab\) and by \ref{it:adjimply} then \(\saa\shand\lnot\sazero\shleq\sab\).
        Note that \(\sazero\shleq\lnot\saa\) and then by \ref{it:adjimply} then \(\sazero\shand\saa\shleq\sazero\).
        Consequently, \(\saa\shleq\lnot\sazero\) and thus \(\saa\shyol\saa\shand\saa\shyol\saa\shand\lnot\sazero\shleq\sab\).
        For the backward direction note that \(\saa\shand\lnot\sazero\shyol\saa\shleq\sab\) by \ref{it:andproj} and then \(\saa\shleq\lnot\sazero\shimply\sab\).
        From \classicality then \(\saa\shleq\sazero\shor\sab\).

        \case{(\ref{it:orzero})}
        The forward implication is immediate from \eqref{it:implev}.
        For the reverse direction suppose \(\saa\shand\lnot\sab\shleq\sazero\).
        Then by \eqref{it:adjimply} also \(\saa\shleq\lnot\sab\shimply\sazero\) and by \classicality then \(\saa\shleq\sab\shor\sazero\) and by \eqref{it:orzero} then \(\saa\shleq\sab\).

        \case{(\ref{it:booleandne})}
        For \(\saa\shyol\lnot\lnot\saa\) suppose first \(\sab\shleq\saa\).
        Then by \ref{it:implev} also \(\sab\shand\lnot\saa\shleq\sazero\) and by \ref{it:adjimply} also \(\sab\shleq\lnot\lnot\saa\).

        Now assume \(\lnot\lnot\saa\shleq\sab\) then by \ref{it:implev} then \(\lnot\lnot\saa\shand\lnot\sab\shleq\sazero\) and by \ref{it:orless} then \(\lnot\sab\shleq\lnot\saa\).
        By \ref{it:adjimply} then \(\lnot\sab\shand\saa\shleq\sazero\) and by \ref{it:orless} then \(\saa\shleq\sab\).

        Next for \(\lnot\lnot\saa\shyol\saa\) suppose first \(\sab\shleq\lnot\lnot\saa=\lnot\saa\shimply\sazero\) and by \classicality then \(\sab\shleq\saa\shor\sazero\). By \eqref{it:orzero} then \(\sab\shleq\saa\).
        Finally suppose \(\saa\shleq\sab\).
        Then by \ref{it:implev} then \(\saa\shand\lnot\sab\shleq\sazero\) and by \ref{it:adjimply} then \(\lnot\sab\shleq\lnot\saa\).
        Again by \ref{it:implev} then \(\lnot\lnot\saa\shand\lnot\sab\shleq\sazero\).
        By \eqref{it:orless} then \(\lnot\lnot\saa\shleq\sab\).

        \case{(\ref{it:negationswap})}
        Note the pairwise equivalences
        \[\saa\shleq\lnot\sab\iff\saa\shand\sab\shleq\sazero\iff\sab\shand\saa\shleq\sazero\iff\sab\shleq\lnot\saa\]
        where the first and the last equivalence is by \ref{it:adjimply} and the second is by \itrefof{lem:andproperties}{it:commutativeand}.
        is again by \ref{it:adjimply}.

        \case{(\ref{it:negantitone})}
        If \(\saa\shleq\sab\) then by \ref{it:implev} also  \(\saa\shand\lnot\sab\shyol\sazero\) and thus \(\saa\shand\lnot\sab\shleq\sazero\).
        By \ref{it:adjimply} then \(\lnot\sab\shleq\lnot\saa\).
        For the converse suppose \(\lnot\sab\shleq\lnot\saa\) and note by \ref{it:adjimply} then \(\lnot\sab\shand\saa\shleq\sazero\) and \(\saa\shleq\lnot\lnot\sab\).
        Then by \classicality \(\saa\shleq\sab\).

        \case{(\ref{it:yonedaneg})} Straightforward.

        \case{(\ref{it:moveing})} Assume \(\saa\shand\lnot\sab\shleq\sac\).
        By \classicality then \(\saa\shand\lnot\sab\shleq\lnot\lnot\sac\).
        Then \(\saa\shand\lnot\sab\shand\lnot\sac\shleq\sazero\) and \(\saa\shand\lnot\sac\shleq\lnot\lnot\sab\) by \ref{it:adjimply}.
        Again by classicality \(\saa\shand\lnot\sac\shleq\sab\).
    \end{caselist}
\end{proofE}

\begin{lemmaE}[Properties of \(\shor\)][all end]\label{lem:orproperties}
    In \semiBooleanalgebra{}s:
    \begin{enumerate}
        \item \((\saa\shor\sab)\shor\sac\shiso\saa\shor(\sab\shor\sac)\)\label{it:associativeor}
        \item \(\saa\shor\sab\shiso\sab\shor\saa\)\label{it:commutativeor}
        \item \(\saa\shor\sab\shyol\saa'\shor\sab'\) for all \(\saa\shyol\saa'\) and \(\sab\shyol\sab'\) \label{it:yonedaormonotonicity}
        \item \(\saa\shor\sab\shyol\sac\) iff \(\saa\shyol\sac\) and \(\sab\shyol\sac\)\label{it:yonedaoradj}
        \item \(\sazero\shor\saa\shiso\saa\shiso\saa\shor\sazero\)\label{it:orzeroneutral}
        \item \(\saa\shor\saone\shiso\saone\shiso\saone\shor\saa\)\label{it:orone}
    \end{enumerate}
\end{lemmaE}
\begin{proofE}
    First note that \(\saa\shor\sab\shyol\saa\shor\sab'\) whenever \(\sab\shyol\sab'\).
    To see this consider first \(\saa\shor\sab'\shleq\sac\), so that \(\saa\shleq\sac\) and \(\sab'\shleq\sac\) by \ref{it:adjor}.
    As \(\sab\shyol\sab'\) then \(\sab\shleq\sac\) and thus again by \ref{it:adjor} with \(\saa\shleq\sac\) it follows that \(\saa\shor\sab\shleq\sac\).
    Next consider \(\sac\shleq\saa\shor\sab\).
    Then by \classicality \(\saa\shand\lnot\saa\shleq\sab\) and as \(\sab\shyol\sab'\) it follows that \(\saa\shand\lnot\saa\shleq\sab'\).
    Again by \classicality \(\saa\shleq\saa\shor\sab'\).

    \begin{caselist}
        \case{(\ref{it:associativeor})} Immediate by classicality and \yonedaassociativity of \(\shand\).

        \case{(\ref{it:commutativeor})} Note that
        \begin{align*}
            \saa\shor\sab\shyol\sab\shor(\saa\shor\sab)
            \shyol (\sab\shor\saa)\shor\sab\shyol(\sab\shor\saa)\shor(\sab\shor\saa)
            \shyol\sab\shor\saa
        \end{align*}
        where the first inequality is by \ref{it:orinj}, the second is by \yonedaassociativity \ref{it:associativeand}, the third is by the preliminary claim and \(\sab\shyol\sab\shor\saa\) by \ref{it:orinj} and the last inequality is by \ref{it:ordouble}.

        \case{(\ref{it:yonedaormonotonicity})} Immediate from the preliminary claim and commutativity~\ref{it:commutativeor}.

        \case{(\ref{it:yonedaoradj})}
        For the forward direction note \(\saa\shyol\saa\shor\sab\shyol\sac\) and \(\sab\shyol\saa\shor\sab\shyol\sac\) by \ref{it:orinj}.
        For the backward direction consider first \(\sac\shleq\sad\) and note that as \(\saa\shyol\sac\) and \(\sab\shyol\sac\) also \(\saa\shleq\sad\) and \(\sab\shleq\sad\).
        By \ref{it:adjor} thus also \(\saa\shor\sab\shleq\sad\).
        Next consider \(\sad\shleq\saa\shor\sab\) so that \(\sad\shand\lnot\saa\shleq\sab\) by \classicality.
        It follows that \(\sad\shand\lnot\saa\shleq\sac\) as \(\sab\shyol\sac\) and again \(\sad\shleq\saa\shor\sac\).
        By commutativity (\ref{it:commutativeor}) \(\sad\shleq\sac\shor\saa\) and thus \(\sad\shand\lnot\sac\shleq\saa\).
        Then \(\sad\shand\lnot\sac\shleq\sac\) since \(\saa\shyol\sac\) and again \(\sad\shleq\sac\shor\sac\).
        Conclude \(\sad\shleq\sac\) with \ref{it:ordouble}.

        \case{(\ref{it:orzeroneutral})}
        Note that \(\saa\shor\sazero\shyol\saa\) by \ref{it:orinj}.
        Conversely by the preliminary claim and \(\sazero\shyol\saa\) it follows that \(\saa\shor\sazero\shyol\saa\shor\saa\shyol\saa\) by \ref{it:ordouble}.
        The other equivalence follows from commutativity~\ref{it:commutativeor}.

        \case{\ref{it:orone}} Note \(\saone\shyol\saa\shor\saone\) by \irref{it:orinj} and \(\saa\shor\saone\shyol\saone\) by \irref{it:zeroone}.
        The other equivalence follows from commutativity~\ref{it:commutativeor}.
    \end{caselist}
\end{proofE}

\begin{propositionE}[Implication Adjunction][all end]\label{lem:impladjyon}
    In a \SemiBooleanalgebra \(\saa'\shimply\sab\shyol\saa\shimply\sab'\) if \(\saa\shyol\saa'\) and \(\sab\shyol\sab'\)
\end{propositionE}

\begin{proofE}
    Note that \(\lnot\saa'\shyol\lnot\saa\shyol\lnot\saa\shor\sab\) by \itrefof{lem:implyproperties}{it:yonedaneg} and \ref{it:orinj}.
    By \ref{it:adjor} also \(\sab\shyol\lnot\saa\shor\sab\).
    By \itrefof{lem:orproperties}{it:yonedaoradj} then \(\lnot\saa'\shor\lnot\sab\shyol\lnot\saa\shor\sab'\) and the claim follows from \classicality.
\end{proofE}

\textEnd{The usual distributivity and De Morgan laws are consequences of the characterization of implication as material implication in a \semiHeytingalgebra{}.}

\begin{propositionE}[Distributivity][all end]\label{distributive}
    In \semiBooleanalgebra{}s:
    \begin{enumerate}
        \item \(\saa\shand(\sab\shor\sac)\shiso(\saa\shand\sab)\shor(\saa\shand\sac)\)\label{it:anddistror}
        \item \(\saa\shor(\sab\shand\sac)\shiso(\saa\shor\sab)\shand(\saa\shor\sac)\)\label{it:ordistrand}
        \item \(\lnot(\saa\shand\sab)\shiso\lnot\saa\shor\lnot\sab\)\label{it:demorganliteral}
        \item \(\lnot(\saa\shor\sab)\shiso\lnot\saa\shand\lnot\sab\)\label{it:demorganliteraltwo}
        \item \((\saa\shand\sab)\shimply\sac\shiso(\saa\shimply\sac)\shor(\sab\shimply\sac)\)\label{it:demorganand}
        \item \((\saa\shor\sab)\shimply\sac\shiso(\saa\shimply\sac)\shand(\sab\shimply\sac)\)\label{it:demorganor}
    \end{enumerate}
\end{propositionE}

\begin{proofE}
    \begin{caselist}
        \case{(\ref{it:anddistror})}
        By \ref{it:andproj} note \(\saa\shand\sab\shyol\sab\) and \(\saa\shand\sac\shyol\sac\) with \itrefof{lem:orproperties}{it:yonedaoradj} yield \(\saa\shand\sab\shyol \sab\shor\sac\) and \(\saa\shand\sac\shyol \sab\shor\sac\).
        Again by \ref{it:andproj} then \(\saa\shand\sab\shyol\saa\) and \(\saa\shand\sac\shyol\saa\), so that by \itrefof{lem:andproperties}{it:andyonadj} also \(\saa\shand\sab\shyol\saa\shand (\sab\shor\sac)\) and \(\saa\shand\sac\shyol\saa\shand (\sab\shor\sac)\).
        Finally \((\saa\shand\sab)\shor(\saa\shand\sac)\shyol\saa\shand (\sab\shor\sac)\) follows with \itrefof{lem:orproperties}{it:yonedaoradj}.

        For the converse consider first \((\saa\shand\sab)\shor(\saa\shand\sac)\shleq\sad\) then by \ref{it:adjor} note \(\saa\shand\sab\shleq\sad\) and \(\saa\shand\sac\shleq\sad\).
        From \itrefof{lem:andproperties}{it:commutativeand} it follows that \(\sab\shand\saa\shleq\sad\) and \(\sac\shand\saa\shleq\sad\).
        With \ref{it:adjimply} then \(\sab\shleq\saa\shimply\sad\) and \(\sac\shleq\saa\shimply\sad\).
        By \ref{it:adjor} then \(\sab\shor\sac\shleq\saa\shimply\sad\), so that by \ref{it:adjimply} also \((\sab\shor\sac)\shand\saa\shleq\sad\).
        Again with \itrefof{lem:andproperties}{it:commutativeand} then \(\saa\shand (\sab\shor\sac)\shleq\sad\).

        Next consider \(\sad\shleq\saa\shand(\sab\shor\sac)\) then by \ref{it:adjand} also \(\sad\shleq\saa\) and \(\sad\shleq\sab\shor\sac\).
        By \classicality then \(\sad\shand\lnot\sab\shleq\sac\) and by \Cref{it:andproj} also \(\sad\shand\lnot\sab\shyol\sad\shleq\saa\).
        So by \itrefof{lem:andproperties}{it:andyonadj} then \(\sad\shand\lnot\sab\shleq\saa\shand\sac\).
        Applying \classicality again twice \(\sad\shand\lnot(\saa\shand\sac)\shleq\sab\) and \(\sad\shand\lnot(\saa\shand\sac)\shyol\sad\shleq\saa\) by \ref{it:andproj}.
        Hence \(\sad\shand\lnot(\saa\shand\sac)\shleq\saa\shand\sab\) by \ref{it:adjand}.
        Again by \classicality then \(\sad\shleq(\saa\shand\sac)\shor\saa\shand\sab\) and the equivalence follows with \itrefof{lem:andproperties}{it:commutativeand}.

        \case{(\ref{it:ordistrand})}
        First observe that as \(\sab\shand\sac\shyol\sab\) and \(\sab\shand\sac\shyol\sab\) by \ref{it:orinj}, by \itrefof{lem:orproperties}{it:ormonotonicity} also \(\saa\shor(\sab\shand\sac)\shyol\saa\shor\sab\) and  \(\saa\shor(\sab\shand\sac)\shyol\saa\shor\sac\).
        So by \itrefof{lem:andproperties}{it:yonedaoradj} also \(\saa\shor(\sab\shand\sac)\shyol(\saa\shor\sab)\shand(\saa\shor\sac)\).

        For the converse consider first \(\saa\shor(\sab\shand\sac)\shleq\sad\).
        Then \(\saa\shleq\sad\) and \(\sab\shand\sac\shleq\sad\) by \ref{it:adjor}.
        Then \(\sab\shleq\sac\shimply\sad\) by \ref{it:adjimply}.
        Since \(\saa\shand\sac\shyol\saa\shleq\sad\) by \ref{it:andproj} then also \(\saa\shleq\sac\shimply\sad\) by \ref{it:adjimply}.
        Hence by \ref{it:adjor} also \(\saa\shor\sab\shleq\sac\shimply\sad\).
        Using \ref{it:adjimply} twice with \itrefof{lem:andproperties}{it:commutativeand} it follows that \(\sac\shleq(\saa\shor\sab)\shimply\sad\).
        Since again \(\saa\shand(\saa\shor\sab)\shyol\saa\shleq\sad\) by \ref{it:andproj} then also \(\saa\shleq(\saa\shor\sab)\shimply\sad\) by \ref{it:adjimply}.
        Note by \ref{it:adjor} then \(\saa\shor\sac\shleq(\saa\shor\sab)\shimply\sad\).
        From \ref{it:adjimply} and \itrefof{lem:andproperties}{it:commutativeand} then \((\saa\shor\sab)\shand(\saa\shor\sac)\shyol\saa\shor(\sab\shand\sac)\)

        Next consider \(\sad\shleq(\saa\shor\sab)\shand(\saa\shor\sac)\), so that \(\sad\shleq\saa\shor\sab\) and \(\sad\shleq\saa\shor\sac\) by \ref{it:adjand}.
        By \classicality then \(\sad\shand\lnot\saa\shleq\sab\) and \(\sad\shand\lnot\saa\shleq\sac\).
        By \ref{it:adjor} then \(\sad\shand\lnot\saa\shleq\sab\shand\sac\).
        Again by \classicality \(\sad\shleq\saa\shor (\sab\shand\sac)\).

        \case{(\ref{it:demorganliteral})}
        Note the following equivalences
            {\renewcommand{\iff}{\text{ iff }}
                \begin{align*}
                    \sac\shleq\lnot(\saa\shand\sab)
                    \iff
                    \sac\shand\saa\shand\sab\shleq\sazero
                    \iff
                    \sac\shand\saa\shleq\lnot\sab
                    \iff
                    \sac\shleq\saa\shimply\lnot\sab
                    \iff
                    \sac\shleq\lnot\saa\shor\lnot\sab
                \end{align*}
                and
                \begin{align*}
                    \lnot\saa\shor\lnot\sab\shleq\sac
                    \iff
                    \lnot\saa\shleq\sac \text{ \& } \lnot\sab\shleq\sac
                    \iff
                    \lnot\sac\shleq\saa \text{ \& } \lnot\sac\shleq\sab
                    \iff
                    \lnot\sac\shleq\saa \shand\sab
                    \iff
                    \lnot(\saa \shand\sab)\shleq \sac
                \end{align*}

                \case{(\ref{it:demorganliteraltwo})}
                Note the following equivalences
                \begin{align*}
                    \sac\shleq\lnot(\saa\shor\sab)
                     & \iff
                    \sac\shand(\saa\shor\sab)\shleq\sazero
                    \iff
                    (\sac\shand\saa)\shor(\sac\shand\sab)\shleq\sazero
                    \\
                     &
                    \iff
                    \sac\shand\saa\shleq\sazero \text{ and } \sac\shand\sab\shleq\sazero
                    \iff
                    \sac\shleq\lnot\saa \text{ and } \sac\shleq\lnot\sab
                    \iff
                    \sac\shleq\lnot\saa\shand\lnot\sab
                \end{align*}
                and
                \begin{align*}
                    \lnot\saa\shor\lnot\sab\shleq\sac
                    \iff
                    \lnot\saa\shleq\sab\shor\sac
                    \iff
                    \saone\shleq\saa\shor\sab\shor\sac
                    \iff
                    \lnot(\saa\shor\sab)\shleq\sac
                \end{align*}
            }
        \case{(\ref{it:demorganand}) and (\ref{it:demorganor})} Immediate from \ref{it:demorganliteral} with classicality.
    \end{caselist}
\end{proofE}

\begin{example}\label{ex:kleene}
    An important example of a \emph{\classical} \semiHeytingalgebra{} is Kleene Logic~\(\kleenethree\) consisting of the three elements \(0,\midpoint,1\), and ordered by \(<\).
    The operations \(\shand\) and \(\shor\) are defined as the maximum and the minimum, respectively. Implication is defined by
    \[
        \saa\shimply\sab = (1-\saa)\shor\sab =\begin{cases}
            1                  & \text{if } \saa = 0       \\
            \sab\shor\midpoint & \text{if } \saa=\midpoint \\
            \sab               & \text{if } \saa=1
        \end{cases}
        \qquad \text{or } \lnot\saa = 1-\saa \text{ and }\saa\limply\sab=\lnot\saa\shor\sab.
    \]
    This is a minimal example of a \semiBooleanalgebra{} that is \emph{not} a Boolean algebra and plays a crucial `dualizing' role in the completeness proof in \Cref{sec:denotational}.
\end{example}

\begin{example}\label{ex:pairsemiba}
    For a set \(\setdoma\) the structure \(\settopowob[\setdoma]=\{(\selop,\selcl):\selop\subseteq\selcl\subseteq\setdoma\}\) with connectives
    \begin{align*}
         &
        \topzero=(\emptyset,\emptyset)
         &   &
        \sela \topand \selb=(\selopa\intersection\selopb,\selcla\intersection\selclb)
         &   &
        \sela \topor \selb=(\selopa\union\selopb,\selcla\union\selclb)
        \\
         &
        \topone=(\tops,\tops)
         &   &
        \sela\tople\selb\iff\selcla\subseteq\selopb
         &   &
        \topimply{\sela}{\selb}=(\setcomp{\selcla}\union\selopb,\setcomp{\selopa}\union\selclb).
    \end{align*}
    for \(\sel_i = ({\selop}_i,{\selcl}_i)\), is the powerset \semiBooleanalgebra{}
    The powerset \semiBooleanalgebra is isomorphic to the \SemiBooleanalgebra of \emph{partial functions} \(f:X\to\{\mathrm{True}, \mathrm{False}\}\) via the map \(f\mapsto (\{x: f(x)=\mathrm{True}\},\{x: f(x)\neq\mathrm{False}\})\).
\end{example}

A particularly interesting and promising observation is that Belnap's influential four-valued logic (for reasoning epistemically about possibly contradictory information) \cite[\S81]{belnap-entailment} also describes a \semiBooleanalgebra but not a Boolean algebra.
This evidences that \semiBooleanalgebra{s} are relevant and applicable algebraic structures and (by soundness) that \subrefllogic is a useful generalization of classical logic to previously studied domains.

In \semiBooleanalgebra{s}, unlike \semiHeytingalgebra{s}, \(\shiso\) is a congruence relation with respect to the operations (as shown in \Cref{lem:impladjyon}).
Consequently, \(\shiso\) can be taken to be literal equality.
A \SemiBooleanalgebra is \emph{antisymmetric} if \(\saa=\sab\) whenever \(\saa\shiso\sab\).

\begin{theoremE}[][]
    Every \semiBooleanalgebra is reflected in an antisymmetric one.
\end{theoremE}

\begin{proofE}
    Let \(\sbs'\) be the quotient of \(\sbs\) by the equivalence relation~\(\shiso\) and let \(\pi\) be the projection onto the quotient, which assigns every element \(\saa\) its \(\shiso\)-equivalence class.
    A \semiHeytingalgebra{} structure on \(\sbs'\) can be defined by
    \begin{align*}
         &
        \sazero=\pi(\sazero)
         &   &
        \sazero=\pi(\saone)
        \\
         &
        \pi(\saa)\shleq\pi(\sab)\Leftrightarrow\saa\shleq\sab
         &   &
        \pi(\saa)\shimply\pi(\sab)=\pi(\saa\shimply\sab)
        \\
         &
        \pi(\saa)\shand\pi(\sab)=\pi(\saa\shand\sab)
         &   &
        \pi(\saa)\shor\pi(\sab)=\pi(\saa\shor\sab)
    \end{align*}
    The monotonicity of  \itrefof{lem:andproperties}{it:yonedaandmonotonicity}, \itrefof{lem:orproperties}{it:yonedaormonotonicity} and \Cref{lem:impladjyon} ensures that this is well-defined.
    Moreover, it is clearly an order-reflecting homomorphism of \semiHeytingalgebra{}s.
\end{proofE}
\SemiBooleanalgebra{s} generalize Boolean algebras up to reflexivity:

\begin{theoremE}
    An antisymmetric \semiBooleanalgebra is a Boolean algebra iff \(\shleq\) is reflexive.
\end{theoremE}

\begin{proofE}
    By \Cref{prop:reflexivesemiheyting} it suffices to show that Boolean algebras are exactly the Heyting algebras satisfying \(\saa\limply\sab=\lnot\saa\shor\sab\) and \(\lnot\saa\limply\sab=\saa\shor\sab\).
    Any Heyting algebra satisfying this also clearly satisfies \(\lnot\lnot\saa=\lnot\saa\shimply\sazero=\saa\shor\sazero=\saa\) and therefore is a Boolean algebra.
    Conversely, the equalities hold in any Boolean algebra.
\end{proofE}

\subsection{Semialgebraic Models}\label{sec:semisemantics}

\semiHeytingalgebra{}s provide concrete semantics for \subrefllogic.
Formulas of propositional logic are interpreted with respect to \emph{\semiHeytingalgebra{} models} \((\shs,\heyint)\) which consist of a \semiHeytingalgebra{} \(\shs\) with a valuation \(\heyint:\patoms\to\shs\) for atomic propositions.

\begin{definition}[Semantics]
    For every formula \(\fml\) of propositional logic, the \emph{semantics} of \(\fml\) in the \semiHeytingalgebra{} model \((\shs,\heyint)\) is defined inductively:
    \begin{align*}
         &
        \sasem{\heyint}{\ftrue} = \saone
         &   &
        \sasem{\heyint}{\patom} = \heyintof{\patom}
         &   &
        \sasem{\heyint}{\fml\land\fmlb} = \sasem{\heyint}{\fml} \shand \sasem{\heyint}{\fmlb}
        \\
         &
        \sasem{\heyint}{\ffalse} = \sazero
         &   &
        \sasem{\heyint}{\fml\limply\fmlb} = \sasem{\heyint}{\fml} \shimply \sasem{\heyint}{\fmlb}
         &   &
        \sasem{\heyint}{\fml\lor\fmlb} = \sasem{\heyint}{\fml} \shor \sasem{\heyint}{\fmlb}.
    \end{align*}
\end{definition}

For a sequence of formulas \(\fmlset\), write \(\sasem{\heyint}{\fmlset}\) for \(\bigshand_{\fmlb\in\fmlset}\sasem{\heyint}{\fmlb}\).
(This is well-defined by \yonedaassociativity \ref{it:associativeand}.)
As usual, the conjunction over the empty set is \(\saone\).
A sequent \(\lsequent{\fmlset}{\fml}\) is said to be \emph{true} in a \semiHeytingalgebra model \((\shs,\heyint)\) if \(\sasem{\heyint}{\fmlset}\shleq\sasem{\heyint}{\fml}\).
In this case, write \(\satrue{\heyint}{\lsequent{\fmlset}{\fml}}\) and if \(\fmlset=\emptyset\) write \(\satrue{\heyint}{\fml}\).
A formula \(\fml\) is a \emph{\semiHeytingalgebra{} consequence} \(\saconseq{\relfmlset}{\fml}\) (\emph{semiboolean-semantic consequence} \(\saconseqbool{\relfmlset}{\fml}\)) of a set of sequents \(\relfmlset\) iff \(\satrue{\heyint}{\fml}\) for all \semiHeytingalgebra{} (\semiBooleanalgebra) models \((\shs,\heyint)\) with \(\satrue{\heyint}{\relfmlset}\).

\textEnd{
    To ensure that \semiHeytingalgebra{}s properly capture the theory of \subrefllogic, it is important that the definition of a \semiHeytingalgebra{} is restrictive enough, to ensure that it captures no more than \subrefllogic, which is guaranteed by soundness:}

\begin{theoremE}[Algebraic Soundness][all end]\label{prop:semiheytingsound}
    The calculus of \subrefllogic is sound for \semiHeytingalgebra{} semantics: if \(\subreflprov{\relfmlset}{\fml}\) then \(\saconseq{\relfmlset}{\fml}\).
\end{theoremE}
\begin{proofE}{
        It is shown that in any \semiHeytingalgebra{}, the conclusion of a rule is true, if all its assumptions are true.
        For readability the model annotation is dropped and instead of \(\sasem{\heyint}{\fml}\) we write \(\sasemsimple{\heyint}{\fml}\).

        \begin{caselist}
            \case{\irref{andL}}
            Let \(\saa=\sasemsimple{\heyint}{\fml_1}\shand(\sasemsimple{\heyint}{\fml_2}\shand\sasemsimple{\heyint}{\fmlset})\)
            Observe
            \[\sasemsimple{\heyint}{\fml_1\land\fml_2}\shand\sasemsimple{\heyint}{\fmlset}\shiso\saa\shyol\saa\shand\saa\shyol
                \sasemsimple{\heyint}{\fml_1}\shand\saa
                \shiso\sasemsimple{\heyint}{\fml_1}\shand\sasemsimple{\heyint}{\fml_1\land\fml_2}\shand\sasemsimple{\heyint}{\fmlset}
            \]
            by \ref{it:associativeand}, \ref{it:anddouble} and property \itrefof{lem:andproperties}{it:yonedaandmonotonicity} with \(\saa\shyol\sasemsimple{\heyint}{\fml_1}\) by \ref{it:andproj}.
            Similarly, using additionally the commutativity property \itrefof{lem:andproperties}{it:commutativeand}.

            \case{\irref{andR}} Immediate from \ref{it:adjand}.

            \case{\irref{orL}}  Assume \(\sasemsimple{\heyint}{\fml}\shand\sasemsimple{\heyint}{\fmlset} \shleq\sasemsimple{\heyint}{\fmlc}\) and \(\sasemsimple{\heyint}{\fmlb}\shand\sasemsimple{\heyint}{\fmlset} \shleq\sasemsimple{\heyint}{\fmlc}\).
            By \ref{it:adjimply} it follows that \(\sasemsimple{\heyint}{\fml}\shleq\sasemsimple{\heyint}{\fmlset}\shimply\sasemsimple{\heyint}{\fmlc}\) and \(\sasemsimple{\heyint}{\fmlb}\shleq\sasemsimple{\heyint}{\fmlset}\shimply\sasemsimple{\heyint}{\fmlc}\).
            Hence by \ref{it:adjor} then \(\sasemsimple{\heyint}{\fml}\shor\sasemsimple{\heyint}{\fmlb}\shleq\sasemsimple{\heyint}{\fmlset}\shimply\sasemsimple{\heyint}{\fmlc}\).
            By definition of the semantics of \(\lor\) and \ref{it:adjimply} this yields \(\sasemsimple{\heyint}{\fml\lor\fmlb}\shand\sasemsimple{\heyint}{\fmlset}\shleq\sasemsimple{\heyint}{\fmlc}\).

            \case{\irref{orR}} Immediate from \ref{it:orinj}.

            \case{\irref{implyL}} Let \(\saa=(\sasemsimple{\heyint}{\fml}\shimply\sasemsimple{\heyint}{\fmlb})\shand\sasemsimple{\heyint}{\fmlset}\) and assume \(\saa\shleq\sasemsimple{\heyint}{\fml}\) and \(\sasemsimple{\heyint}{\fmlb}\shand\sasemsimple{\heyint}{\fmlset}\shleq\sasemsimple{\heyint}{\fmlc}\).
            Observe that by \ref{it:adjimply} also \(\sasemsimple{\heyint}{\fmlb}\shleq\sasemsimple{\heyint}{\fmlset}\shimply\sasemsimple{\heyint}{\fmlc}\).
            Note next that
            \[\saa\shyol\saa\shand\saa\shyol(\sasemsimple{\heyint}{\fml}\shimply\sasemsimple{\heyint}{\fmlb})\shand\saa\shyol\sasemsimple{\heyint}{\fmlb}\]
            where the first inequality is by \ref{it:anddouble} and the second follows from \ref{it:andproj} together with monotonicity property \itrefof{lem:andproperties}{it:yonedaandmonotonicity} and the last equality is by \ref{it:implev} using \(\saa\shleq\sasemsimple{\heyint}{\fml}\).
            From this inequality it follows that \(\saa\shleq\sasemsimple{\heyint}{\fmlset}\shimply\sasemsimple{\heyint}{\fmlc}\)
            and from \ref{it:adjimply} then \(\saa\shand\sasemsimple{\heyint}{\fmlset}\shleq\sasemsimple{\heyint}{\fmlc}\).
            Then \(\saa\shleq\sasemsimple{\heyint}{\fmlc}\) follows as
            \(\saa\shyol\saa\shand\saa\shyol\saa\shand \sasemsimple{\heyint}{\fmlset}\)
            again from \ref{it:anddouble}, \ref{it:andproj} and \itrefof{lem:andproperties}{it:yonedaandmonotonicity}.

            \case{\irref{implyR}} Immediate by \ref{it:adjimply}.
        \end{caselist}

        Next the \(\relfmlset\)-assumption axiom and rules are considered.

        \begin{caselist}
            \case{\irref{assumption}} Immediate.

            \case{\irref{assumptionL}} Note that by the side condition assumption \(\sasemsimple{\heyint}{\fmlset}\shleq\sasemsimple{\heyint}{\fml}\) and thus \(\sasemsimple{\heyint}{\fmlset}\shyol\sasemsimple{\heyint}{\fml}\).
            Hence \(\sasemsimple{\heyint}{\fmlset}\shyol\sasemsimple{\heyint}{\fmlset}\shand \sasemsimple{\heyint}{\fmlset}\shyol \sasemsimple{\heyint}{\fml} \shand \sasemsimple{\heyint}{\fmlset}\) by \ref{it:anddouble} and \itrefof{lem:andproperties}{it:yonedaandmonotonicity}.
            So whenever \(\sasemsimple{\heyint}{\fml} \shand \sasemsimple{\heyint}{\fmlset}\shleq\sasemsimple{\heyint}{\fmlb}\) then \(\sasemsimple{\heyint}{\fmlset} \shleq \sasemsimple{\heyint}{\fmlb}\).

            \case{\irref{assumptionR}} Assume \(\sasemsimple{\heyint}{\fmlset}\shleq\sasemsimple{\heyint}{\fml}\) and thus \(\sasemsimple{\heyint}{\fmlset}\shyol\sasemsimple{\heyint}{\fml}\).
            Hence \(\sasemsimple{\heyint}{\fmlset}\shyol\sasemsimple{\heyint}{\fmlset}\shand \sasemsimple{\heyint}{\fmlset}\shyol \sasemsimple{\heyint}{\fml} \shand \sasemsimple{\heyint}{\fmlset}\) by \ref{it:anddouble}  and \itrefof{lem:andproperties}{it:yonedaandmonotonicity}.
            By the side condition \(\sasemsimple{\heyint}{\fml} \shand \sasemsimple{\heyint}{\fmlset}\shleq\sasemsimple{\heyint}{\fmlb}\), so that \(\sasemsimple{\heyint}{\fmlset} \shleq \sasemsimple{\heyint}{\fmlb}\).
        \end{caselist}
    }
\end{proofE}

\textEnd{
    Note that admissibility of the rules \irref{weak}, \irref{contraction}, \irref{cut} and the inverse rule, does not mean that the rules are sound in the sense that the truth of a premise entails the truth of the conclusion.
    Admissibility in a sound calculus itself only ensures that the conclusion is true, whenever the premises are \emph{provable}.
    However, the rules are still sound.}

\begin{propositionE}[Structural  Soundness][all end]\label{prop:structuralsoundness}
    The structural rules \irref{weak}, \irref{contraction}, \irref{cut} and the inverse rules are sound, i.e. in every \semiHeytingalgebra{} model, in which the premises are true, so is the conclusion.

    The same rules and the inverse classical rules \irref{orRclassicalminus}, \irref{implyLclassicalminusone} and \irref{implyLclassicalminustwo} are sound for \semiBooleanalgebra{s}.
\end{propositionE}

\begin{proofE}{Again proceed with a fixed \semiHeytingalgebra{} model \((\shs,\heyint)\) and drop the model annotation.

        \begin{caselist}
            \case{\irref{weak}} Immediate from \ref{it:andproj}.

            \case{\irref{contraction}} Assume \(\sasemsimple{\heyint}{\fml}\shand\sasemsimple{\heyint}{\fml}\shand\sasemsimple{\heyint}{\fmlset}\shleq\sasemsimple{\heyint}{\fmlc}\).
            Then by \ref{it:adjimply} it follows that \(\sasemsimple{\heyint}{\fml}\shand\sasemsimple{\heyint}{\fml}\shleq\sasemsimple{\heyint}{\fmlset}\shimply\sasemsimple{\heyint}{\fmlc}\) and then by \ref{it:anddouble} also \(\sasemsimple{\heyint}{\fml}\shleq\sasemsimple{\heyint}{\fmlset}\shimply\sasemsimple{\heyint}{\fmlc}\).
            Again by \ref{it:adjimply} then \(\sasemsimple{\heyint}{\fml}\shand\sasemsimple{\heyint}{\fmlset}\shleq\sasemsimple{\heyint}{\fmlc}\).

            \case{\irref{cut}}
            Assume \(\sasemsimple{\heyint}{\fmlset}\shleq\sasemsimple{\heyint}{\fml}\) and \(\sasemsimple{\heyint}{\fmlset}\shand\sasemsimple{\heyint}{\fml}\shleq\sasemsimple{\heyint}{\fmlc}\).
            Note that \(\sasemsimple{\heyint}{\fmlset}\shyol\sasemsimple{\heyint}{\fmlset}\shand\sasemsimple{\heyint}{\fmlset}\shyol\sasemsimple{\heyint}{\fml}\shand\sasemsimple{\heyint}{\fmlset}\) by \ref{it:anddouble} and \itrefof{lem:andproperties}{it:yonedaandmonotonicity}.
            Hence \(\sasemsimple{\heyint}{\fmlset}\shleq\sasemsimple{\heyint}{\fml}\)

            \case{\irref{andLminus}}
            Assume \(\sasemsimple{\heyint}{\fml_1\land\fml_2}\shand \sasemsimple{\heyint}{\fmlset}\shleq\sasemsimple{\heyint}{\fmlc}\).
            By definition of the semantics and \irref{it:adjimply} also \(\sasemsimple{\heyint}{\fml_1}\shand\sasemsimple{\heyint}{\fml_2}\shleq\sasemsimple{\heyint}{\fmlset}\shimply\sasemsimple{\heyint}{\fmlc}\).
            As \(\sasemsimple{\heyint}{\fml_i}\shyol\sasemsimple{\heyint}{\fml_1}\land\sasemsimple{\heyint}{\fml_2}\) by \ref{it:andproj} it follows that \(\sasemsimple{\heyint}{\fml_1}\shleq\sasemsimple{\heyint}{\fmlset}\shimply\sasemsimple{\heyint}{\fmlc}\).
            By \irref{it:adjand} then \(\sasemsimple{\heyint}{\fml_1}\shand(\sasemsimple{\heyint}{\fml_1}\shand\sasemsimple{\heyint}{\fml_2})\shleq\sasemsimple{\heyint}{\fmlset}\shimply\sasemsimple{\heyint}{\fmlc}\) and by \irref{it:adjimply} then \(\sasemsimple{\heyint}{\fml_1}\shand(\sasemsimple{\heyint}{\fml_1}\shand\sasemsimple{\heyint}{\fml_2})\shand\sasemsimple{\heyint}{\fmlset}\shleq\sasemsimple{\heyint}{\fmlc}\).
            Immediate from the definition of the semantics of \(\land\) (and associativity \ref{it:associativeand}).

            \case{\irref{andRminus}} Immediate from \ref{it:andproj}.

            \case{\irref{orLminus}} Assume \((\sasemsimple{\heyint}{\fml_1}\shor\sasemsimple{\heyint}{\fml_2})\shand\sasemsimple{\heyint}{\fmlset}\shleq\sasemsimple{\heyint}{\fmlc}\).
            Then by \irref{it:adjimply} also \(\sasemsimple{\heyint}{\fml_1}\shor\sasemsimple{\heyint}{\fml_2}\shleq\sasemsimple{\heyint}{\fmlset}\shimply\sasemsimple{\heyint}{\fmlc}\) and by \irref{it:adjor} then \(\sasemsimple{\heyint}{\fml_i}\shleq\sasemsimple{\heyint}{\fmlset}\shimply\sasemsimple{\heyint}{\fmlc}\).
            By \irref{it:adjimply} it follows that \(\sasemsimple{\heyint}{\fml_i}\shand\sasemsimple{\heyint}{\fmlset}\shleq\sasemsimple{\heyint}{\fmlc}\).

            \case{\irref{orRminus}} Follows immediately from \ref{it:ordouble}.

            \case{\irref{implyRminus}} Immediate from \irref{it:adjimply}.
        \end{caselist}}

    Soundness of the classical inverse rules is similar.
\end{proofE}

\begin{propositionE}[\Classical Soundness][all end]\label{prop:classicalsemiheytingsound}
    \Classicalsubrefl is sound for \semiBooleanalgebra models, i.e. if \(\classsubreflprov{\relfmlset}{\fml}\) then \(\saconseqbool{\relfmlset}{\fml}\).
\end{propositionE}

\begin{proofE}
    Again it is shown that in any \semiBooleanalgebra{}, the conclusions of the rules \irref{implyLclassical} and \irref{orRclassical} is true, if the respective assumptions are true.
    For readability the model annotation is dropped again.

    \begin{caselist}
        \case{\irref{orRclassical}}
        Suppose \(\sasemsimple{\heyint}{\fmlset}\shand \lnot\sasemsimple{\heyint}{\fml_j}\shleq\sasemsimple{\heyint}{\fml_i}\), then by \ref{it:adjimply} \(\sasemsimple{\heyint}{\fmlset}\shleq \lnot\sasemsimple{\heyint}{\fml_j}\shimply\sasemsimple{\heyint}{\fml_i}\) and by \classicality then \(\sasemsimple{\heyint}{\fmlset}\shleq \sasemsimple{\heyint}{\fml_j}\shor\sasemsimple{\heyint}{\fml_i}\).

        \case{\irref{implyLclassical}}
        Let \(\saa=\lnot\sasemsimple{\heyint}{\fmlc}\shand(\sasemsimple{\heyint}{\fml}\shimply\sasemsimple{\heyint}{\fmlb})\shand\sasemsimple{\heyint}{\fmlset}\) and assume \(\saa\shleq\sasemsimple{\heyint}{\fml}\) and \(\sasemsimple{\heyint}{\fmlb}\shand\sasemsimple{\heyint}{\fmlset}\shleq\sasemsimple{\heyint}{\fmlc}\).
        Observe that by \ref{it:adjimply} also \(\sasemsimple{\heyint}{\fmlb}\shleq\sasemsimple{\heyint}{\fmlset}\shimply\sasemsimple{\heyint}{\fmlc}\).
        Note next that
        \[\saa\shyol\saa\shand\saa\shyol(\sasemsimple{\heyint}{\fml}\shimply\sasemsimple{\heyint}{\fmlb})\shand\saa\shyol\sasemsimple{\heyint}{\fmlb}\]
        where the first inequality is by \ref{it:anddouble} and the second follows from \ref{it:andproj} together with monotonicity property \itrefof{lem:andproperties}{it:yonedaandmonotonicity} and the last equality is by \ref{it:implev} using \(\saa\shleq\sasemsimple{\heyint}{\fml}\).
        From this inequality it follows that \(\saa\shleq\sasemsimple{\heyint}{\fmlset}\shimply\sasemsimple{\heyint}{\fmlc}\)
        and from \ref{it:adjimply} then \(\saa\shand\sasemsimple{\heyint}{\fmlset}\shleq\sasemsimple{\heyint}{\fmlc}\).
        Then \(\saa\shleq\sasemsimple{\heyint}{\fmlc}\) follows as
        \(\saa\shyol\saa\shand\saa\shyol\saa\shand \sasemsimple{\heyint}{\fmlset}\)
        again from \ref{it:anddouble}, \ref{it:andproj} and \itrefof{lem:andproperties}{it:yonedaandmonotonicity}.
        By \ref{it:adjimply}, \classicality and \ref{it:anddouble} then \(\sasemsimple{\heyint}{\fml}\shimply\sasemsimple{\heyint}{\fmlb}\shand\sasemsimple{\heyint}{\fmlset}\shleq \lnot \sasemsimple{\heyint}{\fmlc} \shimply\sasemsimple{\heyint}{\fmlc} \shiso \sasemsimple{\heyint}{\fmlc}\shor \sasemsimple{\heyint}{\fmlc}\shyol\sasemsimple{\heyint}{\fmlc}\).
    \end{caselist}
\end{proofE}

Completeness ensures that \SemiHeytingalgebra{}s and \semiBooleanalgebra{}s capture the reasoning of \subrefllogic precisely.
The importance of algebraic soundness and completeness is not the proof itself, but to confirm that  the correct algebraic counterparts for \subrefllogic are \emph{found} in \semiHeytingalgebra{s} and \semiBooleanalgebra{s}.

\textEnd{
    For a set of formulas \(\relfmlset\) the ordering \(\fml\fmlless\fmlb\) on formulas is defined by \(\subreflprov{\relfmlset}{\fml\limply\fmlb}\).
    Note that by \irref{implyR} and \irref{implyRminus} relation \(\fml\fmlless\fmlb\) holds if there is a derivation of \(\lsequent{\fml}{\fmlb}\) from \(\relfmlset\).
    Similarly define the ordering \(\fml\fmlcless\fmlb\) if the implication \(\fml\limply\fmlb\) is derivable from \(\relfmlset\) in \classicalsubrefl.}

\begin{theoremE}[][all end]\label{thm:structureissemi}
    The set of propositional formulas \(\allfmls\) with the structure \(\allfmlssha=(\allfmls,\fmlless,\ffalse,\ftrue,\land,\lor,\limply)\) is a \semiHeytingalgebra{}.

    The structure \(\allfmlssba\) is a \semiBooleanalgebra{}.
\end{theoremE}

\begin{proofE}\renewcommand{\minpwidth}{0.3\textwidth}
    {\renewcommand{\linferPremissSeparation}{\hspace{0.2cm}}
        The defining properties of \Cref{def:semiheytingalgebra} are verified.

        \begin{caselist}
            \case{\ref{it:transitive}} Immediate from \irref{cut}.

            \case{\ref{it:associativeand}}
            By the following derivations it follows that \((\fml\land\fmlb)\land\fmlc\shiso\fml\land(\fmlb\land\fmlc)\):

            \begin{minipage}[t]{0.4\textwidth}
                \begin{sequentdeduction}[array]
                    \linfer[andL+andL]
                    {
                        \linfer[andLminus]
                        {
                            \linfer[andLminus]
                            {\lsequent{(\fml\land\fmlb)\land\fmlc}{\fmld}}
                            {\lsequent{\fml\land\fmlb,\fmlc}{\fmld}}
                        }
                        {\lsequent{\fml,\fmlb,\fmlc}{\fmld}}
                    }
                    {\lsequent{\fml\land(\fmlb\land\fmlc)}{\fmld}}
                \end{sequentdeduction}
            \end{minipage}
            \qquad
            \begin{minipage}[t]{0.4\textwidth}
                \begin{sequentdeduction}[array]
                    \linfer[andL+andL]
                    {
                        \linfer[andLminus]
                        {
                            \linfer[andLminus]
                            {\lsequent{\fml\land(\fmlb\land\fmlc)}{\fmld}}
                            {\lsequent{\fml,\fmlb\land\fmlc}{\fmld}}
                        }
                        {\lsequent{\fml,\fmlb,\fmlc}{\fmld}}
                    }
                    {\lsequent{(\fml\land\fmlb)\land\fmlc}{\fmld}}
                \end{sequentdeduction}
            \end{minipage}

            \bigskip

                \begin{sequentdeduction}[array]
                    \linfer[andR]
                    {
                        \linfer[andRminus]
                        {\linfer[andRminus]
                            {\lsequent{\fmld}{(\fml\land\fmlb)\land\fmlc}}
                            {\lsequent{\fmld}{\fml\land\fmlb}}}
                        {\lsequent{\fmld}{\fml}}
                        !
                        \linfer[andR]
                        {
                            \linfer[andRminus]
                            {\lsequent{\fmld}{(\fml\land\fmlb)\land\fmlc}}
                            {
                                \linfer[andRminus]
                                {\lsequent{\fmld}{\fml\land\fmlb}}
                                {\lsequent{\fmld}{\fmlb}}
                            }
                            !
                            \linfer[andRminus]
                            {\lsequent{\fmld}{(\fml\land\fmlb)\land\fmlc}}
                            {\lsequent{\fmld}{\fmlc}}
                        }
                        {\lsequent{\fmld}{\fmlb\land\fmlc}}
                    }
                    {\lsequent{\fmld}{\fml\land(\fmlb\land\fmlc)}}
                \end{sequentdeduction}

            \bigskip

                \begin{sequentdeduction}[array]
                    \linfer[andR]
                    {
                        \linfer[andR]
                        {
                            \linfer[andRminus]
                            {\lsequent{\fmld}{\fml\land(\fmlb\land\fmlc)}}
                            {\lsequent{\fmld}{\fml}}
                            !
                            \linfer[andRminus]
                            {
                                \linfer[andRminus]
                                {\lsequent{\fmld}{\fml\land(\fmlb\land\fmlc)}}
                                {\lsequent{\fmld}{\fmlb\land\fmlc}}
                            }
                            {\lsequent{\fmld}{\fmlb}}
                        }
                        {\lsequent{\fmld}{\fml\land\fmlb}}
                        !
                        \linfer[andRminus]
                        {\linfer[andRminus]
                            {\lsequent{\fmld}{\fml\land(\fmlb\land\fmlc)}}
                            {\lsequent{\fmld}{\fmlb\land\fmlc}}}
                        {\lsequent{\fmld}{\fmlc}}
                    }
                    {\lsequent{\fmld}{(\fml\land\fmlb)\land\fmlc}}
                \end{sequentdeduction}

            \case{\ref{it:zeroone}} From the derivations

            \begin{minipage}[t]{\minpwidth}
                \begin{sequentdeduction}[array]
                    \linfer[implyR]
                    {
                        \linfer[]
                        {
                            \lclose{}
                        }
                        {\lsequent{\ffalse}{\fml}}
                    }
                    {\lsequent{}{\ffalse\limply\fml}}
                \end{sequentdeduction}
            \end{minipage}
            \begin{minipage}[t]{\minpwidth}
                \begin{sequentdeduction}[array]
                    \linfer[implyR]
                    {
                        \linfer[]
                        {
                            \lclose{}
                        }
                        {\lsequent{\fml}{\ftrue}}
                    }
                    {\lsequent{}{\fml\limply\ftrue}}
                \end{sequentdeduction}
            \end{minipage}
            \qquad\qquad\qquad

            \noindent observe that \(\ffalse\fmlless\fml\) and \(\fml\fmlless\ftrue\) for all formulas \(\fml\).

            \case{\ref{it:adjand}}
            Consider the following three derivation

            \begin{minipage}[t]{\minpwidth}
                \begin{sequentdeduction}[array]
                    \linfer[andR]
                    {
                        \lsequent{\fmlc}{\fml}
                        !
                        \lsequent{\fmlc}{\fml}
                    }
                    {\lsequent{\fmlc}{\fml\land\fmlb}}
                \end{sequentdeduction}
            \end{minipage}
            \begin{minipage}[t]{\minpwidth}
                \begin{sequentdeduction}[array]
                    \linfer[andRminus]
                    {
                        \lsequent{\fmlc}{\fml\land\fmlb}
                    }
                    {\lsequent{\fmlc}{\fml}}
                \end{sequentdeduction}
            \end{minipage}
            \begin{minipage}[t]{\minpwidth}
                \begin{sequentdeduction}[array]
                    \linfer[andRminus]
                    {
                        \lsequent{\fmlc}{\fml\land\fmlb}
                    }
                    {\lsequent{\fmlc}{\fml}}
                \end{sequentdeduction}
            \end{minipage}

            \noindent
            By the left derivation \(\fmlc\fmlless\fml\land\fmlb\) if \(\fmlc\fmlless\fml\) and \(\fmlc\fmlless\fmlb\).
            The converse implication holds by the other derivations

            \case{\ref{it:anddouble}}
            By the derivations

            \begin{minipage}[t]{\minpwidth}
                \begin{sequentdeduction}[array]
                    \linfer[andRminus]
                    {
                        \lsequent{\fmlc}{\fml}
                    }
                    {\lsequent{\fmlc}{\fml\land\fml}}
                \end{sequentdeduction}
            \end{minipage}
            \begin{minipage}[t]{\minpwidth}
                \begin{sequentdeduction}[array]
                    \linfer[contraction]
                    {
                        \linfer[andLminus]
                        {
                            \lsequent{\fml\land\fml}{\fmlc}
                        }
                        {\lsequent{\fml,\fml}{\fmlc}}
                    }
                    {\lsequent{\fml}{\fmlc}}
                \end{sequentdeduction}
            \end{minipage}
            \qquad\qquad\qquad

            \noindent if \(\subreflprov{\relfmlset}{\fmlc\limply (\fml\land\fml)}\) then \(\subreflprov{\relfmlset}{\fmlc\limply \fml}\) and if \(\subreflprov{\relfmlset}{\fml\limply\fmlc}\) then \(\subreflprov{\relfmlset}{(\fml\land\fml)\limply\fmlc}\).
            In other words \(\fml\shyol\fml\land\fml\).

            \case{\ref{it:andproj}} Note that by the derivations

            \begin{minipage}[t]{\minpwidth}
                \begin{sequentdeduction}[array]
                    \linfer[andRminus]
                    {
                        \lsequent{\fmlc}{\fml\land\fmlb}
                    }
                    {\lsequent{\fmlc}{\fml}}
                \end{sequentdeduction}
            \end{minipage}
            \begin{minipage}[t]{\minpwidth}
                \begin{sequentdeduction}[array]
                    \linfer[andL]
                    {
                        \linfer[weak]
                        {
                            \lsequent{\fml}{\fmlc}
                        }
                        {\lsequent{\fml,\fmlb}{\fmlc}}
                    }
                    {\lsequent{\fml\land\fmlb}{\fmlc}}
                \end{sequentdeduction}
            \end{minipage}
            \qquad\qquad\qquad

            \noindent
            clearly \(\fml\land\fmlb\shyol\fml\).
            Similarly show \(\fml\land\fmlb\shyol\fmlb\).

            \case{\ref{it:adjor}}
            Consider the three derivations

            \begin{minipage}[t]{\minpwidth}
                \begin{sequentdeduction}[array]
                    \linfer[orL]
                    {
                        \lsequent{\fml}{\fmlc}
                        !
                        \lsequent{\fmlb}{\fmlc}
                    }
                    {\lsequent{\fml\lor\fmlb}{\fmlc}}
                \end{sequentdeduction}
            \end{minipage}
            \begin{minipage}[t]{\minpwidth}
                \begin{sequentdeduction}[array]
                    \linfer[orLminus]
                    {
                        \lsequent{\fml\lor\fmlb}{\fmlc}
                    }
                    {\lsequent{\fml}{\fmlc}}
                \end{sequentdeduction}
            \end{minipage}
            \begin{minipage}[t]{\minpwidth}
                \begin{sequentdeduction}[array]
                    \linfer[orLminus]
                    {
                        \lsequent{\fml\lor\fmlb}{\fmlc}
                    }
                    {\lsequent{\fmlb}{\fmlc}}
                \end{sequentdeduction}
            \end{minipage}

            \noindent
            where \(\fml\lor\fmlb\fmlless\fmlc\) if \(\fmlb\fmlless\fmlc\) and \(\fml\fmlless\fmlc\) holds by the left derivation. And the converse implication holds by the other two derivations.

            \case{\ref{it:orinj}}
            By the following two derivations

            \begin{minipage}[t]{\minpwidth}
                \begin{sequentdeduction}[array]
                    \linfer[andRminus]
                    {
                        \lsequent{\fmlc}{\fml}
                    }
                    {\lsequent{\fmlc}{\fml\lor\fmlb}}
                \end{sequentdeduction}
            \end{minipage}
            \begin{minipage}[t]{\minpwidth}
                \begin{sequentdeduction}[array]
                    \linfer[orLminus]
                    {
                        \lsequent{\fml\lor\fmlb}{\fmlc}
                    }
                    {\lsequent{\fml}{\fmlc}}
                \end{sequentdeduction}
            \end{minipage}
            \qquad\qquad\qquad

            \noindent it follows that \(\fml\shyol\fml\lor\fmlb\).
            Similarly show \(\fmlb\shyol\fml\lor\fmlb\).

            \case{\ref{it:ordouble}}
            By the following two derivations

            \begin{minipage}[t]{\minpwidth}
                \begin{sequentdeduction}[array]
                    \linfer[orRminus]
                    {
                        \lsequent{\fmlc}{\fml\lor\fml}
                    }
                    {\lsequent{\fmlc}{\fml}}
                \end{sequentdeduction}
            \end{minipage}
            \begin{minipage}[t]{\minpwidth}
                \begin{sequentdeduction}[array]
                    \linfer[orL]
                    {
                        \lsequent{\fml}{\fmlc}
                        !
                        \lsequent{\fml}{\fmlc}
                    }
                    {\lsequent{\fml\lor\fml}{\fmlc}}
                \end{sequentdeduction}
            \end{minipage}
            \qquad\qquad\qquad

            \noindent so \(\fml\lor\fml\shyol\fml\).

            \case{\ref{it:adjimply}}
            By the following derivations

            \begin{minipage}[t]{\minpwidth}
                \begin{sequentdeduction}[array]
                    \linfer[andL]
                    {
                        \linfer[implyRminus]
                        {
                            \lsequent{\fml}{\fmlb\limply\fmlc}
                        }
                        {\lsequent{\fml,\fmlb}{\fmlc}}
                    }
                    {\lsequent{\fml\land\fmlb}{\fmlc}}
                \end{sequentdeduction}
            \end{minipage}
            \begin{minipage}[t]{\minpwidth}
                \begin{sequentdeduction}[array]
                    \linfer[implyR]
                    {
                        \linfer[andLminus]
                        {
                            \lsequent{\fml\land\fmlb}{\fmlc}
                        }
                        {\lsequent{\fml,\fmlb}{\fmlc}}
                    }
                    {\lsequent{\fml}{\fmlb\limply\fmlc}}
                \end{sequentdeduction}
            \end{minipage}
            \qquad\qquad\qquad

            \case{\ref{it:implw}}
            Suppose that there is a \(\relfmlset\) derivation of \(\lsequent{\fml}{\fmlb}\).
            By the derivations

            \begin{minipage}[t]{\minpwidth}
                \begin{sequentdeduction}[array]
                    \linfer[implyR]
                    {
                        \linfer[andR]
                        {
                            \linfer[weak]
                            {
                                \lsequent{\fmld}{\fmlc}
                            }
                            {\lsequent{\fmld,\fml}{\fmlb}}
                            !
                            \linfer[weak]
                            {
                                \lsequent{\fml}{\fmlc}
                            }
                            {\lsequent{\fmld,\fml}{\fmlb}}
                        }
                        {\lsequent{\fmld,\fml}{\fmlb\land\fmlc}}
                    }
                    {\lsequent{\fmld}{\fml\limply(\fmlb\land\fmlc)}}
                \end{sequentdeduction}
            \end{minipage}
            \begin{minipage}[t]{0.5\textwidth}
                \begin{sequentdeduction}[array]
                    \linfer[implyLminus]
                    {
                        \lsequent{\fml}{\fmlb}
                        !
                        \lsequent{\fml\limply(\fmlb\land\fmlc)}{\fmld}
                    }
                    {\lsequent{\fmlc}{\fmld}}
                \end{sequentdeduction}
            \end{minipage}

            \noindent conclude \(\fmlc\shyol\fml\limply(\fmlb\land\fmlc)\).

            \case{\ref{it:implev}}
            Suppose that there is a \(\relfmlset\) derivation of \(\lsequent{\fml}{\fmlb}\).
            By the derivations

            \begin{minipage}[t]{\minpwidth}
                \begin{sequentdeduction}[array]
                    \linfer[andL]
                    {
                        \linfer[implyL]
                        {
                            \lsequent{\fml}{\fmlb}
                            !
                            \lsequent{\fmlc}{\fmld}
                        }
                        {\lsequent{\fmlb\limply\fmlc,\fml}{\fmld}}
                    }
                    {\lsequent{(\fmlb\limply\fmlc)\land\fml}{\fmld}}
                \end{sequentdeduction}
            \end{minipage}
            \qquad
            \begin{minipage}[t]{0.5\textwidth}
                \begin{sequentdeduction}[array]
                    \linfer[cut]{
                        \linfer[andRminus]{
                            \lsequent{\fmld}{(\fml\limply\fmlc)\land\fml}
                        }
                        {\lsequent{\fmld}{\fml}}
                        !
                        \linfer[implyRminus]{
                            \linfer[andRminus]{
                                \lsequent{\fmld}{(\fml\limply\fmlc)\land\fml}
                            }
                            {\lsequent{\fmld}{\fml\limply\fmlc}}
                        }
                        {\lsequent{\fmld,\fml}{\fmlc}}
                    }
                    {\lsequent{\fmld}{\fmlc}}
                \end{sequentdeduction}
            \end{minipage}

            \noindent it follows that \((\fml\limply\fmlc)\land\fml\shyol\fmlc\).
        \end{caselist}
        This concludes the verification that \(\allfmlssha\) is a \semiHeytingalgebra{}.

        It remains to check that \(\allfmlssba\) is a \classical \semiHeytingalgebra{}.
        Then \(\lnot\fml\lor\fmlb\shiso\fml\shimply\fmlb\) by the derivations

        \begin{minipage}[t]{\minpwidth}
            \begin{sequentdeduction}[array]
                \linfer[implyR]
                {
                    \linfer[raa]
                    {
                        \linfer[implyRminus]
                        {
                            \linfer[orRclassicalminus]
                            {
                                \lsequent{\fmlc}{\lnot\fml\lor\fmlb}
                            }
                            \lsequent{\fmlc,\lnot\fmlb}{\lnot\fml}
                        }
                        {\lsequent{\fmlc,\fml,\lnot\fmlb}{\ffalse}}
                    }
                    {\lsequent{\fmlc,\fml}{\fmlb}}
                }
                {\lsequent{\fmlc}{\fml\shimply\fmlb}}
            \end{sequentdeduction}
        \end{minipage}
        \qquad
        \begin{minipage}[t]{0.5\textwidth}
            \begin{sequentdeduction}[array]
                \linfer[orRclassical]{
                    \linfer[implyR]{
                        \linfer[implyL]{
                            \linfer[implyRminus]{
                                \lsequent{\fmlc}{\fml\limply\fmlb}
                            }
                            {
                                \lsequent{\fmlc,\fml}{\fmlb}
                            }
                            !
                            \linfer{\lclose
                            }
                            {
                                \lsequent{\fmlc,\ffalse}{\ffalse}
                            }
                        }
                        {\lsequent{\fmlc,\lnot\fmlb, \fml}{\ffalse}}
                    }
                    {\lsequent{\fmlc,\lnot\fmlb}{\lnot\fml}}
                }
                {\lsequent{\fmlc}{\lnot\fml\lor\fmlb}}
            \end{sequentdeduction}
        \end{minipage}

        \bigskip

            \begin{sequentdeduction}[array]
                \linfer[implyLclassical]
                {
                    \linfer[weak]
                    {
                        \linfer[implyLclassical]
                        {
                            \linfer[orLminus]
                            {
                                \lsequent{\lnot\fml\lor\fmlb}{\fmlc}
                            }
                            {\lsequent{\lnot\fml}{\fmlc}}
                            !
                            \linfer{\lclose}{
                                \lsequent{\ffalse}{\fml}
                            }
                        }
                        {\lsequent{\lnot\fmlc}{\fml}}
                    }
                    {\lsequent{\lnot\fmlc,\fml\shimply\fmlb}{\fml}}
                    !
                    \linfer[orLminus]{
                        \lsequent{\lnot\fml\lor\fmlb}{\fmlc}
                    }
                    {\lsequent{\fmlb}{\fmlc}}
                }
                {\lsequent{\fml\shimply\fmlb}{\fmlc}}
            \end{sequentdeduction}

        \bigskip

            \begin{sequentdeduction}[array]
                \linfer[orL]{
                    \linfer[implyLclassical]{
                        \linfer[implyLclassicalminustwo]{
                            \lsequent{\fml\limply\fmlb}{\fmlc}
                        }
                        {
                            \lsequent{\lnot\fmlc,\lnot\fml}{\fml}
                        }
                        !
                        \linfer{\lclose
                        }
                        {
                            \lsequent{\ffalse}{\fmlc}
                        }
                    }
                    {\lsequent{\lnot\fml}{\fmlc}}
                    !
                    \linfer[implyLclassicalminusone]{
                        \lsequent{\fml\limply\fmlb}{\fmlc}
                    }
                    {
                        \lsequent{\fmlb}{\fmlc}
                    }
                }
                {\lsequent{\lnot\fml\lor\fmlb}{\fmlc}}
            \end{sequentdeduction}

        And by the following derivations \(\fml\lor\fmlb\shiso\lnot\fml\shimply\fmlb\):

        \begin{minipage}[t]{\minpwidth}
            \begin{sequentdeduction}[array]
                \linfer[implyR]
                {
                    \linfer[implyLclassical]
                    {
                        \linfer[orRminus]
                        {
                            \lsequent{\fmlc}{\fml\lor\fmlb}
                        }
                        {\lsequent{\fmlc,\lnot\fmlb}{\fml}}
                        !
                        \linfer{\lclose}{
                            \lsequent{\fmlc,\ffalse}{\fmlb}
                        }
                    }
                    {\lsequent{\fmlc,\lnot\fml}{\fmlb}}
                }
                {\lsequent{\fmlc}{\lnot\fml\shimply\fmlb}}
            \end{sequentdeduction}
        \end{minipage}
        \hfill
        \begin{minipage}[t]{0.5\textwidth}
            \begin{sequentdeduction}[array]
                \linfer[orRclassical]{
                    \linfer[implyRminus]{
                        \lsequent{\fml}{\lnot\fml\limply\fmlb}
                    }
                    {\lsequent{\fmlc,\lnot\fml}{\fmlb}}
                }
                {\lsequent{\fmlc}{\fml\lor\fmlb}}
            \end{sequentdeduction}
        \end{minipage}

        \bigskip

            \begin{sequentdeduction}[array]
                \linfer[implyLclassical]
                {
                    \linfer[weak]
                    {
                        \linfer[implyR]
                        {
                            \linfer[implyL]
                            {
                                \linfer[orRminus]
                                {
                                    \lsequent{\fml\lor\fmlb}{\fmlc}
                                }
                                {\lsequent{\fml}{\fmlc}}
                                !
                                \linfer{\lclose}{
                                    \lsequent{\ffalse}{\ffalse}
                                }
                            }
                            {\lsequent{\lnot\fmlc,\fml}{\ffalse}}
                        }
                        {\lsequent{\lnot\fmlc}{\lnot\fml}}
                    }
                    {\lsequent{\lnot\fmlc,\fml\shimply\fmlb}{\lnot\fml}}
                    !
                    \linfer[orLminus]{
                        \lsequent{\fml\lor\fmlb}{\fmlc}
                    }
                    {\lsequent{\fmlb}{\fmlc}}
                }
                {\lsequent{\lnot\fml\shimply\fmlb}{\fmlc}}
            \end{sequentdeduction}

            \bigskip

            \begin{sequentdeduction}[array]
                \linfer[orL]{
                    \linfer[raa]{
                        \linfer[implyRminus]{
                            \linfer[implyLclassicalminustwo]{
                                \lsequent{\lnot\fml\limply\fmlb}{\fmlc}
                            }
                            {
                                \lsequent{\lnot\fmlc}{\lnot\fml}
                            }
                        }
                        {
                            \lsequent{\fml,\lnot\fmlc}{\ffalse}
                        }
                    }
                    {\lsequent{\fml}{\fmlc}}
                    !
                    \linfer[implyLclassicalminusone]{
                        \lsequent{\lnot\fml\limply\fmlb}{\fmlc}
                    }
                    {
                        \lsequent{\fmlb}{\fmlc}
                    }
                }
                {\lsequent{\fml\lor\fmlb}{\fmlc}}
            \end{sequentdeduction}
    }

    \noindent
    This completes the proof of soundness.
\end{proofE}

\begin{theoremE}[][all end]\label{thm:semiheytingcountermodel}
    For any formula \(\fml\) such that \(\notsubreflprov{\relfmlset}{\fml}\) there is a \semiHeytingalgebra{} model \((\allfmlssha,\heyint)\) with \(\satrue[\allfmlssha]{\heyint}{\relfmlset}\) and \(\sasem[\allfmlssha]{\heyint}{\fml}\not\shiso\saone\).

    If \(\notclasssubreflprov{\relfmlset}{\fml}\) then there is a \semiBooleanalgebra model \((\allfmlssba,\heyint)\) with \(\satrue[\allfmlssba]{\heyint}{\relfmlset}\) and \(\sasem[\allfmlssba]{\heyint}{\fml}\not\shiso\saone\).
\end{theoremE}

\begin{proofE}
    Let \(\heyintof[\heyint]{\patom}=\patom\in\allfmls\) and note by induction that \(\sasem[\allfmlssha]{\heyint}{\fml}=\fml\).
    Observe that by the deduction
    \begin{sequentdeduction}[array]
        \linfer[implyR]
        {
            \linfer[weak]
            {
                \lsequent{} {\fmlb}
            }
            {\lsequent{\ftrue} {\fmlb}}
        }
        {\lsequent{} {\ftrue\limply\fmlb}}
    \end{sequentdeduction}
    also \(\subreflprov{\relfmlset}{\ftrue\limply\fmlb}\) and hence \(\ftrue\fmlless\fmlb\) for all \(\fmlb\in\relfmlset\).
    Thus \(\satrue[\allfmlssha]{\heyint}{\relfmlset}\).
    If \(\saone\shleq\sasem[\allfmlssha]{\heyint}{\fml}=\fml\) were true then \(\subreflprov{\relfmlset}{\saone\limply\fml}\) and by
    \begin{sequentdeduction}[array]
        \linfer[]
        {
            \linfer[implyRminus]
            {
                \lsequent{} {\ftrue\limply\fmlb}
            }
            {\lsequent{\ftrue} {\fmlb}}
        }
        {\lsequent{} {\fmlb}}
    \end{sequentdeduction}
    then \(\subreflprov{\relfmlset}{\fml}\) contradicting the assumption.

    The case for \semiBooleanalgebra{}s is analogous.
\end{proofE}

\begin{theoremE}[Algebraic Completeness][]\label{thm:algebraiccomp}
    The calculus of \subrefllogic is complete for \semiHeytingalgebra{}s, and \classicalsubrefl is complete for \semiBooleanalgebra{}s.
    \[\subreflprov{\relfmlset}{\fml}\;\Leftrightarrow\;\saconseq{\relfmlset}{\fml} \qquad \text{and}\qquad \classsubreflprov{\relfmlset}{\fml}\;\Leftrightarrow\;\saconseqbool{\relfmlset}{\fml}.\]
\end{theoremE}

\begin{proofE}
    Soundness is immediate from \Cref{prop:semiheytingsound} and \Cref{prop:classicalsemiheytingsound}.

    Suppose \(\notsubreflprov{\relfmlset}{\fml}\).
    Then by \Cref{thm:semiheytingcountermodel} there is a \semiHeytingalgebra{} model with \(\satrue[\allfmlssha]{\heyint}{\relfmlset}\) and \(\sasem[\allfmlssha]{\heyint}{\fml}\not\shiso\saone\).
    Then by the assumption \(\saconseq{\relfmlset}{\fml}\) note that \(\saone\shleq\sasem[\allfmlssha]{\heyint}{\fml}\) is a contradiction.

    The case for \classicalsubrefl is analogous.
\end{proofE}
In fact, the proof of soundness shows local soundness: In any (Heyting) \semiBooleanalgebra in which the premise of one of the rules is true, the conclusion is true. Moreover, this also holds for many admissible rules, such as the inverse rules.
Completeness is proved by constructing a \emph{syntactic \semiHeytingalgebra} (see \Cref{sec:proofs}).

\section{Completeness for Denotational Semantics}\label{sec:denotational}

To understand \subrefllogic intuitively, it is useful to consider formulas as denoting the sets of states that satisfy the formula.
This concretely realizes the \emph{abstract} syntactic calculus into concrete set-theoretic semantics, similar to Tarski's semantics for first-order logic.
Such semantics for \subrefllogic are presented in this section, and completeness is proved in a general way that is amenable to extension to first-order logic.

\subsection{\Classical Denotational Semantics}\label{sec:classical}

To interpret \subrefllogic in usual denotational set models of logic, it is necessary to provide two interpretations for atomic propositions.
A set model specifies where a formula is \emph{absolutely true} and where it is only \emph{marginally true}.

\begin{definition}[Set-Model]
    A \emph{set-model} for \subrefllogic consists of a domain \(\setdoma\) together with valuations \(\setinterprob(\patom)\subseteq\setinterpalm(\patom)\) for all \(\patom\in\patoms\).
\end{definition}
Intuitively, the valuation assigns a region \(\setinterprob(\patom)\) to any \emph{atom} \(\patom\), where it is \emph{robustly true}, and a region \(\setinterpalm(\patom)\), where it is \emph{almost true}.
For example, in a topological space \(\setdoma\), a natural interpretation can define marginal truth as the topological closure \(\setinterpalm(\patom)=\closure{\setinterprob(\patom)}\).
However, \subrefllogic can be interpreted more generally in any set-model to describe robust and marginal truth.
The semantics extends from atoms to propositions as follows:
\begin{align*}
     &
    \setsem{\setinterp}{\patom} = \setinterprob(\patom)
     &   &
    \setsem{\setinterp}{\ffalse} = \setsemcl{\setinterp}{\ffalse}=\emptyset
     &   &
    \setsemcl{\setinterp}{\fml\land\fmlb} = {\setsemcl{\setinterp}{\fml}} \intersection {\setsemcl{\setinterp}{\fmlb}}
     &   &
    \setsemcl{\setinterp}{\fml\lor\fmlb} = {\setsemcl{\setinterp}{\fml}} \union {\setsemcl{\setinterp}{\fmlb}}
    \\
     &
    \setsemcl{\setinterp}{\patom} = \setinteralmpof{\patom}
     &   &
    \setsem{\setinterp}{\ftrue} = \setsemcl{\setinterp}{\ftrue} = \tops
     &   &
    \setsem{\setinterp}{\fml\land\fmlb} = {\setsem{\setinterp}{\fml}} \intersection {\setsem{\setinterp}{\fmlb}}
     &   &
    \setsem{\setinterp}{\fml\lor\fmlb} = {\setsem{\setinterp}{\fml}} \union {\setsem{\setinterp}{\fmlb}}
    \\
     &
    \setsem{\setinterp}{\fml\limply\fmlb}={\setcomp{({\setsemcl{\setinterp}{\fml}})}} \union \setsem{\setinterp}{\fmlb}\span\span
     &   &
    \setsemcl{\setinterp}{\fml\limply\fmlb}={\setcomp{({\setsem{\setinterp}{\fml}})}} \union \setsemcl{\setinterp}{\fmlb}.
    \span\span
\end{align*}
As for atoms, \(\setsem{\setinterp}{\fml}\) denotes the states where \(\fml\) is robustly true and \(\setsemcl{\setinterp}{\fml}\) the set of states where \(\fml\) is almost true.
Note that \(\setsem{\setinterp}{\fml}\subseteq\setsemcl{\setinterp}{\fml}\) for any formula.
This makes the semantics a natural tolerant/strict semantics for propositional logic.
Formulas that are assumed (in contravariant position) need to be tolerantly true, while formulas that are shown (in covariant position) need to be strictly true.
For a set \(\fmlset\) of formulas, let \(\setsemcl{\setinterp}{\fmlset}=\Intersection_{\fml\in\fmlset}\setsemcl{\setinterp}{\fml}\).
A sequent \(\lsequent{\fmlset}{\fml}\) is \emph{true} (\(\settrue{\setinterp}{\lsequent{\fmlset}{\fml}}\)) in a set model \((\setdoma,\setinterp)\) if \(\setsemcl{\setinterp}{\fmlset}\subseteq\setsem{\setinterp}{\fml}\).
A formula \(\fml\) is a \emph{semantic consequence} \(\setconseq{\relfmlset}{\fml}\) of a set of sequents \(\relfmlset\) if \(\settrue{\setinterp}{\fml}\) for all set models \((\setdoma,\setinterp)\) with \(\settrue{\setinterp}{\relfmlset}\).

From the definition of set semantics, it is clear that the set-interpretation coincides with the \SemiBooleanalgebra semantics in the powerset \semiBooleanalgebra{}.
Any set model \((\setdoma,\setinterp)\) is a \semiHeytingalgebra{} model on \(\settopowob[\setdoma]\) with \(\sasem[{\settopowob[\setdoma]}]{\heyintderived{\setinterp}}{\fml}=(\setsem{\setinterp}{\fml},\setsemcl{\setinterp}{\fml})\).

\begin{theoremE}[Soundness]\label{thm:setsoundness}
    The calculus of \classicalsubrefl is sound for set semantics: if \(\subreflprov{\relfmlset}{\fml}\) then \(\setconseq{\relfmlset}{\fml}\).
\end{theoremE}
\begin{proofE}\label{proof:setsoundness}
    Consider a set model \((\setdoma,\setinterp)\) with \(\settrue{\setinterp}{\relfmlset}\).
    Then \(\satrue[{\settopowob[\tops]}]{\setinterp}{\relfmlset}\) and \(\satrue[{\settopowob[\tops]}]{\setinterp}{\fml}\) by \Cref{prop:semiheytingsound}.
    Hence, \(\settrue{\setinterp}{\fml}\) as any \((\setdoma,\setinterp)\) set model is a \semiHeytingalgebra{} model on \(\settopowob[\setdoma]\) with \(\sasem[{\settopowob[\setdoma]}]{\heyintderived{\setinterp}}{\fml}=(\setsem{\setinterp}{\fml},\setsemcl{\setinterp}{\fml})\).
\end{proofE}
The converse, i.e., completeness for set semantics of \classicalsubrefl will be shown in \Cref{thm:setcompleteness}.
As mentioned in \Cref{ex:pairsemiba}, these semantics can be interpreted equivalently as partial \(2\)-valued functions, and the next two examples show how the set interpretation of \subrefllogic captures the partial function interpretation.

\begin{example}[Program Termination]
    For a computer program \(A\), the robust truth set \(\setsem{\setinterp}{\fml}\) can be interpreted as the set of inputs the program accepts and \(\setsemcl{\setinterp}{\fml}\) as the set of inputs which program \(A\) rejects.
    Implication can be understood as conditional acceptance, i.e., \(\fml\limply\fmlb\) means `\(\fmlb\) accepts all inputs that \(\fml\) does not reject'.
    In other words, program~\(\fmlb\) is a safe version of \(\fml\).
    In this context, it is critical that \(\fml\limply\fmlb\) is not interpreted as `\(\fmlb\) accepts all inputs that \(\fml\) accepts'.
    In fact, no program \(\fmlc\) could realize this property due to the undecidability of the halting problem. (If \(\fmlb\) always rejects, then \(\fmlc\) would have to accept exactly those inputs where \(\fml\) does not accept.)
    The identity principle forces \emph{totality}.
\end{example}

\begin{example}[Logic of Logic]
    There are interesting mathematical structures, whose logic is naturally of this form.
    Suppose the atomic propositions denote formulas in some other logical calculus.
    The denotation \(\setsem{\setinterp}{\fml}\) can be the set of assumptions from which a formula \(\fml\) is provable, while the denotation \(\setsemcl{\setinterp}{\fml}\) is the set of formulas which are consistent with \(\fml\).
    Then, the interpretation of implication can be understood as: a formula proves \(\fml\limply\fmlb\) iff it is inconsistent with \(\fml\) or proves \(\fmlb\).
    Here the formula \(\fml\limply\fml\) denotes exactly the formulas that ensure \(\fml\) is either provable or inconsistent.
    The identity principle forces \emph{completeness}.
\end{example}

\subsection{Representation of \SemiBooleanAlgebra{}s}\label{sec:representation}\label{sec:setcompsection}

\SemiBooleanalgebra{}s are complex algebraic structures.
This section presents a representation theorem for these structures, expressing them concretely as powerset \semiBooleanalgebra{}s, similar to Stone duality \cite{alma990003267950106470}.
The key ingredient here is the Kleene three-element \semiBooleanalgebra{} \(\kleenethree\), which serves (almost) as a dualizing object \cite{alma990003267950106470}.
Because of the nonreflexivity of \semiBooleanalgebra{s}, this requires a careful extension and in particular a new notion of a semiultrafilter.
The definition of a filter for a \semiHeytingalgebra is as usual:

\begin{definition}
    A subset \(\filter\subseteq\shs\) of a \semiHeytingalgebra{} is a \emph{filter}~if  \(\sazero\notin\filter\) and
    \begin{enumerate}
        \item if \(\saa\in\filter\) and \(\saa\shyol\sab\) , then \(\sab\in\filter\)
        \item if \(\saa\in\filter\) and \(\sab\in\filter\), then \(\saa\shand\sab\in\filter\)
    \end{enumerate}
    A filter is \emph{prime} if \(\saa\in\filter\) or \(\sab\in\filter\), whenever \(\saa\shor\sab\in\filter\).
\end{definition}

Note that filters are upward closed with respect to the \yonedaorder and not the nonreflexive order \(\shleq\).
Since \subrefllogic does not necessarily satisfy the law of the excluded middle, i.e., \(\lnot\saa\shor\saa\), it is impossible to use an ultrafilter construction for the representation theorem.
The key ingredient instead are semiultrafilters.

\begin{definition}[Semiultrafilter]
    For a subset \(X\subseteq\shs\) of a \semiHeytingalgebra let~\(X' = \{\saa : \mexists{\sab\in\filter}\sab\shleq\saa\}\).
    A filter \(\filter\subseteq\shs\) is a \emph{semiultrafilter} if \(\saa\notin\filter\) implies \(\lnot\saa\in\filter'\).
\end{definition}

Note that in a semiultrafilter, \(\saa\notin\filter\) if and only if \(\lnot\saa\in\filter'\).
Intuitively, a filter \(\filter\) contains all propositions with truth value \(\geq\midpoint\), whereas a semiultrafilter should ensure that every proposition is assigned a nonzero truth value or its negation is assigned a truth value \(1\).
It suffices to ensure that its negation is a \(\shleq\)-consequence of another filter element.

\begin{lemmaE}[Filter Existence][]\label{lem:ultrafilterexistence}
    Let \(\shs\) be a Heyting (Boolean) semialgebra and \(\saa\not\shleq\sab\).
    There is a prime filter (semiultrafilter) \(\filter\) such that \(\saa\in\filter\), \(\sab\notin\filter'\).
\end{lemmaE}

\begin{proofE}
    Let \(\filterp=\{\sac : \saa\shyol\sac\}\) and note that \(\filterp\) is a filter such that \(\saa\in\filtermax\) and \(\sab\notin\filtermax'\).

    By Zorn's lemma, there is a maximal filter \(\filtermax\) such that \(\saa\in\filtermax\) and \(\sab\notin\filtermax'\).
    It remains to show that this filter is prime.
    So consider \(\sac\shor\sad\in\filtermax\) and define
    \begin{align*}
        \filterb_\sac & =\{x: y \in \filtermax\mand y\shand \sac \shyol x\}
        \qquad        &
        \filterb_\sad & =\{x: y \in \filtermax\mand y\shand \sad \shyol x\}
    \end{align*}
    We claim that for either \(\sab\notin\filterb_\sac'\) or \(\sab\notin\filterb_\sad'\).
    If this were not the case, there would be \(y_1,y_2\in \filtermax\) such that \(y_1\shand\sac\shleq\sab\) and \(y_2\shand\sad\shleq\sab\).
    Consequently  \((y_1\shand y_2)\shand(\sac\shor\sad)\shleq\sab\) and as \((y_1\shand y_2)\shand(\sac\shor\sad)\in\filtermax\) by definition \(\sab\in \filtermax'\), which contradicts \(\sab\notin\filtermax^1\).
    Without loss of generality, \(\filterb_\sac\) is a filter with \(\sab\notin\filterb_\sac\) and \(\filtermax\subseteq\filterb_\sac\).
    So by maximality \(\sac\in\filterb_\sac \subseteq \filtermax\).

    Now suppose \(\shs\) is a \semiBooleanalgebra and first show that \(\lnot\sab\in\filter\).
    Note that \(\filterb_{\lnot{\sab}}\), as defined above, is a filter with \(\sab\notin\filterb_{\lnot\sab}\).
    To see this suppose for a contradiction that there were \(y\in\filtermax\) such that \(y\shand\lnot\sab\shleq\sab\) then
    \(y\shand\lnot\sab\shyol y\shand\lnot\sab\shand\lnot\sab\shyol \sazero\) ant thus \(y\shand\lnot\sab\shleq\sazero\).
    Then \(y\shand\shleq\lnot\lnot\sab\) by \ref{it:adjimply} and by \classicality this would mean \(y\shleq\sab\) contradicting \(\sab\notin\filter'\).
    Thus by maximality \(\lnot\sab\in\filter\).
    To prove that \(\filter\) is a semiultrafilter by the contrapositive assume \(\lnot\sac\notin\filter'\) and note that
    Again by maximality it suffices to show that \(\sab\notin\filterb_{\lnot{\sac}}'\).
    If this were the case, there would be \(y\in\filtermax\) such that \(y\shand\lnot\sac\shleq\sab\).
    By \itrefof{lem:implyproperties}{it:moveing}, equivalently \(y\shand\lnot\sab\shleq\sac\).
    Since \(y\shand\lnot\sab\in\filtermax\) then \(\lnot\sac\notin\filter'\)
\end{proofE}
The existence of semiultrafilters suffices to show that enough homomorphism into \(\kleenethree\) exist to separate elements.

\begin{lemmaE}[Separation][]\label{lem:separation}
    Let \(\shs\) be a \semiBooleanalgebra{} with \(\saa\not\shleq\sab\).
    There is a \semiBooleanalgebra{} homomorphism \(f:\shs\to\kleenethree\) with \(f(\saa)\geq \midpoint\geq f(\sab)\).
\end{lemmaE}

\begin{proofE}
    Pick a semiultrafilter \(\filter\) such that \(\saa\in\filter\) and \(\sab\notin\filter'\) by \Cref{lem:ultrafilterexistence}.
    Let
    \[
        f(\sac) =
        \begin{cases}
            0         & \text{if } \sac\notin\filter \\
            1         & \text{if } \sac\in \filter'  \\
            \midpoint & \text{otherwise}
        \end{cases}
    \]
    Note \(f(\saa)\geq \midpoint\geq f(\sab)\), since \(\saa\in\filter\) and \(\sab\notin\filter'\).
    It is not hard to see that \(f\) is a \semiBooleanalgebra{} homomorphism.
\end{proofE}
Interestingly, the lemma does \emph{not} guarantee that \(f(\saa) \neq f(\sab)\).
Nonetheless \Cref{lem:separation} provides sufficient separation.
Ordinarily, homomorphism are only required to preserve \(\shleq\), however the separating homomorphism of \Cref{lem:separation} also preserves~\(\shyol\).

\begin{propositionE}[][]\label{prop:setissemi}
    There is a  \semiBooleanalgebra{} isomorphism
    \(\{f :\tops\to\kleenethree\}\cong \settopowob[\setdoma]\).
\end{propositionE}
\begin{proofE}
    Function \(f :\tops\to\kleenethree\) bijectively homomorphically maps to \((f^{-1}(\{1\}),f^{-1}(\{1,\midpoint\}))\).
\end{proofE}
In fact, \(\settopowfunc:\opcategory{\catset}\to\catsba\) is a contravariant functor acting on functions \(f:X\to Y\) as \(\settopowmorph{f} = f^{-1}\times f^{-1}\) and is naturally isomorphic to \(\Homs{\catset}{-}{\kleenethree}\) when each \(\Homs{\catset}{X}{\kleenethree}\) is equipped with the pointwise \semiBooleanalgebra{} structure.
A homomorphism \(f\) of \semiHeytingalgebra{}s is \emph{\yonedainjective} if \(\saa\shiso\sab\) whenever \(f(a)\shiso f(b)\).

\begin{theoremE}[\SemiBooleanAlgebra{} Representation][]\label{thm:representation}
    For any \semiBooleanalgebra{} \(\shs\), there is a set \(\lathomalg\) and a \yonedainjective homomorphism of \semiBooleanalgebra{}s \(\iota:\shs\to\settopowob[\lathomalg]\).
\end{theoremE}

\begin{proofE}
    For a \semiBooleanalgebra{} \(\shs\) let \(\lathomalg=\Homs{\catsba}{\shs}{\kleenethree}\) be the set of all homomorphisms \(\lathom:\shs\to\threealg\).
    The set of functions \(\lathomalg\to\kleenethree\) carries a pointwise \semiBooleanalgebra{} structure.
    For every \(\saa\in\shs\) let \(\iota(\saa):\lathomalg\to\kleenethree\) be defined by \(\iota(\saa)(f)=f(\saa)\) for all \(f\in\lathomalg\).
    Note that this is a \semiBooleanalgebra{} homomorphism.

    For \yonedainjectivity consider \(\saa\not\shyol\sab\).
    Then there are \(\sac\shleq\saa\) and \(\sac\not\shleq\sab\) and by \Cref{lem:separation} there is a homomorphism \(f\in\lathomalg\) such that \(f(\sac)\geq \midpoint\geq f(\sab)\).
    Then \[\iota(\saa)(f)=f(\saa)>f(\sac)\geq \midpoint \geq f(\sab)=\iota(\sab)(f)\]
    Define \(g:\lathomalg\to\kleenethree\) so that \(g(f)=\midpoint\) and \(g(h)=0\) for \(h\neq f\).
    Then \(g\shleq \iota(a)\) and \(g\not\shleq\iota(b)\) and thus \(\iota(a)\not\shiso\iota(b)\).

    It is easy to see that \(\settopowob[\lathomalg]\) as defined in \Cref{ex:pairsemiba} is isomorphic to the pointwise \semiBooleanalgebra of functions \(\lathomalg\to\kleenethree\).
    Thus \(\iota\) can be viewed as a \yonedainjective homomorphism of \semiBooleanalgebra{}s \(\iota:\shs\to\settopowob[\lathomalg]\) by \Cref{prop:setissemi}.
\end{proofE}
The representation theorem shows that all \semiBooleanalgebra{}s behave essentially like the powerset \semiBooleanalgebra{}.
This provides a powerful means to study the abstract algebraic structure concretely.
The proof indicates the fundamental importance of the three-element Kleene logic \(\kleenethree\) for logics without identity.

In categorical terms, \Cref{thm:representation} shows that \(\Homs{\catsba}{-}{\kleenethree}:\catsba\to\opcategory{\catset}\) is left adjoint to \(\settopowfunc\).
In fact, the adjunction is essentially the duality given by the dualizing object \(\kleenethree\):
\[
    \Homs{\catsba}{-}{\kleenethree} \dashv \Homs{\opcategory{\catset}}{-}{\kleenethree} \cong \settopowfunc(-).
\]
By \Cref{prop:setissemi}, the functors \(\Homs{\opcategory{\catset}}{-}{\kleenethree}\) and \(\settopowfunc\) are essentially the same, when the Hom-set is equipped with the pointwise \semiBooleanalgebra{} structure.
The embedding \(\iota\) is the unit of the adjunction, and the critical part is showing that it is \yonedainjective.

Finally, the representation theorem yields a strong completeness result for \subrefllogic.
Provability in \classicalsubrefl is characterized completely by validity in powerset \semiBooleanalgebra{}s and equivalently in terms of robust semantics for formulas.

\begin{theoremE}[Set Completeness]\label{thm:setcompleteness}
    The calculus of \classicalsubrefl is sound and complete for set semantics: \(\setconseq{\relfmlset}{\fml}\) if and only if \(\classsubreflprov{\relfmlset}{\fml}\).
\end{theoremE}
\begin{proofE}
    Suppose for the contrapositive \(\notclasssubreflprov{\relfmlset}{\fml}\).
    By \Cref{thm:semiheytingcountermodel} there is a \semiBooleanalgebra model \((\shs,\heyint)\) with \(\satrue[\shs]{\heyint}{\relfmlset}\) and \(\sasem[\shs]{\heyint}{\fml}\not\shiso\saone\).
    By \Cref{thm:representation} the model \(\shs\) injects homomorphically \(\iota:\shs\to \settopowob[\lathomalg]\), so that \(\satrue[{\settopowob[\shs]}]{J}{\relfmlset}\) where \(J(\patom)=\iota(\heyint(\patom))\).
    As \(\iota\) is \yonedainjective \(\sasem[{\settopowob[\shs]}]{J}{\fml}\not\shiso\saone\) and then \(\notsetconseq{\relfmlset}{\fml}\).
\end{proofE}

\section{Categorical Semantics}\label{sec:categorical}

\Subrefllogic describes ordinary propositional logic without reliance on the axiom of identity.
To understand this phenomenon structurally, this section presents a categorical interpretation of \subrefllogic.
However, as any category necessarily contains identity morphisms, an interpretation in any kind of category is impossible, without inadvertently reintroducing reflexivity.
Instead, \subrefllogic can be interpreted in \emph{semicategories}, which are categories except that identity morphisms are not required to exist.
The difficulty when interpreting logic in semicategories lies in properly defining the typical categorical limits (products, coproducts and exponentials) in semicategories, which in ordinary categories rely critically on identity morphisms, without inadvertently reintroducing identities.
This section shows that these limits can be analyzed structurally \emph{without} identity morphisms and that identity morphisms are not crucial to the structural integrity of logic.
To do this uniformly, this section generalizes adjunctions to the identity-free setting.

The results from this section specialize to \semiHeytingalgebra{s} \(\shs\), which are semicategories where the objects are the elements \(\saa\in\shs\) and there is an arrow \(a\to b\) iff \(\saa\shleq\sab\).
The lack of identity morphisms corresponds directly to the failure of reflexivity.
The semicategorical interpretation explains the structure behind \semiHeytingalgebra{s}.

\subsection{Adjunctions in Semicategories}

The key insight to defining adjunctions in semicategories is to shift the perspective from the morphisms themselves to their behaviour under composition.
While this perspective is ordinarily equivalent by Yoneda's lemma, this is not the case in semicategories where Yoneda's lemma does not hold.
The following definition describes a category that captures the relative behaviour of morphisms in a semicategory.

\begin{definition}
    For a semicategory \(\semicat\), the \yonedacatname \(\derivedcat[\semicat]\) has the same objects as \(\semicat\) and has as morphisms \(\dercatmorph:\catobj\to\catobjb\) pairs \(\dercatmorphpre,\dercatmorphpost\) of natural transformations
    \begin{align*}
        \dercatmorphpre:\Covhom{\semicat}{\catobjb}\naturalto\Covhom{\semicat}{\catobj}   
         &  &
        \dercatmorphpost:\Conthom{\semicat}{\catobj}\naturalto\Conthom{\semicat}{\catobjb}, 
    \end{align*}
    such that \(f\circ \dercatmorphpost(g) = \dercatmorphpre(f)\circ g\) (consistency) with composition \(\dercatmorph\circ\dercatmorphb=(\dercatmorphpre[\dercatmorph]\circ\dercatmorphpre[\dercatmorphb],\dercatmorphpost[\dercatmorphb]\circ\dercatmorphpost[\dercatmorph])\).
\end{definition}

The \yonedacatname is a notion of reflexive closure, which reflects the identity-free structure completely.
Intuitively, morphisms in \(\derivedcat\) describe a pre-composition operation \(\dercatmorphpre\) and a post-composition operation \(\dercatmorphpost\).
In particular, any \(\semicat\)-morphism \(\catmorph:\catobj\to\catobjb\) gives rise to a derived \(\derivedcat\)-morphism \(\derivedmor[\catmorph]:\catobj\to\catobjb\) defined as precomposition \(\dercatmorphpre[{{\derivedmor[\catmorph]}}]=\derivedprecomp{\catmorph}\) and postcomposition \(\dercatmorphpost[{{\derivedmor[\catmorph]}}]=\derivedpostcomp{\catmorph}\).
Note that \(\catmorph\mapsto\derivedmor[\catmorph]\) gives a semifunctor from \(\semicat\) to \(\derivedcat[\semicat]\).
For ordinary categories, this (Yoneda) embedding is a full and faithful functor and ensures that \(\semicat\) and \(\derivedcat\) are isomorphic.
(The consistency assumption ensures that the two natural transformations arise from the same morphism.)
However, in semicategories the \yonedacatname is a new perspective that adds identities on a meta-level and leaves the original category undisturbed.

In the case of a \semiHeytingalgebra, the \yonedacatname is exactly given by the \yonedaorder, i.e., the morphisms of the \yonedacatname are \(\saa\shimply\sab\) if and only if \(\saa\shyol\sab\).
In both cases the idea is to look at at the relative ordering/Hom-sets instead of the actual ordering/morphisms themselves in a reflexive order/category.
In the following, all morphisms are assumed to be in the original semicategory \(\semicat\) unless explicitly stated  otherwise.

The unifying categorical framework for understanding logical connectives is based on adjunctions (or logical harmony).
For example, conjunction arises as the right adjoint to the diagonal functor \(\Delta:\semicat\to\semicat\times\semicat\).
In semicategories, semiadjunctions can be defined naturally (and in keeping with the Yoneda perspective) in terms of Hom-set bijections.
A \emph{(semi)functor} \(\semifunc: \semicat\to\semicatb\) between semicategories maps the objects and morphisms of \(\semicat\) to the objects and morphisms of \(\semicatb\) and respects composition (but not necessarily identity morphisms), i.e., \(\semifunc(f \compose g) = \semifunc(f) \compose \semifunc(g)\).

\begin{definition}
    A weak \emph{semiadjunction} between two semifunctors \(\semifunc:\semicat\to\semicatb\) and \(\semifuncb:\semicatb\to\semicat\)  of semicategories is a family of bijections
    \(\adjbi : \Homs{\semicatb}{\semifunc\semiobj}{\semiobjb} \to \Homs{\semicat}{\semiobj}{\semifuncb\semiobjb}\)
    natural in \(\semiobj,\semiobjb\).\footnote{Naturality means that
        \(\adjbi[\semiobj',\semiobjb](h\circ \semifunc f) = \adjbi[\semiobj,\semiobjb](h)\circ f\) and
        \(\adjbi[\semiobj,\semiobjb']( g\circ h) = \semifuncb g\circ \adjbi[\semiobj,\semiobjb](h)\)
        for all \(h\in\Homs{\semicatb}{\semifunc\semiobj}{\semiobjb}\), \(f\in\Homs{\semicatb}{\semiobj'}{\semiobj}\), \(g\in\Homs{\semicat}{\semiobjb}{\semiobjb'}\). Note that since \(\semicat,\semicatb\) are only semicategories, naturality in both variables separately is potentially stronger than naturality in the product.}
\end{definition}

This notion is slightly different from the previously considered notion of \emph{semiadjunction} between semifunctors on \emph{categories} \cite{DBLP:journals/tcs/Hayashi85}.

Weak semiadjunctions are insufficient to fully characterize logical connectives.
The reason is that semiadjunctions do not guarantee the existence of \emph{internal} units and counits, which are derived from identity morphisms.
Demanding the existence of units and counits in semicategories is problematic, as they impose additional identity-like structure on the semicategory.
This is undesirable, since a semiadjunction should only capture a \emph{connection} between structure-preserving operations and not consist of structure itself.
The literal existence of units and counits would reduce everything to ordinary categories.
For example from the unit \(d:\semiobj\to\semiobj\times\semiobj\) and the counit \(\pi_i:\semiobj_1\times\semiobj_2\to\semiobj_i\) of the typical product adjunction \(\strongadj{\Delta}{\times}\), one can recover all the identities as \(\id_\semiobj = \pi_1\circ d\) by the triangle identities.

Nonetheless, the unit and counit play an indispensable role.
For example, the counit \(\pi_i:\semiobj_1\times\semiobj_2\to\semiobj_i\) of the product adjunction is the projection, which is critical for the elimination of conjunction (as in \irref{andL}).
This indicates that unlike in ordinary categories, where the product can be defined as the right adjoint of the diagonal functor~\(\Delta\), if it exists, the notion of semiadjunction as a Hom-set bijection by itself is not sufficient to fully capture the required notion of product (conjunction).
This is resolved by the \yonedacatname.

The difficulty is that simply requiring a standard adjunction on \(\derivedcat\) is insufficient for two reasons.
First, this is not even definable, since functors do not naturally lift from \(\semicat\) to \(\derivedcat[\semicat]\).
Second, as \Cref{ex:strangethree} shows, the adjunction on \(\derivedcat\) is actually too strong a notion.
The solution is to characterize unit and counit in \(\derivedcat\), without inadvertently forcing their existence as morphisms in the original semicategory.

\begin{definition}
    A \emph{semiadjunction} \(\strongadj{\semifunc}{\semifuncb}\) between semifunctors consists of transformations
    \begin{align*}
        \eta_{\semiobj,\semiobjb} :\Homs{\semicat}{\semifuncb\semifunc \semiobj}{\semiobjb} \to \Homs{\semicat}{\semiobj}{\semiobjb} \qquad\qquad
        \varepsilon_{\semiobj,\semiobjb} :\Homs{\semicatb}{\semiobj}{\semifunc\semifuncb\semiobjb} \to \Homs{\semicat}{\semiobj}{\semiobjb}
    \end{align*}
    natural in \(\semiobj\) and \(\semiobjb\) such that the following \emph{triangle identities} holds
    \[
        \varepsilon_{\semifunc \semiobj,\semiobjb}(\semifunc(\eta_{\semiobj,\semifuncb\semiobjb}(\semifuncb g)))= g \text{ for } g:\semifunc \semiobj\to\semiobjb
        \qquad\qquad
        \eta_{\semiobj,\semifuncb\semiobjb}(\semifuncb(\varepsilon_{\semifunc\semiobj,\semiobjb}(\semifunc f)))= f \text{ for } f: \semiobj\to\semifuncb\semiobjb.
    \]
    The transformation \(\eta\) is the \emph{unit} and \(\varepsilon\) is the \emph{counit}.
\end{definition}

The unit \(\eta\) can be viewed as a family of \(\derivedcat[\semicat]\)-morphisms \(\semiobj \to \semifuncb\semifunc\semiobj\), where \(\dercatmorphpre[\eta_\semiobj] = \eta_{\semiobj, -}\) and \(\dercatmorphpost[\eta_\semiobj](f)=\eta_{\semiobj, -}(\semifuncb\semifunc f)\), which is natural in \(\semiobj\).
Dually, the counit \(\varepsilon\) is a family of \(\derivedcat[\semicatb]\)-morphisms \(\semifunc\semifuncb\semiobj \to \semiobj\) where \(\dercatmorphpost[\varepsilon_\semiobj] = \varepsilon_{\semiobj,-}\) and \(\dercatmorphpre[\varepsilon_\semiobj](g) = \varepsilon_\semiobj(\semifunc\semifuncb g)\), which is natural in \(\semiobj\).

\begin{theorem}
    Any semiadjunction \(\eta,\varepsilon\) is a weak semiadjunction with
    \[
        \theta_{a,b}(g)=\eta_\semiobj(\semifuncb g) \text{ for }g:\semifunc\semiobj\to\semiobjb
        \qquad\text{ with inverse }\qquad
        \theta_{a,b}^{-1}(f)=\varepsilon_\semiobj(\semifunc f) \text{ for }:\semiobj\to\semifuncb\semiobjb.
    \]
\end{theorem}

The converse, that any weak semiadjunction is a semiadjunction, is not necessarily true in general semicategories, since functors do not necessarily lift from \(\semicat\) to \(\derivedcat\).
If \(\semicat\), \(\semicatb\) are ordinary categories a weak semiadjunction is an ordinary adjunction and the lifting of the ordinary unit \(\eta_\semiobj\) and counits \(\varepsilon_\semiobj\) to \(\derivedcat[\semicat],\derivedcat[\semicatb]\) make it a semiadjunction.

The concept of semiadjunctions makes it possible to categorically define the structure of a \emph{product} in a semicategory as a right semiadjoint to the diagonal functor \(\Delta\).
\begin{example}\label{ex:prodadj}
    The product semiadjunction \(\strongadj{\Delta}{\times}\) consists of the following unit and counit
    \begin{align*}
        \eta
        =\Covhom{\semicat}{\semiobj\times\semiobj}\naturalto\Covhom{\semicat}{\semiobj}
        \quad\quad
        \varepsilon
        =\Conthom{\semicat\times\semicat}{\semiobj_1\times \semiobj_2}\naturalto\Conthom{\semicat\times\semicat}{(\semiobj_1,\semiobj_2)}.
    \end{align*}
    In the interpretation of the \yonedaorder for a \semiHeytingalgebra{}, these correspond exactly to conditions \ref{it:anddouble} and \ref{it:andproj} of \Cref{def:semiheytingalgebra}.
\end{example}

\Cref{ex:prodadj} gives an identity free definition of the product.
This is in contrast to other definitions, which always involve projection maps that inadvertently reintroduce identities.

Similar to \Cref{ex:prodadj}, \emph{coproducts} in a semicategory are defined as left semiadjoints to the diagonal functor.
The \emph{exponential adjunction} \(\strongadj{-\times \semiobj}{(-)^\semiobj}\) defines implication, and from a semiadjunction perspective, it is imperative that this is viewed as a \emph{parameterized semiadjunction}, rather than a family of semiadjunctions for every object \(\semiobj\).
Most importantly, the unit and counit are not indexed by objects \(\semiobj\) of the semicategory, but by \emph{morphisms} of the parameter.
For example, the unit of the adjunction is a family of \(\derivedcat\)-morphisms
\(\eta_{f,\semiobj}:\semiobj\to(\semiobjb\times\semiobjc)^\semiobjc\)
for every \(\semicat\) morphism \(f:\semiobj\to\semiobjb\).
In contrast, in ordinary categories indexing by identity morphisms is sufficient.
This explains the additional side condition in \ref{it:implw} and \ref{it:implev} of \Cref{def:semiheytingalgebra}.
Without this side condition, \semiHeytingalgebra{s} would not be a complete model for \subrefllogic.

\subsection{Semicategorical Logic}

A \emph{cartesian closed semicategory} \(\semicat\) is a category with functors \(1: 1 \to \semicat\), \(\times : \semicat\times\semicat\to\semicat\) and \((-)^\semiobj :\semicat\to\semicat\) with terminal objects, products, coproducts and exponentials (i.e., semiadjunctions \(\strongadj{!}{1}\), \(\strongadj{\Delta}{\times}\), \(\strongadj{-\times \semiobj}{(-)^\semiobj}\) where \(!\) is the functor from any category to the terminal category).
A \emph{bicartesian closed semicategory} is a cartesian closed semicategory with additional semiadjoints \(\strongadj{0}{{!}}\) and \(\strongadj{\sqcup}{\Delta}\).
Just as a pre-order can be viewed as a thin category, a transitive order \(\shleq\) can be viewed as a thin semicategory where there is an arrow \(\saa\to\sab\) iff \(\saa\shleq\sab\).
From this perspective a \semiHeytingalgebra{} is a bicartesian closed semicategory.
Surprisingly, the converse is not true, as the product in (thin) a bicartesian closed semicategory does not need to be (associative) monoidal.
However, clearly:

\begin{theorem}
    A thin bicartesian closed semicategory \(\semicat\) is a \semiHeytingalgebra{} iff \(\derivedcat[\semicat]\) is monoidal with respect to the product.
\end{theorem}

The algebraic semantics of \Cref{sec:semisemantics} can be generalized to arbitrary monoidal bicartesian closed semicategories.
The following completeness result is inherited directly from the case for \semiHeytingalgebra{s} and the soundness direction is proved just like for \semiHeytingalgebra{s}.

\begin{theorem}
    \Subrefllogic is sound and complete with respect to bicartesian closed semicategories with monoidal product.
\end{theorem}

As in categorical logic, semicategories provide a setting to investigate proofs of \subrefllogic structurally.
Let \(\subrefllogiccat\) be the category whose objects are propositions and whose morphisms \(\fml\to\fmlb\) are derivations witnessing \(\subreflprov{\relfmlset}{\lsequent{\fml}{\fmlb}}\) up to cut-elimination-step equivalence.
This category is a bicartesian closed semicategory with monoidal product.

\section{Related Work}

Robust classical tolerant/strict and strict/tolerant interpretations have been defined semantically \cite{DBLP:journals/jphil/CobrerosERR12} to avoid paradoxes via nonreflexivity \cite{French2016StructuralRA} and nontransitivity \cite{DBLP:journals/jphil/BarrioRT15}.
Unlike these previously considered perspectives on tolerant/strict logics, which have focused on meta-inferences \cite{DBLP:journals/jphil/BarrioPS20}, this work concerns object-level inferences.
Failure of reflexivity has been shown to be dual (in metainferences) to failure of cut in a precise sense \cite{DBLP:journals/jancl/RePST20}.
Previous proof-theoretic work \cite{DBLP:journals/sLogica/NicolaiR23} studied purely classical logical theories of non-reflexive consequence relations \(C(\ulcorner \fml\urcorner,\ulcorner \fmlb\urcorner)\) to describe semantic truth predicates and proves cut-elimination.
The restriction of initial sequents to subsets of identities (uratoms) in \LJ with definitions was shown to admit cut-elimination \cite{Schroeder-Heister2016}.
Unlike these studies, this paper investigates the structural importance of identity in the intuitionistic and classical case and provides algebraic, denotational and categorical semantics that are foundational for all logics without identity.

Categorically, alternative adjunctions between semifunctors on categories have been used for nonextensional \(\lambda\)-calculus \cite{DBLP:journals/tcs/Hayashi85} in categories with identities.
Another \emph{reflexive} generalization of Heyting algebras to semi-Heyting algebras and their corresponding semi-intuitionistic logics have been considered \cite{DBLP:journals/sLogica/Cornejo11}.

\Subrefllogic belongs to the family of substructural logics, which also includes linear logic~\cite{DBLP:journals/tcs/Girard87} as a contraction-free logic and relevance logics~\cite{DBLP:journals/jsyml/Belnap60,DBLP:journals/jphil/Dosen92a}.
Instances of sequent calculi without identity assumptions have been applied in logic programming \cite{DBLP:journals/logcom/Stark91}.
Notions of robustness have been used for real arithmetic~\cite{DBLP:conf/lics/GaoAC12,DBLP:conf/ijcar/AbouElWafaP26}.

\section{Conclusion}\label{sec:conclusion}

\Subrefllogic has been introduced as a conceptually simple generalization of propositional logic, which, despite the lack of the identity principle, has an elegant syntactic and semantic theory.
The paper demonstrates that rather than being foundational and indispensable, \emph{identity is an assumption like any other that may be made, but need not be imposed.}
The subtle syntactic and semantic theory for \subrefllogic was developed as parallel to ordinary logic, showing that by dropping identity nothing more is lost.

From the proof-theoretic side, a general cut-elimination theorem for \subrefllogic was proved.
The investigation of the semantics explains exactly how the usual semantic characterizations of logical primitives subtly rely on the identity principle and shows how this assumption can be removed.

The denotational set-theoretic semantics reveal the close connection between \subrefllogic, partiality, three-valued logics and robust semantics, which interpret implication as robust (tolerant/strict) implication and have important applications \cite{DBLP:conf/ijcar/AbouElWafaP26}.
\Subrefllogic can reason about robust implication, and conversely \subrefllogic is the logic of robust implication.

The categorical analysis justifies the definition of \semiHeyting{} and \semiBooleanalgebra{}s.
From a categorical perspective, it shows that dropping identity makes the categorical story more subtle, mostly because of the failure of the Yoneda lemma.
However, by taking the right categorical perspective, the elegant theory can be recovered.

\textbf{Future Work.}
\Subrefllogic opens a very large area of research.
The extension to first-order logic is of foundational interest, and the systematic study of applications such as typing of partial functions, robust real-arithmetic and computable analysis through a Curry-Howard correspondence in \subrefllogic is promising.
Moreover, it would be of interest to extend the representation theorem to a geometric duality for \semiBooleanalgebra as an extension of the duality between distributive lattices and bitopological spaces \cite{DBLP:journals/mscs/BezhanishviliBGK10}.

\bibliography{subreflexive}

\appendix

\section{Rules for Classical \SubreflLogic}\label{sec:classicalrules}

As mentioned above, \classical \subrefllogic is the extension of intuitionistic \subrefllogic with sequents that admit multiple succedents that are read disjunctively together with the \LK versions of the \irref{orR} and \irref{implyL} rules.
Instead of making this formal, \classical \subrefllogic is conveniently defined as an extension of \subrefllogic by two rules with single succedents and this formulation will be shown to be equivalent to the multiple succedent version.
\begin{center}
    \begin{calculuscollection}{\renewcommand{\linferPremissSeparation}{\hspace{0.3cm}}%
            \hspace{-1.1em}
            \begin{calculus}
                \cinferenceRule[orRclassical|$\lor$R${}_c$]{classical or right proof rule}
                {
                    \linferenceRule[sequent]
                    {\lsequent{\fmlset,\lnot{\fml_{j}}} {\fml_i}}
                    {\lsequent{\fmlset} {\fml_1\lor\fml_2}}\quad
                }{\color{darkishgray}$i\neq j$}
            \end{calculus}}
        \qquad
        \begin{calculus}
            \cinferenceRule[implyLclassical|$\limply$L${}_c$]{classical imply left proof rule}
            {
                \linferenceRule[sequent]
                {\lsequent{\lnot\fmlc,\fml\limply\fmlb,\fmlset} {\fml}
                    &
                    \lsequent{\fmlb,\fmlset} {\fmlc}
                }
                {\lsequent{\fml\limply\fmlb,\fmlset} {\fmlc}}
            }{}
        \end{calculus}
    \end{calculuscollection}
\end{center}

\noindent
The classical intuitionistic subreflexive calculus is obtained as the intuitionistic subreflexive calculus with the addition of rules \irref{orRclassical} and \irref{implyLclassical}.

Note that by the two derivation
\begin{sequentdeduction}[array]
    \linfer[orRminus]
    {
        \linfer[orRclassical]
        {
            \linfer[implyL]
            {
                \lsequent{\fmlset,\fml}{\ffalse}
                !
                \linfer{\lclose}
                {\lsequent{\fmlset,\ffalse}{\fml}}
            }
            {\lsequent{\fmlset,\lnot\fml} {\fml}}
        }
        {\lsequent{\fmlset} {\fml\lor\fml}}
    }
    {\lsequent{\fmlset} {\fml}}
\end{sequentdeduction}
any sequent \(\lsequent{\fmlset}{\fml_1\ldots,\fml_n}\) can be viewed as a single succedent sequent \(\lsequent{\fmlset,\lnot\fml_1,\ldots,\lnot\fml_n}{\ffalse}\) equivalently.
Thus rules \irref{orRclassical} and \irref{implyLclassical} capture \classical \subrefllogic up to notation.
These rules make it easier to work with \classical \subrefllogic as a strict extension of intuitionistic \subrefllogic only by \emph{rules}.

\section{Derived and Inverse Rules}\label{sec:derivedinverse}

In this section, some useful derived and admissible rules are introduced, which are important for completeness and cut-elimination.

\subsection{Structural Rules}

The structural rules \irref{weak} and \irref{contraction} are admissible in intuitionistic and classical \subrefllogic.
\begin{center}
    \begin{calculuscollection}
        \begin{calculus}
            \cinferenceRule[weak|W]{left weakening proof rule}
            {
                \linferenceRule[sequent]
                {
                    \lsequent{\fmlset} {\fmlc}
                }
                {\lsequent{\fml,\fmlset} {\fmlc}}
            }{}
        \end{calculus}
        \qquad
        \begin{calculus}
            \cinferenceRule[contraction|C]{contraction proof rule}
            {
                \linferenceRule[sequent]
                {
                    \lsequent{\fml,\fml,\fmlset} {\fmlc}
                }
                {\lsequent{\fml,\fmlset} {\fmlc}}
            }{}
        \end{calculus}
    \end{calculuscollection}
\end{center}
\begin{proof}
    Admissibility of \irref{weak}, \irref{contraction} is shown by a standard induction over the length of the derivation.
    Note that for admissibility of \irref{weak} it is used that \irref{assumption}, \irref{assumptionL} and \irref{assumptionR} allows choosing subsets and for admissibility of \irref{contraction} it is used that these subsets are repetition free.
\end{proof}

\subsection{Rules for Negation}

It is interesting to consider the derived rules for negation, as they behave differently in \subrefllogic.
The following rules are easily \emph{derivable}:

\begin{calculuscollection}
    \begin{calculus}
        \cinferenceRule[notR|$\lnot$R]{negation introduction}
        {
            \linferenceRule[sequent]
            {
                \lsequent{\fmlset,\fml} {\ffalse}
            }
            {\lsequent{\fmlset} {\lnot\fml}}
        }{}
    \end{calculus}
    \qquad
    \begin{calculus}
        \cinferenceRule[exp|exp]{explosion}
        {
            \linferenceRule[sequent]
            {
                \lsequent{\fmlset,\fml} {\fml}
            }
            {\lsequent{\fmlset,\fml,\lnot\fml} {\fmlc}}
        }{}
    \end{calculus}
    \qquad
    \begin{calculus}
        \cinferenceRule[DNI|DNI]{double negation introduction}
        {
            \linferenceRule[sequent]
            {
                \lsequent{\fmlset,\fml} {\fml}
            }
            {\lsequent{\fmlset,\fml} {\lnot\lnot\fml}}
        }{}
    \end{calculus}
\end{calculuscollection}

\noindent
Note that the principle of explosion and double negation elimination are no longer derivable.
These rules require \(\fml\) to imply itself.
\begin{proof}
    Rule \irref{notR} is an instance of \irref{implyR}.
    The derivations

    \begin{minipage}[t]{\minpwidth}{\renewcommand{\linferPremissSeparation}{\hspace{0.5cm}}
            \begin{sequentdeduction}[array]
                \linfer[implyL]
                {
                    \lsequent{\fmlset,\fml}{\fml}
                    !
                    \linfer[]
                    {\lclose}
                    {\lsequent{\ffalse} {\fmlc}}
                }
                {\lsequent{\fmlset,\fml,\lnot\fml} {\fmlc}}
            \end{sequentdeduction}}
    \end{minipage}
    \hfill
    \begin{minipage}[t]{\minpwidth}
        \begin{sequentdeduction}[array]
            \linfer[implyR]
            {
                \linfer[exp]
                {\lsequent{\fmlset,\fml}{\fml}}
                {\lsequent{\fmlset,\fml,\lnot\fml} {\ffalse}}
            }
            {\lsequent{\fmlset,\fml} {\lnot\lnot\fml}}
        \end{sequentdeduction}
    \end{minipage}
    \qquad\qquad

    \noindent
    show derivability of \irref{exp} and \irref{DNI}.
\end{proof}

\subsection{Inverse Rules}

The \emph{inverse rules} of \subrefllogic for connectives are:

\begin{calculuscollection}
    \begin{calculus}
        \cinferenceRule[andLminus|$\land$L${}^-$]{and left minus proof rule}
        {
            \linferenceRule[sequent]
            {
                \lsequent{\fml_1\land\fml_2,\fmlset} {\fmlc}
            }
            {\lsequent{\fml_i,\fml_1\land\fml_2,\fmlset} {\fmlc}}
        }{}
        \cinferenceRule[andRminus|$\land$R${}^-$]{and right minus proof rule}
        {
            \linferenceRule[sequent]
            {
                \lsequent{\fmlset} {\fml_1\land\fml_2}
            }
            {\lsequent{\fmlset} {\fml_i}}
        }{}
    \end{calculus}
    \qquad
    \begin{calculus}
        \cinferenceRule[orLminus|$\lor$L${}^-$]{or left minus proof rule}
        {
            \linferenceRule[sequent]
            {
                \lsequent{\fml_1\lor\fml_2,\fmlset} {\fmlc}
            }
            {\lsequent{\fml_i,\fmlset} {\fmlc}}
        }{}
        \cinferenceRule[orRminus|$\lor$R${}^-$]{or right minus proof rule}
        {
            \linferenceRule[sequent]
            {
                \lsequent{\fmlset} {\fml\lor\fml}
            }
            {\lsequent{\fmlset} {\fml}}
        }{}
    \end{calculus}
    \qquad
    \begin{calculus}
        \cinferenceRule[implyRminus|$\limply$R${}^-$]{imply right minus a proof rule}
        {
            \linferenceRule[sequent]
            {
                \lsequent{\fmlset} {\fml\limply\fmlb}
            }
            {
                \lsequent{\fmlset,\fml} {\fmlb}
            }
        }{}
        \cinferenceRule[implyLminus|$\limply$L${}^-$]{imply left minus proof rule}
        {
            \linferenceRule[sequent]
            {
                \lsequent{\fml}{\fmld}
                &
                \lsequent{\fml\limply(\fmld\land\fmlb),\fmlset} {\fmlc}
            }
            {\lsequent{\fmlb,\fmlset} {\fmlc}}
        }{}
    \end{calculus}
\end{calculuscollection}

\begin{lemmaE}[][normal]\label{lem:inverseadm}
    The inverse rules \irref{andLminus}, \irref{andRminus}, \irref{orLminus}, \irref{orRminus}, \irref{implyRminus}, \irref{implyLminus} are admissible in proofs \(\subreflprov{\relfmlset}{\fml}\) if \(\relfmlset\) is \saturated.
\end{lemmaE}

\begin{proofE}
    The rule \irref{andLminus} is an instance of the admissible \irref{weak} rule.
    The other rules are proved by a straightforward induction on the length of the derivation and distinguishing based on the last applied rule.
    The fact that \(\relfmlset\) is saturated is used for \irref{assumption}, \irref{assumptionL} and \irref{assumptionR}.
\end{proofE}

Inverse rules for classical \subrefllogic:
\begin{center}
    \begin{calculuscollection}{\renewcommand{\linferPremissSeparation}{\hspace{0.3cm}}%
            \hspace{-1.1em}
            \begin{calculus}
                \cinferenceRule[orRclassicalminus|$\lor$R${}_c^-$]{classical or right inverse proof rule}
                {
                    \linferenceRule[sequent]
                    {\lsequent{\fmlset} {\fml_1\lor\fml_2}}
                    {\lsequent{\fmlset,\lnot{\fml_{j}}} {\fml_i}}
                    \quad
                }{\color{darkishgray}$i\neq j$}
            \end{calculus}}
        \qquad
        \begin{calculus}
            \cinferenceRule[implyLclassicalminusone|$\limply$L${}_c^{-1}$]{classical imply left inverse proof rule}
            {
            \linferenceRule[sequent]
            {\lsequent{\fml\limply\fmlb,\fmlset} {\fmlc}}
            {
                \lsequent{\fmlb,\fmlset} {\fmlc}
            }
            }{}
        \end{calculus}
        \qquad
        \begin{calculus}
            \cinferenceRule[implyLclassicalminustwo|$\limply$L${}_c^{-2}$]{classical imply left inverseproof rule}
            {
            \linferenceRule[sequent]
            {\lsequent{\fml\limply\fmlb,\fmlset} {\fmlc}}
            {\lsequent{\lnot\fmlc,\fml\limply\fmlb,\fmlset} {\fml}
            }
            }{}
        \end{calculus}
    \end{calculuscollection}
\end{center}

A set \(\relfmlset\) is \classicallysaturated if it is \saturated and closed under the classical inverse proof rules \irref{orRclassicalminus}, \irref{implyLclassicalminusone} and \irref{implyLclassicalminustwo}.

\begin{lemmaE}[][normal]\label{lem:classinverseadm}
    The inverse rules \irref{orRclassicalminus}, \irref{implyLclassicalminusone} and \irref{implyLclassicalminustwo} are admissible in proofs \(\classsubreflprov{\relfmlset}{\fml}\) if \(\relfmlset\) is \classicallysaturated.
\end{lemmaE}

\begin{lemmaE}[][normal]
    The contradiction proof rule \irref{raa} is derivable:
    \begin{calculuscollection}{
            \begin{calculus}
                \cinferenceRule[raa|RAA]{reductio ad absurdum rule}
                {
                    \linferenceRule[sequent]
                    {
                        \lsequent{\fmlset,\lnot\fml} {\ffalse}
                    }
                    {\lsequent{\fmlset} {\fml}}
                }{}
            \end{calculus}}
    \end{calculuscollection}
\end{lemmaE}
\begin{proofE}
    \begin{sequentdeduction}[array]
        \linfer[cut]
        {
            \linfer[implyR]{
                \linfer{\lclose}
                {\lsequent{\fmlset,\ffalse}{\ffalse}}
            }
            {\lsequent{\fmlset}{\lnot\ffalse}}
            !
            \linfer[orRminus]{
                \linfer[orRclassical]{
                    \lsequent{\fmlset,\lnot\fml}{\ffalse}
                }
                {\lsequent{\fmlset}{\fml\lor\ffalse}}
            }
            {\lsequent{\fmlset,\lnot\ffalse}{\fml}}
        }
        {\lsequent{\fmlset}{\fml}}
    \end{sequentdeduction}
\end{proofE}

\section{Proofs}\label{sec:proofs}

\printProofs

\clearpage

\end{document}